\documentclass[reprint,showpacs,twocolumn,superscriptaddress,aps]{revtex4-2}

\usepackage[utf8]{inputenc}
\usepackage[american,british]{babel}
\usepackage[T1]{fontenc}
\usepackage[pdftex]{graphicx}  
\usepackage{graphicx}
\usepackage{xcolor}
\usepackage{mathtools}
\usepackage{bbm}
\usepackage{subfigure}
\usepackage{dcolumn}
\usepackage{braket}
\usepackage{bm}
\usepackage{amsmath,amsthm,amssymb}
\usepackage{color}
\usepackage{verbatim}
\usepackage{comment}

\definecolor{Red}{HTML}{E53E30}  
\definecolor{Green}{HTML}{00AD69} 
\definecolor{Blue}{HTML}{2171b5}
\definecolor{Purple}{HTML}{652F6C}

\usepackage[colorlinks,citecolor=blue,linkcolor=blue,urlcolor=blue,hyperindex]{hyperref}

\usepackage[normalem]{ulem}

\newtheorem{proposition}{Proposition}

\usepackage[nameinlink,noabbrev]{cleveref}
\makeatletter
\def\@seccntformat#1{\@ifundefined{#1@cntformat}%
   {\csname the#1\endcsname\quad}
   {\csname #1@cntformat\endcsname}
}
\makeatother
\DeclareMathAlphabet\mathbfcal{OMS}{cmsy}{b}{n}

\usepackage{dsfont}
\usepackage{physics}

\begin{document}
\title{Entanglement barrier and its symmetry resolution: theory and experiment}

\author{Aniket Rath}
\thanks{These authors contributed equally.}
\affiliation{Univ.  Grenoble Alpes, CNRS, LPMMC, 38000 Grenoble, France}

\author{Vittorio Vitale}
\thanks{These authors contributed equally.}
\affiliation{SISSA, via Bonomea 265, 34136 Trieste, Italy}
\affiliation{International Centre for Theoretical Physics (ICTP), Strada Costiera 11, 34151 Trieste, Italy}

\author{Sara Murciano}
\thanks{These authors contributed equally.}
\affiliation{SISSA, via Bonomea 265, 34136 Trieste, Italy}
\affiliation{INFN, via Bonomea 265, 34136 Trieste, Italy}

\author{Matteo Votto}
\affiliation{Univ.  Grenoble Alpes, CNRS, LPMMC, 38000 Grenoble, France}

\author{J\'er\^ome Dubail}
\affiliation{Universit\'{e} de Lorraine, CNRS, LPCT, F-54000 Nancy, France}

\author{Richard Kueng}
\affiliation{Institute for Integrated Circuits, Johannes Kepler University Linz, Altenbergerstrasse 69, 4040 Linz, Austria}

\author{Cyril Branciard}
\affiliation{Universit\'e Grenoble Alpes, CNRS, Grenoble INP, Institut N\'eel, 38000 Grenoble, France}

\author{Pasquale Calabrese}
\affiliation{SISSA, via Bonomea 265, 34136 Trieste, Italy}
\affiliation{International Centre for Theoretical Physics (ICTP), Strada Costiera 11, 34151 Trieste, Italy}
\affiliation{INFN, via Bonomea 265, 34136 Trieste, Italy}

\author{Beno\^it Vermersch}
\affiliation{Univ.  Grenoble Alpes, CNRS, LPMMC, 38000 Grenoble, France}
\affiliation{Institute for Theoretical Physics, University of Innsbruck, Innsbruck, Austria}
\affiliation{Institute for Quantum Optics and Quantum Information of the Austrian Academy of Sciences,  Innsbruck A-6020, Austria}

\date{\today}

\begin{abstract}
    The operator entanglement (OE) is a key quantifier of the complexity of a reduced density matrix. In out-of-equilibrium situations, e.g. after a quantum quench of a product state, it is expected to exhibit an \emph{entanglement barrier}. The OE of a reduced density matrix initially grows linearly as entanglement builds up between the local degrees of freedom, it then reaches a maximum, and ultimately decays to a small finite value as the reduced density matrix converges to a simple stationary state through standard thermalization mechanisms. 
    Here, by performing a new data analysis of the published experimental results of [\href{http://doi.org/10.1126/science.aau4963}{Brydges {\it et al.}, Science {\bf 364}, 260 (2019)}], we obtain the first experimental measurement of the OE of a subsystem reduced density matrix in a quantum many-body system. We employ the randomized measurements toolbox and we introduce and develop a new efficient method to post-process experimental data in order to extract higher-order density matrix functionals and access the OE. The OE thus obtained displays the expected \emph{barrier} as long as the experimental system is large enough.
    For smaller systems, we observe a  barrier with a double-peak structure, whose origin can be interpreted in terms of  pairs of quasi-particles being reflected at the boundary of the qubit chain.
    As $U(1)$ symmetry plays a key role in our analysis, we introduce the notion of symmetry resolved operator entanglement (SROE), in addition to the total OE. To gain further insights into the SROE, we provide a thorough theoretical analysis of this new quantity in chains of non-interacting fermions, which, in spite of their simplicity, capture most of the main features of OE and SROE. In particular, we uncover three main physical effects: the presence of a barrier in any charge sector, a time delay for the onset of the growth of SROE, and an effective equipartition between charge sectors.
\end{abstract} 

\maketitle

\section{Introduction}

The investigation of the non-equilibrium dynamics of isolated many-body quantum systems is a major challenge of modern physics.
Owing to the highly tunable modern experimental settings for analog simulations \cite{blatt2012quantum,monroe2021programmable,browaeys2020many,Yoshi2022}, it has become possible to engineer Hamiltonian dynamics of isolated quantum systems, ranging from integrable to chaotic systems, and measure non-trivial physical properties, such as the entanglement growth following a quantum quench~\cite{Kaufman_2016,vovrosh2021confinement,Brydges2019probing,Elben2020b,satzinger_realizing_2021} and out-of-time ordered correlators~\cite{li2017measuring,garttner_measuring_2017,landsman_verified_2019,Yoshi2020,googleotoc}.

Unfortunately, the absence of numerical algorithms to effectively simulate these systems on a classical computer for large times is the main obstacle toward the complete understanding of quantum relaxation and thermalization. 
In this respect, the most effective and versatile  algorithms are surely those based on matrix product state (MPS) and tensor network methods \cite{o-19,v-04,vmc-08,s-11,pksmsh-19}.
However, the linear growth of the entanglement entropy~\cite{cc-05,albac-14} requires an exponential complexity (in bond dimension) of the MPS approximating the physical state which severely limits the largest simulable times~\cite{Schuch_2008}.

Typically, these systems relax to statistical ensembles with little or no entanglement. How is this compatible with the growth of complexity of the MPS approximation?
The solution of such a conundrum is simple: 
relaxation happens locally~\cite{bs-08,cde-08,cef-11}, hence it is enough to focus solely on the reduced density matrix $\rho$ of a subsystem $S$, rather than on the entire pure state containing physically irrelevant correlations.
Indeed, rather than an MPS, the subsystem density matrix is approximated by a matrix product operator (MPO) with small bond dimension $D$~\cite{Prosen2007,Pizorn2009,Zhou2017,dubail,noh2020efficient,rakovszky2022dissipation,Alba2022rise}. 
What is then the quantity that correctly assesses the validity of this approximation? It is the \textit{operator entanglement} (OE) of the reduced density matrix, which is the main subject of this manuscript. 

To introduce this quantity, it is useful to note that every bipartite density matrix $\rho_{AB}$ can be decomposed as follows:
\begin{equation}\label{eq:schmidtO0}
\frac{\rho_{AB}}{\sqrt{\mathrm{Tr}(\rho_{AB}^2)}} =\sum_i \lambda_i \, O_{A,i} \otimes O_{B,i}\,,
\end{equation}
where all expansion coefficients $\lambda_i$ are real and positive, and the associated operators $O_{A,i},O_{B,i}$ are orthonormal with respect to the Hilbert-Schmidt inner product (\mbox{$\mathrm{Tr}( O_{A,i}^\dagger O_{A,j}) = \mathrm{Tr} ( O_{B,i}^\dagger O_{B,j})= \delta_{i,j}$}, where $\delta_{i,j}$ is the Kronecker delta), see e.g.\ \cite{GUHNE20091}.
Such a decomposition is called an \emph{operator Schmidt decomposition} of $\rho_{AB}$, and has the property that the set of \emph{Schmidt coefficients} $\lambda_i$ is unique (although the whole decomposition is not). The number of such non-zero coefficients is called the \emph{operator Schmidt rank} of $\rho_{AB}$. Our choice of normalization on the lhs of Eq.~\eqref{eq:schmidtO0} ensures that the Schmidt coefficients thus introduced obey $\sum_i \lambda_i^2=1$, i.e.\ the set $\left\{ \lambda_i^2\right\}$ forms a probability distribution of (squared) Schmidt values.
In an MPO algorithm, the density matrix $\rho_{AB}$ is  `compressed' by truncating the full sum to only the $D$ largest contributions, for some reasonably low value of $D$. Typically, the approximation is accurate provided that the distribution $\left\{\lambda_i^2\right\}$ of squared Schmidt values in Eq.~\eqref{eq:schmidtO0} has small Shannon entropy, i.e.\ 
$S(\rho_{AB}) = S\left(\left\{\lambda_i^2\right\}\right)=- \sum_i \lambda_i^2 \log (\lambda_i^2)$ is small enough.
This quantity is called  the \emph{operator entanglement}~\cite{zanardi2000entangling,zanardi2001entanglement,Prosen2007,Pizorn2009} of the bipartite density matrix $\rho_{AB}$.
In this paper we will focus on the operator entanglement of a reduced density matrix. Its main physical feature is the presence of an \emph{entanglement barrier}~\cite{dubail,Wang2019barrier,Reid2021Entanglement}: after a quantum quench from a low-entangled state, the OE of a subsystem density matrix initially grows linearly and then decays at longer times, thus displaying a barrier-shaped curve. The initial linear growth is a consequence of the generic linear growth of the (state) entanglement entropy after a quench~\cite{cc-05,albac-14}, while the decay at later times reflects the convergence of the reduced density matrix towards a simple stationary state \cite{dubail}, through the  mechanism of thermalization~\cite{rigol2008thermalization,mori2018thermalization,mallayya2019prethermalization,eth1,eth2,dkpr-16,cde-08,bs-08} (or relaxation to a Generalized Gibbs ensemble \cite{Rigol2007,Cramer2010,vr-16,cem-16}).

The emergence of the entanglement barrier for the OE of a reduced density matrix in ergodic dynamics can be linked straightforwardly to the distribution of squared Schmidt values $\{ \lambda_i^2 \}$ from Eq.~\eqref{eq:schmidtO0}. At early times the evolution starts from a pure product state,  when only a single Schmidt value is different from zero. The building up of entanglement is reflected in the increasing number of non-zero
Schmidt coefficients $\lambda_i$.
For long times the system eventually locally approaches a Gibbs or Generalized Gibbs ensemble, which obeys the operator area law~\cite{dubail}, i.e. it is constant in the subsystem size, and again only few Schmidt values give a sizeable contribution to the OE. For example, in the infinite temperature limit, since the density matrix is proportional to the identity, $\rho\propto \mathbb{I}=(\mathbb{I}_A\otimes \mathbb{I}_B)$, only a single Schmidt value is different from zero and the OE vanishes.

\begin{figure*}
\begin{minipage}{\linewidth}
    \includegraphics[width=\linewidth]{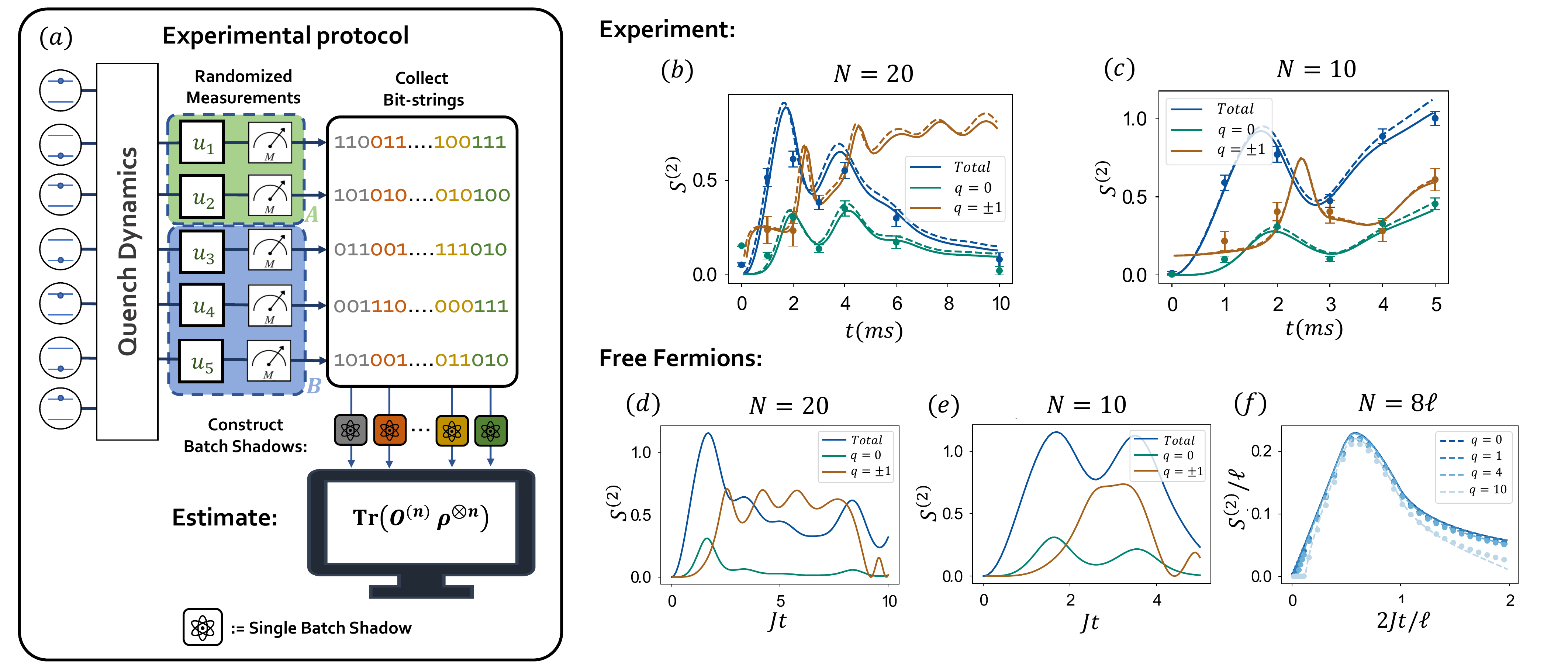}
\end{minipage}
    \caption{
    Overview of the results: 
a) Schematic of the method to post-process the experimental data. After the quench dynamics, randomized measurements are performed. The collected bit-strings are grouped to construct \emph{batch shadows} that provide estimates of OE and the SROE. b)-c) Experimental results for the R\'enyi 2-OE (Eq.~\eqref{eq:renyi-n-OE}) and its symmetry resolution (Eq.~\eqref{eq:renyi-n-OE-SR}) of a reduced density matrix formed from a partition of $4$ ions out of $20$ (panel b)) and $10$ ions (panel c)) from the data of~\cite{Brydges2019probing}  after the global quantum quench. The points correspond to the experimental data, the curves are numerical results obtained via tensor network algorithms with (solid) or without (dashed) dissipation considered. The entanglement barrier is visible for the total operator entanglement and the symmetry sector $q=0$ with $N=20$. 
d)-f) Symmetry resolution of the OE of the reduced density matrix after a global quantum quench in a free fermion chain under unitary evolution.
d) (and e)) Symmetry resolution of the OE of the reduced density matrix after a global quantum quench, for 4 sites out of a 20 (out of 10) sites chain. Comparing with the experimental results in b) and c) respectively, we can spot several qualitative features of OE even though the model is short ranged and there is no dissipation. f) Numerical data (circles) with subsystem length $\ell=256, \ell_A=\ell-\ell_B=120$ compared with quasiparticle prediction Eq.\eqref{eq:sqsp} (continuous lines). This plot shows the three main features of the SROE in the thermodynamic limit, i.e. the barrier in each sector $q$, the time delay and the equipartition for small $q$. In the free-fermion plots, $J$ is the hopping term and we consider, for the sake of simplicity, $\hbar=1$ and
lattice constant $a = 1$.}
    \label{fig:summary}
\end{figure*}

Inspired by the relevance of the \emph{entanglement barrier}, our goal is to observe it 
in an experimental quantum many-body system, using the randomized measurement data of the trapped ion experiment of Ref.~\cite{Brydges2019probing}.
The randomized measurement toolbox  \cite{RMtoolbox} has enabled measuring state-agnostically properties of the underlying quantum state, such as purity and R\'enyi entanglement entropies~\cite{Brydges2019probing,huang2020predicting,RMtoolbox,Rath2021,satzinger_realizing_2021}, negativities \cite{Zhou2020,neven2021symmetry,vitale2022symmetry}, state fidelities \cite{Elben2020a,huang2020predicting,zhu2021crossplatform} with a lower measurement cost compared to quantum state tomography~\cite{Gross2010,Kueng2017,Wright2016,Haah2017,Guta2020}.
One particular fruitful development is the formalism of classical shadows \cite{pk-19,huang2020predicting} that provides  estimations of additional non-linear functionals of the density matrix, such as the OE~\cite{liu2022detecting}.
However, measuring OE using the current randomized measurement toolbox requires a prohibitively expensive postprocessing method. To overcome this limitation and observe OE for a reasonable system size, we introduce in this work  the \emph{batch shadows estimator}. This new estimator, which should be of independent interest for the randomized measurement toolbox as a whole, provides a fast postprocessing technique for estimating multi-copy expectation values as functions of the density matrix. 
Importantly, this method offers, up to factors of order one in the experimentally relevant scenario, the same performance guarantees as classical shadows in terms of required number of measurements to overcome statistical errors. This enables us to experimentally access the OE and, in turn, witness the entanglement barrier. 

As a second important result, the experimental setup of Ref.~\cite{Brydges2019probing} provides us with an opportunity to study how the OE content is structured due to the presence of a symmetry which, here, is a $U(1)$ symmetry associated with the number of spin excitations.
In the case of pure state entanglement, the fruitful notion of symmetry-resolved entropies~\cite{laflorencie2014spin,moshe2018symmetry,xavier2018equipartition} has been introduced recently, computed theoretically~\cite{bpc-21,pbc-21,azses2021observation,turkeshi2020entanglement,murciano2020entanglement,capizzi2021symmetry,murciano2020symmetry_2D,murciano2020symmetry,murciano2021symmetry,shachar2021entanglement} and experimentally observed~\cite{lukin2019probing,vitale2022symmetry,neven2021symmetry}. Here, we generalize this to the case of OE.
Based on suitable supercharge operators (throughout this paper, we refer to the operators in the space of operators as superoperators and to the corresponding charges as supercharges), 
we introduce a notion of \emph{symmetry resolved operator entanglement (SROE)}, for which we also provide tractable estimation protocols. 
Using the SROE, we can theoretically study and experimentally observe a symmetry-resolved entanglement barrier. This is relevant for understanding thermalization in $U(1)$ symmetric non-equilibrium quantum systems, but also for numerical simulations, because symmetries can be incorporated in MPO algorithms.

    \label{fig:summary}
\section{Summary of results}\label{sec:summary}

Here we provide a bird's eye view of the results in this paper. The remaining manuscript is organised as follows.
In Sec.~\ref{sec:definition&results}, we introduce formally the operator entanglement (OE) and its symmetry resolution (SROE).
In Sec.~\ref{sec:experiment}, we demonstrate and describe the details of the randomized measurement protocol used to measure OE and SROE of the experiment performed in Ref.~\cite{Brydges2019probing}.
In Sec.~\ref{sec:reduced_quench}, we study analytically the SROE of the reduced density matrix of a subsystem of critical free fermionic chains after a quantum quench.
Finally we draw our conclusions in Sec.~\ref{sec:conclusion}. We include five
appendices with more details about the analytical and numerical computations. In App.~\ref{app:entanglementconditions}, we provide entanglement conditions for mixed states using quantities introduced for the operator entanglement, App.~\ref{app:proofs} follows up by detailing the proof of the symmetry-resolution of the operator Schmidt decomposition from Eq.~\eqref{eq:schmidtO0}. In App.~\ref{app:batchshadows}, we develop an analytical framework for the statistical error analysis of general batch shadows estimators, followed by describing the setup of the experiment that we have considered in this manuscript in App~\ref{app:detailsexperiment}. Further details on the treatment of experimental data using batch shadows is provided in App.~\ref{app:sroe_details}.\\

Let us thus start by providing a short summary of the results in this paper. The main points are also illustrated in Fig.~\ref{fig:summary}.

\begin{enumerate}
\item For the first time, we provide a general definition of the symmetry resolution of the OE in the presence of a global $U(1)$ symmetry. This is done formally in Sec.~\ref{sec:definition&results} and App.~\ref{app:proofs}.
\item We introduce a new analysis method, which allows us to measure the OE and SROE of a subsystem's density matrix in a many-body quantum system from the published experimental data of {\it Brydges et al.}~\cite{Brydges2019probing} presenting a $U(1)$ conserved charge.
Namely, we employ the \emph{randomized measurement toolbox}~\cite{RMtoolbox} and propose a new efficient method to post-process experimental data in order to extract arbitrary higher-order density matrix functionals of the form $\mathrm{Tr} \big( O^{(n)} \rho^{\otimes n} \big )$ expressed in terms of an $n$-copy operator $O^{(n)}$. 
A schematic of this procedure is shown in Fig.~\ref{fig:summary}a), and its details are elaborated in Sec.~\ref{sec:experiment}, as well as App.~\ref{app:batchshadows} -~\ref{app:sroe_details}. 
This tool is employed to extract the experimental results presented in  Fig.~\ref{fig:summary}b)-c).
Here we show the measured OE and SROE, as in Eq.~\eqref{eq:renyi-n-OE}, that are supported by tensor network simulations modelling the full experiment, i.e. the open dynamics of a long-range XY model starting from the Néel state, with conserved magnetization along the $z$-axis.
Our main observations of Sec.~\ref{sec:experiment} are summarized here:
\begin{enumerate}
    \item We witness experimentally the entanglement barrier of the OE and the SROE in the charge sector $q=0$ (Fig.~\ref{fig:summary}b), for a bipartite subsystem $A\cup B$ comprised of four out of \mbox{$N=20$} ions. These barriers present bump structures due to finite size effects. 
    For a smaller system $N=10$, we observe a second growth of OE after the first peak, as shown in Fig.~\ref{fig:summary}c). This can be interpreted as an effect of quasi-particles  reflected at the boundary of the chain, as described in Sec.\ref{sec:quasi-pic}.
    \item We observe a qualitative agreement of SROE with the numerical results for charge sectors $q = \pm 1$ at early times. The sizeable deviation between theory and experiment for $N=20$ (Fig.~\ref{fig:summary}b) is caused by the small populations in the corresponding charge sectors and by the low measurement statistics available from the experiment.
\end{enumerate} 

\item To gain insights into SROE  and its own entanglement barrier, we provide a thorough theory analysis in chains of non-interacting fermions, which despite their simplicity capture the main physical features of the OE and SROE.
This is already visible for small system sizes $N$, by comparing Fig. \ref{fig:summary}b) and \ref{fig:summary}c)
with \ref{fig:summary}d) and \ref{fig:summary}e), respectively. For these free models we moreover obtain the general formula in Eq.~\eqref{eq:sqsp}, which governs the evolution of the SROE. This formula allows us to uncover three main physical effects, which we expect to appear more generically in chains of qubits, beyond the simple non-interacting fermion ones.
These effects are:
\begin{enumerate}
    \item the appearance of a barrier for SROE in \emph{any} charge sector, which resembles the behavior of the total OE;
    \item a time delay for the onset of the SROE that grows linearly with the charge sector of the subsystem;
    \item the effective equipartition in the scaling limit of large time and subsystem size for small charges (see Eq.~\eqref{eq:equip}), where by equipartition we mean that the SROE is equally distributed among the different symmetry sectors. 
\end{enumerate}
These effects are visible in Fig.~\ref{fig:summary}f). There we plot OE and SROE of the reduced density matrix, for a bipartition $ A \cup B$, where the numerical results are obtained
for a quench in the tight-binding model from the N\'eel state while the solid lines correspond to Eq.~\eqref{eq:sqsp}. 

\end{enumerate}

\section{Operator entanglement and symmetry resolution}\label{sec:definition&results}

In this section, we formally revisit the OE definition, emphasize its connection to mixed state entanglement and also introduce symmetry resolved OE in the presence of an additive global conserved charge.

\subsection{Definition of Operator Entanglement}
Operator entanglement, or OE for short, can be defined for arbitrary operators acting on a bipartite quantum system $A \cup B$. 
But, for the sake of simplicity and clarity, we here present and discuss it solely for bipartite density operators $\rho_{AB}$. 

Recall from Eq.~\eqref{eq:schmidtO0} that every $\rho_{AB}$ admits an \emph{operator Schmidt decomposition}
\begin{equation}\label{eq:schmidtO}
\frac{\rho_{AB}}{\sqrt{\mathrm{Tr}(\rho_{AB}^2)}}=\sum_{i=1}^{R} \lambda_i \,O_{A,i} \otimes O_{B,i} \,,
\end{equation}
where $R=\mathrm{srank}(\rho_{AB})$ is the (operator) Schmidt rank, $\lambda_1,\ldots,\lambda_R >0$ are the Schmidt coefficients and 
$O_{A,i}$, as well as $O_{B,i}$ denote orthonormal operator families on subsystems $A$ and $B$, respectively. 
In general, for a Hermitian density operator $\rho_{AB}$, these operators $O_{A,j}$ and $O_{B,j}$ can be taken to be Hermitian themselves~\cite{GUHNE20091}, although this is not necessarily imposed.

In a similar way to the more widely known pure state case \cite{benenti2019principles}, the Schmidt values capture some form of entanglement that is present in the system. In fact, there is an intimate connection between the two. The Operator Schmidt decomposition in Eq.~\eqref{eq:schmidtO} arises from first vectorising the re-normalized operator \mbox{$\varrho_{AB} =\rho_{AB}/\sqrt{\Tr (\rho_{AB}^2)}$} (using the Choi-Jamio{\l}kowski isomorphism~\cite{jamiolkowski72,choi75})
\begin{equation}\label{eq:vectorisation}
    \varrho_{AB}=\sum_{ij}(\varrho_{AB})_{ij}\ket{i}\bra{j} \mapsto \ket{ \varrho_{AB}}=\sum_{ij}(\varrho_{AB})_{ij}\ket{i}\ket{j},
\end{equation}
applying the ordinary Schmidt decomposition to the pure state $|\varrho_{AB}\rangle$ and eventually reverting the vectorisation to get back to the space of operators. 
This connection justifies the quantification of OE in terms of the distribution of squared Schmidt values. 
More precisely, recalling that with our choice of normalisation in Eq.~\eqref{eq:schmidtO} the squared Schmidt coefficients $\{\lambda_i^2\}$ define a normalized probability distribution, we define the \emph{R\'enyi $\alpha$-OE}
\begin{equation}\label{eq:renyi-n-OE}
    S^{(\alpha)} (\rho_{AB})\,  := \,  \frac{1}{1-\alpha} \log \sum_{i=1}^R    \left(\lambda_i^2\right)^\alpha \quad \text{  for $\alpha \neq 1$},
\end{equation}
which in the limit $\alpha \to 1$ produces the (Shannon) OE
\begin{equation}
    \label{eq:OEdef}
    S (\rho_{AB})\,  := \,   \sum_{i=1}^R  - \lambda_i^2\log \lambda_i^2.
\end{equation}
It is clear that a state $\rho_{AB}$ has zero OE, $S(\rho_{AB})=0$, if and only if it can be written as a  product state of the form $\rho_{AB}=\rho_A\otimes \rho_B$. When this is not the case, we will call the state `operator-entangled'. It is worthwhile to emphasize that an operator-entangled (i.e., non-product) state may still be non-entangled according to the standard terminology for mixed state entanglement~\cite{GUHNE20091}.

\subsection{Operator Entanglement and entanglement criteria}

As already pointed out above, there is an intimate connection between operator entanglement and the more familiar concept of state entanglement where one is interested in showing that $\rho_{AB}$ cannot be written as a convex mixture of product states \mbox{$\rho_{AB}=\sum_k \alpha_k \rho_A^{(k)}\otimes \rho_B^{(k)}$}, \mbox{$\alpha_k\ge 0$} and $\rho_A^{(k)}, \rho_B^{(k)}$ are subsystem density matrices (i.e.\ positive semidefinite and unit trace)~\cite{GUHNE20091}. 
Let $\lambda_i$ be the coefficients of the operator Schmidt decomposition of $\rho_{AB}$, as introduced above. The  realignement/computable cross norm  criterion (CCNR) \cite{rudolph2000separability,ChenCCNR2003,rudolph2005further}
states that \emph{every} separable (i.e.\ nonentangled) state  produces operator Schmidt coefficients $\lambda_i$ that obey
\begin{equation}
\sum_i \lambda_i\le 1/\sqrt{\mathrm{Tr}\left[\rho_{AB}^{2}\right]}.
\label{eq:ccnr}
\end{equation}
Conversely, entanglement (accross the bipartation $A$ vs.\ $B$) must be present if this relation is violated.
The connection between the CCNR criterion and OE, which are the quantities that can be accessed experimentally, has been recently discussed in Ref.~\cite{liu2022detecting}.
Here we show a slightly weaker, but much more compact, entanglement condition: using the CCNR criterion, we can prove that separability implies that the R\'enyi 2-OE (Eq.~\eqref{eq:renyi-n-OE} for $\alpha=2$) and the R\'enyi 2-entropy $R^{(2)}:=-\log \Tr (\rho_{AB}^2)$ must obey
\begin{equation}
    S^{(2)}(\rho_{AB})\leq -\log(\mathrm{Tr}\left[\rho_{AB}^{2}\right])= R^{(2)}(\rho_{AB}).
\end{equation}
Conversely, if $S^{(2)}(\rho_{AB})> R^{(2)}(\rho_{AB})$, i.e. if $\rho_{AB}$ is more `operator mixed' than `state mixed' with respect to R\'enyi 2-entropies, $\rho_{AB}$ is necessarily entangled. We refer to App.~\ref{app:entanglementconditions} for details. There we also compare the detection power of this method with other entanglement conditions and also present experimental results.

\subsection{Symmetry Resolved Operator Entanglement} \label{subsec:SROE}
In the presence of a global symmetry, the OE of the operator $\rho_{AB}$ can be split into different charge sectors, similarly to that of the state entanglement \cite{xavier2018equipartition,moshe2018symmetry}.

This happens in particular for a global $U(1)$ symmetry, where the $U(1)$ charge operator acting on $A \cup B$ is a sum of the two charge operators acting on subsystems $A$ and $B$, i.e.\ $Q_{AB}=Q_A+Q_B $. From now on, by $Q_{A}+Q_B$ we mean $Q_{A}\otimes \mathbb{I}_B+\mathbb{I}_A\otimes Q_B$. If the density matrix $\rho_{AB}$ commutes with $Q_{AB}$, that is
\begin{equation} \label{eq:commut_QAB_rho}
    [ Q_A + Q_B, \rho_{AB} ] =0 ,
\end{equation}
then it becomes possible to reorganize the terms in the Schmidt decomposition \eqref{eq:schmidtO} according to their `charge' $q$:
\begin{equation}
    \label{eq:rhoABgeneral}
    \frac{\rho_{AB}}{\sqrt{ {\rm Tr} [\rho_{AB}^2]}} = \sum_q \sum_j \lambda^{(q)}_j O_{A,j}^{(q)} \otimes O_{B,j}^{(-q)},
\end{equation}
where 
\begin{equation}\label{eq:qdefinition}
    \Big[Q_A , O_{A,j}^{(q)}\Big] = q \, O_{A,j}^{(q)}, \quad  \Big[Q_B , O_{B,j}^{(-q)}\Big] = -q \, O_{B,j}^{(-q)}
\end{equation}
such that \mbox{$ [Q_{A} + Q_B, O_{A,j}^{(q)} \otimes O_{B,j}^{(-q)} ] =0$}.
Equations \eqref{eq:rhoABgeneral} and
\eqref{eq:qdefinition} are proven in  App.~\ref{app:proofs}.
In particular, we show that the `charge' $q$ that appears in these equations can be introduced, through the vectorisation technique introduced in Eq.~\eqref{eq:vectorisation}, based on the notion of a charge `superoperator'
\begin{equation} \label{eq:charge_superop}
    \mathcal{Q}_{AB}=Q_{AB}\otimes \mathbbm{1}-\mathbbm{1}\otimes Q_{AB}^T.
\end{equation}
Namely, the values $q$ are the eigenvalues of \mbox{$\mathcal{Q}_{A}=Q_{A}\otimes \mathbbm{1}-\mathbbm{1}\otimes Q_{A}^T$}, i.e. the restriction of $\mathcal{Q}_{AB}$ to the subsystem $A$ \cite{moshe2018symmetry} (or equivalently, the eigenvalues of the commutator $[Q_A,\cdot]$, as in Eq.~\eqref{eq:qdefinition}). 
As we prove in App.~\ref{app:proof4}  using a language analogous to the symmetry resolution for a state~\cite{moshe2018symmetry},  the super reduced-density matrix $\mathrm{Tr}_{B \otimes B}(\ket{\rho_{AB}}\bra{\rho_{AB}})$ admits a block decomposition in the eigenspaces corresponding to these charges $q$, which leads then to Eqs.~\eqref{eq:rhoABgeneral} and \eqref{eq:qdefinition}. 
We also provide an illustrative example of Eq.~\eqref{eq:rhoABgeneral} in App.~\ref{sec:example}, starting from a 3-qubit system.

Similarly to the non-symmetry-resolved case of Eq.~\eqref{eq:schmidtO}, the newly constructed operator families $O_{A,j}^{(q)}$ and $O_{B,j}^{(q)}$ in Eq.~\eqref{eq:rhoABgeneral} are orthonormal with respect to the Hilbert-Schmidt inner product, i.e.\ ${\rm Tr} [ (O_{A,j_1}^{(q_1)})^{\dagger} O_{A,j_2}^{(q_2)} ]  \, = \, {\rm Tr} [ (O_{B,j_1}^{(q_1)})^{\dagger} O_{B,j_2}^{(q_2)} ]   \, =\,  \delta_{q_1,q_2} \delta_{j_1,j_2}$.
 In contrast to Eq.~\eqref{eq:schmidtO} however, in the symmetry-resolved Schmidt decomposition~\eqref{eq:rhoABgeneral} these operators can not always be taken to be Hermitian.

 By uniqueness of the Schmidt coefficients, the set of all (non-zero) values $\{\lambda_j^{(q)}\}$ altogether must be the same as the set of values $\{\lambda_i\}$ from Eq.~\eqref{eq:schmidtO}.
We can now define the total weight of the terms at fixed $q$ to be
\begin{equation}
    p(q) := \sum_{j} (\lambda_j^{(q)})^2.
\end{equation}
These weights satisfy $\sum_{q} p(q) = 1$ and give a probability distribution over the different charge sectors. In terms of that probability distribution, the (Shannon) OE from Eq.~\eqref{eq:OEdef} becomes
\begin{equation}\label{eq:sum_rule}
    S(\rho_{AB}) \, = \, \sum_q  p(q) S_q(\rho_{AB})   \, + \, \sum_q - p(q) \log p(q) ,
\end{equation}
where the \emph{symmetry-resolved operator entanglement} (SROE) of $\rho_{AB}$ in the charge sector $q$ is
\begin{equation}\label{eq:renyi-n-OE-SR}
    S_q (\rho_{AB}) \, : =- \, \sum_j   \left(\frac{(\lambda_j^{(q)})^2}{p(q)}\right) \log \left( \frac{ (\lambda_j^{(q)})^2}{p(q)}\right).
\end{equation}

Similarly, for $\alpha \neq 1$, we define the \emph{R\'enyi-$\alpha$ SROE}  to be
\begin{equation}\label{eq:renyi-alpha-OE-SR}
    S_q^{(\alpha)} (\rho_{AB}) \, : = \, \frac{1}{1-\alpha}  \log \left( \sum_j \left(\frac{ (\lambda_j^{(q)})^2}{p(q)} \right)^{\!\!\alpha\,} \right).
\end{equation}
Note however that a formula analogous to Eq.~\eqref{eq:sum_rule} for a R\'enyi index $\alpha \neq 1$, in terms of $p(q)$, cannot be written.

Importantly, in this paper we focus on a density matrix on a bipartite subsystem $A \cup B$ that results from tracing out an additional system $C$.
Let us observe that $\rho_{AB}$ commutes with $Q_{AB}$ as soon as the full system $A \cup B \cup C$ is in a pure state, which is also an eigenstate of the total $U(1)$ charge operator $Q_A+Q_B + Q_C$.  Then tracing out the degrees of freedom in $C$ automatically yields a reduced density matrix $\rho_{AB}$ which is block diagonal in $Q_{AB}$. This, in turn, ensures that the SROE is well defined. This reasoning also extends to mixed states that are block diagonal with respect to the charge operator: if the density matrix of the full system, $\rho_{ABC}$, commutes with $Q_{A}+Q_B+Q_C$, then $\rho_{AB} = {\rm Tr}_C(\rho_{ABC})$ commutes with $Q_{AB}$, and the discussion above also applies. This is because $[Q_{AB},\rho_{AB}] = {\rm Tr}_C ( [ Q_{AB},  \rho_{ABC} ] ) = {\rm Tr}_C \left( [ Q_A+Q_B+Q_C, \rho_{ABC}] \right) - {\rm Tr}_C \left( [ Q_C, \rho_{ABC}] \right) =  - {\rm Tr}_C( [ Q_C, \rho_{ABC}] )$, and it vanishes due to the cyclicity of the partial trace over $C$. In this paper we always deal with full system density matrices $\rho_{ABC}$ that commute with $Q_A + Q_B +Q_C$.

\section{Operator entanglement in the quench dynamics of trapped ions}\label{sec:experiment}
Let us now come to one of the main results of the paper: the development of tractable methods to extract R\'enyi $\alpha-$OE in an experiment, and the corresponding experimental observations of the entanglement barriers with R\'enyi $2-$OE and its symmetry resolution.

In Sec.~\ref{sec:classicalshadows} we detail the experimental protocol of classical shadows and, in Sec.~\ref{sec:batchshadows}, the associated efficient method for the post-processing of the measurement data, dubbed the \emph{batch shadows estimator}. 
In Sec.~\ref{sec:experimentalresults} and Sec.~\ref{sec:quasi-pic}, we discuss the experimental results.

\subsection{R\'enyi OE from randomized measurements}\label{sec:classicalshadows}
In the previous sections, we have expressed OE as a function of the Schmidt spectrum $\{\lambda_i\}$. In order to express estimators of these quantities based on experimental data, one needs to rewrite them into a functional of the density matrix $\rho_{AB}$. In particular, the R\'enyi $2-$OE is a fourth order function of $\rho_{AB}$ that explicitly writes as~\cite{liu2022detecting}:
\begin{equation}
    S^{(2)} = - \log \frac{ \mathrm{Tr} \Big( \mathcal{S} \, \rho_{AB}^{\otimes 4} \Big )}{\mathrm{Tr}(\rho_{AB}^2)^2} = \tilde S^{(2)}(\rho_{AB}) - 2R^{(2)}(\rho_{AB}) \label{eq:renyi2OE},
\end{equation}
where $\mathcal{S} =  \mathbb{S}^{(A)}_{1,4} \otimes  \mathbb{S}^{(A)}_{2,3} \otimes \mathbb{S}^{(B)}_{1,2} \otimes \mathbb{S}^{(B)}_{3,4}$ is defined in terms of the swap operators 
$\mathbb{S}^{(X)}_{k,l}$ that swap the $k^\text{th}$ and $l^\text{th}$ copies of system $X$ (see App.~\ref{app:entanglementconditions}).
We also have defined the unnormalized R\'enyi 2-OE \mbox{$ \tilde S^{(2)}(\rho_{AB}) = -\log \Big( \mathrm{Tr} ( \mathcal{S} \rho_{AB}^{\otimes 4}) \Big )$}, and we note that $R^{(2)}(\rho_{AB})$ can also be written in a similar form as $R^{(2)}(\rho_{AB}) = -\log \Big( \mathrm{Tr} ( \mathbb{S}^{(AB)}_{1,2} \rho_{AB}^{\otimes 2}) \Big )$. We present similar expressions for the SROE in App.~\ref{app:sroe_details}. 

Such functionals on $N$-qubit density matrices can be accessed in qubit experiments via randomized measurements ~\cite{liu2022detecting,huang2020predicting,RMtoolbox}, as shown in Fig.~\ref{fig:summary}a).
We start with the preparation of our $N$-qubit state in the experiment.
We apply local random unitaries $u_i$ ($i=1,\dots,N$), sampled from the circular unitary ensemble (CUE) or a unitary $2-$design to each qubit separately and subsequently measure them in the $z$-basis. The measurement outcomes are recorded as a bit string $s=s_1,\dots ,s_N$.
We repeat this procedure for a set of $N_u$ distinct unitaries $u^{(r)}$ (of the form $u_1\otimes\cdots\otimes u_N$) and collect, for each thus applied unitary, $N_M$ bit-strings $s^{(r,m)} = s_1^{(r,m)}, \ldots, s_N^{(r,m)}$ with $r = 1,\ldots, N_u$ and $m = 1, \ldots, N_M$ .
This recorded data can then be used to construct operators 
\begin{equation}
    \hat{\rho}^{(r,m)}=\bigotimes_{i=1}^N \left[3 (u_i^{(r)})^{\dagger}\ket{s_i^{(r,m)}}\bra{s_i^{(r,m)}}(u_i^{(r)}) - \mathbb{I}_2 \right]. \label{eq:def_shadow_rm}
\end{equation}
These operators are called a \emph{classical shadows}~\cite{huang2020predicting} and constitute independent, unbiased estimators of the underlying quantum state, in the sense that $\mathbb{E}[\hat{\rho}^{(r,m)}]=\rho$, where the expectation value is taken over the applied unitaries and measurement outcomes (see also App.~\ref{app:batchshadows}). 
One can also perform appropriate robust estimations in the presence of an unknown noise channel by constructing robust versions of these classical shadows~\cite{Chen2021robust,Koh2022classicalshadows,vandenbergmitigation}.

In order to measure functions \mbox{$X_n = \mathrm{Tr}(O^{(n)} \rho^{\otimes n})$} that are expectation values of a $n$-copy observable $O^{(n)}$ (here, in particular, we are interested in $O^{(4)}=\mathcal{S}$ on four copies, see Eq.~\eqref{eq:renyi2OE}), one can define an U-statistics estimator $\hat{X}_n$ given by
\begin{equation}
     \hat{X}_n =\frac{1}{n!} \binom{N_u}{n}^{-1} \sum_{r_1 \ne \dots \ne r_n} \mathrm{Tr} \Big[ O^{(n)} \bigotimes_{i = 1}^n \hat{\rho}^{(r_i)}\Big], \label{eq:u_stat}
\end{equation}
where we have introduced the classical shadow \mbox{$\hat{\rho}^{(r)} = \mathbb{E}_{N_M}[\hat{\rho}^{(r,m)}]$} constructed by averaging over all measured bit-strings for an applied unitary $u^{(r)}$. 
The estimator $\hat{X}_n$ is unbiased, i.e $\mathbb{E}[\hat{X}_n] = X_n$~\cite{huang2020predicting}.

This estimator has been used to access experimentally properties involving observables on up to $n=3$ copies~\cite{Elben2020b,neven2021symmetry,vitale2022symmetry}.
However, the underlying procedure quickly becomes computationally unfeasible and impractical as it requires summing over all possible combinations of $n$ distinct shadows $\hat{\rho}^{(r_1)}, \ldots, \hat{\rho}^{(r_n)}$ for $r_i \in [1,\ldots, N_u]$. Furthermore, its runtime scales with the number of terms involved in the above sum: $\mathcal{O}(N_u^n)$, a number that grows exponentially with the polynomial degree $n$. This scaling prevents us in practice from extracting the R\'enyi $2-$OE from experimental data of~\cite{Brydges2019probing} (as $n > 3$). Thus we are in dire need of an alternate method with a substantially reduced runtime. 

\subsection{Fast estimation of high order functionals using randomized measurements data via batch shadows}\label{sec:batchshadows}
In order to improve the post-processing run time of classical shadows, we propose to form \mbox{$b=1,\dots,n'\ge n$} `batch shadows', each of which is an average of $N_u/n'$ shadows (assuming, for simplicity, that $N_u/n'$ is an integer): \mbox{$\tilde \rho^{(b)}=(n'/N_u)\sum_{r=(b-1)N_u/n'+1}^{bN_u/n'}\tilde \rho^{(r)}$}.
This allows us to define an alternate unbiased estimator
\begin{equation}
        \tilde{X}_n^{(n')} = \frac{1}{n!} \binom{n'}{n}^{-1} \sum_{b_1\neq \dots \neq b_n}  \mathrm{Tr} \Big [ O^{(n)} \bigotimes_{i = 1}^n \tilde \rho^{(b_i)}    \Big] , \label{data-split-general}
\end{equation}
which is different from Eq.~\eqref{eq:u_stat} and easier to compute.
The first step involves the construction of the $n'$ batch shadows $\tilde{\rho}^{(b)}$, which obey $\mathbb{E}[\tilde \rho^{(b)}] = \rho$ for all batches $b=1,\ldots,n'$. This is achieved 
by summing up all classical shadows that belong to a respective batch -- a subroutine that requires $O(N_u)$ arithmetic operations (provided that the sample complexity $N_u$ exceeds the total number of degrees of freedom in the reduced density matrix).
These individual summation steps can be obviously paralellized on $n'$ cores.
Note also that, in contrast to the bare classical shadows $\hat{\rho}^{(r,m)}$, the batch shadows $\tilde \rho^{(b)}$ are stored in memory as dense $2^{N} \times 2^{N}$ matrices. 
For typical memory available on current hardware, this limits our fast estimation methods to systems sizes of up to $N \approx 15$ qubits.

The second step requires the evaluation of $\tilde {X}_n^{(n')}$ from the constructed batch shadows, which scales as $\mathcal{O}(n'^n)$. Thus by choosing $n' = n$ and assuming that $N_u \gg n'^n$, we obtain the fastest estimator with an evaluation time $\mathcal{O}(N_u)$. This is a drastic runtime improvement compared to the original U-statistics estimator in Eq.~\eqref{eq:u_stat}: $\mathcal{O}(N_u)$ steps (new) vs.\ $\mathcal{O}(N_u^n)$ steps (old). As we increase $n'$, one starts to incorporate more terms with distinct combinations of $n$ different shadows that were not previously considered. This progression terminates in an eventual convergence to the original U-statistics estimator, i.e $\tilde{X}^{(N_u)}_n = \hat{X}_n$, as well as in an increasing of the post-processing run-time. In order to gauge the performance of the estimator $ \hat{X}^{(n')}_n$, we study its statistical error behavior.

{\it Statistical errors ---} The statistical errors in randomized measurements arise due to applying a finite number of random unitaries $N_u$ and performing a finite number of readout measurements $N_M$.
The statistical errors of any estimator $\hat{X}$ is governed by its variance $\mathrm{Var}[\hat{X}]$. One can provide rigorous performance guarantees to estimate $X_n$  with an accuracy $\epsilon$ from our protocol by bounding this variance and subsequently applying Chebyshev's inequality: $\mathrm{Pr}[|\hat{X}^{(n')}_n - X_n| \geq \epsilon] \leq \mathrm{Var}[\hat{X}^{(n')}_n]/\epsilon^2$.
In App.~\ref{app:batchshadows}, we provide a general framework that can be applied to calculate variance bounds on the batch shadow estimator for arbitrary multi-copy operators.
We can provide then rigorous performance guarantees for our estimation formulas, which we can also compare with the results for classical shadows presented in Ref.~\cite{RathFisher2021}.

From this study, in the limit of $N_M = 1$, as elaborated in App.~\ref{App:Generaltreatement}, we notice that $\mathrm{Var}[\hat{X}^{(n')}_n]$ and $\mathrm{Var}[\hat{X}_n]$ have the same scaling behavior in the high accuracy regime of $\epsilon \to 0$: that is, in first order in $1/N_u$, they both scale \mbox{$\propto n^2/N_u$} with the same proportionality constant. 
Moreover, for $n' = n$, at second order in $1/N_u$, $\mathrm{Var}[\tilde{X}^{(n)}_n]$ exceeds $\mathrm{Var}[\hat{X}_n]$ by only a small factor of $n/(n-1)$. This shows that the required number of measurements to achieve a given accuracy  $\epsilon$ is essentially the same for the fast batch shadow estimator (Eq.~\eqref{data-split-general}) and the standard shadow estimator (Eq.~\eqref{eq:u_stat}).

Of course, we can apply our general variance bound formalism to the quantities of interest for this work:  $O^{(2)}=\mathbb{S}_{1,2}^{(AB)}$
and $O^{(4)}=\mathcal{S}$ that give access to $R^{(2)}(\rho_{AB})$ and $\tilde S^{(2)}(\rho_{AB})$, respectively.
In the case of Clifford shadows (i.e.\ each random unitary is chosen uniformly from the single-qubit Clifford group) and $n'=n$, we find that in order to estimate them with a confidence interval of $\delta$, i.e., to make sure that $\mathrm{Pr}[|\tilde{X}^{(n)}_n - X_n| \geq \epsilon] \leq \delta$, we require a number of measurements that scales as $N_u \propto 3^N/\epsilon^2$ with $N$. We refer to  App.~\ref{app:puritybounds} and App.~\ref{app:boundX4} for further details. 
Hence, in the worst case scenario, our measurement bound of the batch-shadow estimator of $\tilde{X}^{(n)}_n$ scales as $3^N$ irrespective of the order $n = 2,\,4$.
For evaluating $\tilde S^{(2)}(\rho_{AB})$, in particular, this measurement bound is a polynomial improvement over the best previously obtained bounds which only achieve $4^N$~\cite{liu2022detecting}. We conjecture that this desirable scaling persists when we increase  $\alpha$ to evaluate higher-order R\'enyi $\alpha-$OE.
We also complement these rigorous bounds with small-scale numerical simulations in App.~\ref{app:numericalsims}.
\subsection{Experimental results using batch shadows}\label{sec:experimentalresults}
The batch shadow formalism allows us to extract experimentally the R\'enyi $2-$OE along with its symmetry resolution. 
We perform our set of observations by reprocessing batch shadows from the randomized measurement data of two sets of experiments, where a global quench with a long-range XY model was realized on a string of $10$ and $20$ qubits (ions), respectively~\cite{Brydges2019probing}.
The initial state was a N\'eel state, $\ket{\psi}=\ket{01}^{\otimes N/2}$, with vanishing operator (and state)  entanglement entropy.
The global quench was followed by the implementation of randomized measurement protocol involving a total of $N_u  = 500$ Haar random unitaries.
For each of the applied unitaries $N_M = 150$ bit-string measurements were made. 
Details on the modelling of quench dynamics with tensor network algorithms and the protocol are discussed in App.~\ref{app:detailsexperiment}.
 
We consider two bipartite reduced density matrices $\rho_{AB}$ defined on the subsystems $A = [2, \, 3]$ and $B = [4,\, 5]$ and $A = [8, \, 9]$ and $B = [10,\, 11]$ for a total chain of 10 ions and 20 ions, respectively, where we have labelled the ions along the chain from $1$ to $N$. Our observations remain unchanged for other partitions.
Fig.~\ref{fig:summary}(b-c) and Fig.~\ref{fig:fig2mt} show the experimental results with corresponding numerical simulations both with and without decoherence of the experiment. 
Panels a) and b) in Fig.~\ref{fig:fig2mt}, highlight the extracted R\'enyi 2-OE with the simplest batch shadow estimator ($n' = 4$). 

We first observe the entanglement barrier for the considered partition of the $20$ ion system in Fig.~\ref{fig:summary}b) and Fig.~\ref{fig:fig2mt}a). We observe a barrier composed of a growth phase from $t=0$ to $t\approx 3 $ ms, and a decay phase from $t\approx 3 $ ms to the last data point at $t = 10$ ms. The peak at $t\approx 3 $ ms actually looks more like a double-peak with maxima at $t\approx 1.8 $ ms and $t\approx 3.8 $ ms. We interpret this as oscillations on top of the main barrier caused by the small size of subsystems $A$ and $B$. This interpretation is supported by the fact that similar finite-size effects are found in our free fermion model, as shown in Fig.~\ref{fig:fig2mt}d) (see also Sec.~\ref{sec:reduced_quench}). The growth phase at early times signals the creation of correlations between the two subsystems $A$ and $B$, while the decay phase
 reflects the fact that $\rho_{AB}$ goes towards a thermal-like density matrix with small OE. Since the system is finite, we also expect revivals of the OE at longer times, however such revivals are not yet visible in the available time window.  
The barrier can also be understood as a competition between the terms $\tilde S^{(2)}(\rho_{AB})$ and $R^{(2)}(\rho_{AB})$ in the respective regimes as shown in Fig.~\ref{fig:fig2mt}c)~\cite{Wang2019barrier}.
In the growth phase,
the unnormalised R\'enyi 2-OE $\tilde S^{(2)}(\rho_{AB})$ grows at a faster rate compared to the state entropy $2R^{(2)}(\rho_{AB})$. 
In the decay phase, this behavior is inverted. These general features are consistent with the theoretical predictions of different models shown in Refs.~\cite{dubail,Wang2019barrier,Bertini2020}.

Comparing Fig.~\ref{fig:fig2mt}a) and b), we see, however, that in the smaller system of 10 ions no similar barrier is found. In particular, we do not observe the decay phase. We discuss this case in more detail in Sec.~\ref{sec:quasi-pic} below.

\begin{figure}
\begin{minipage}[b]{0.48\linewidth}
\centering
\includegraphics[width=\textwidth]{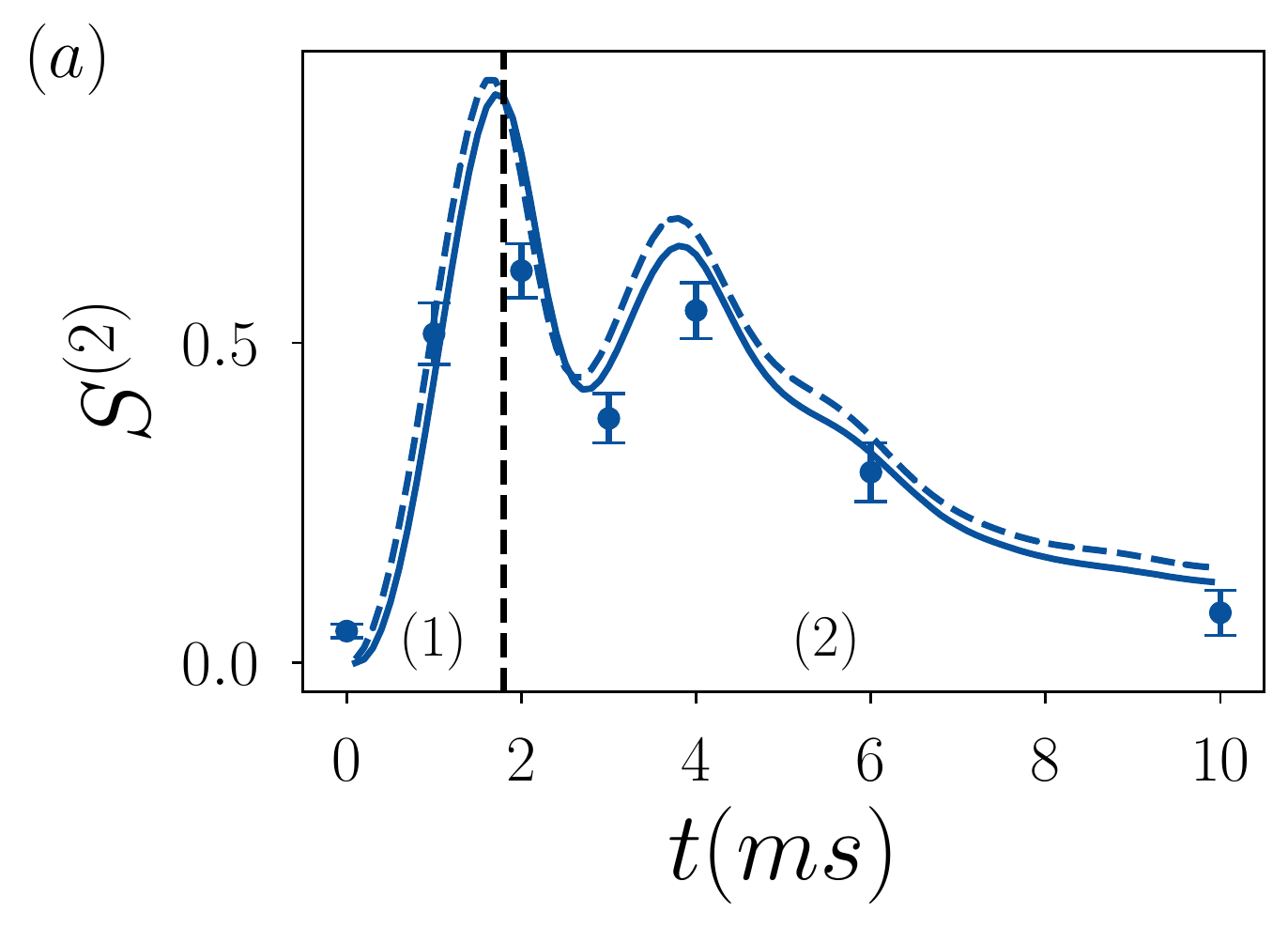}
\end{minipage}
\hskip -0.75ex
\begin{minipage}[b]{0.48\linewidth}
\centering
\includegraphics[width=\textwidth]{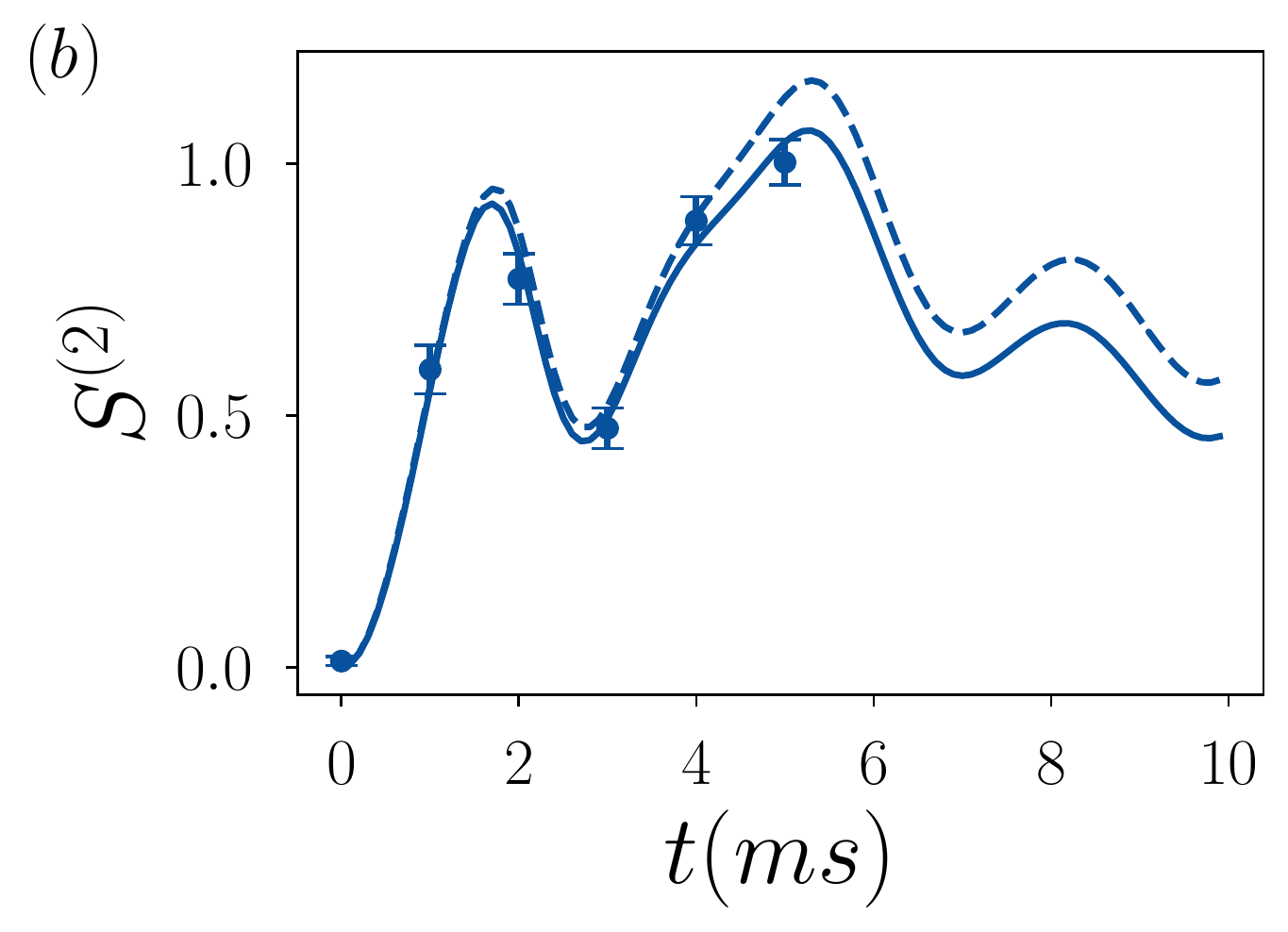}
\end{minipage}
\begin{minipage}[b]{0.48\linewidth}
\centering
\includegraphics[width=\textwidth]{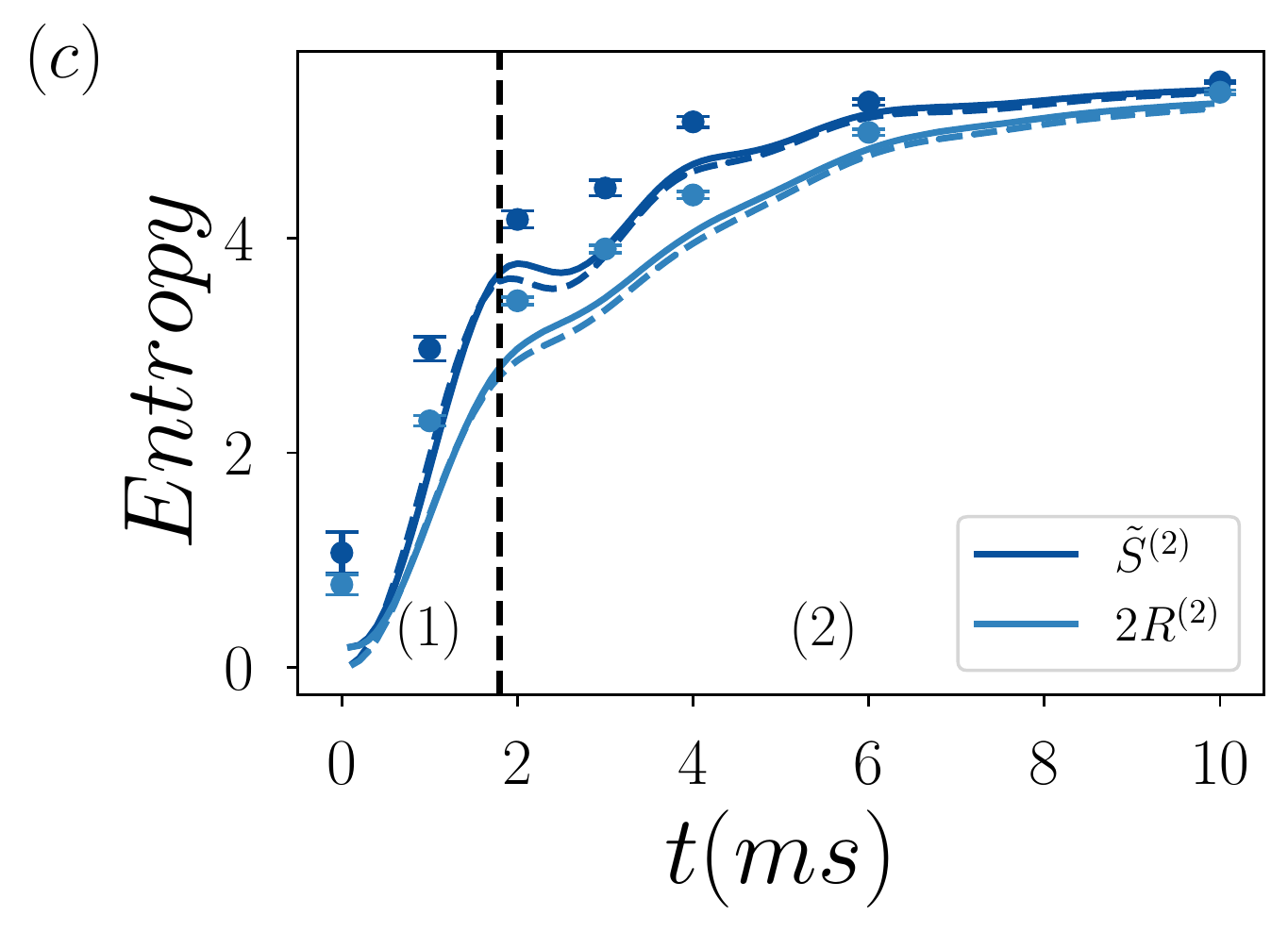}
\end{minipage}
\hskip -0.75ex
\begin{minipage}[b]{0.48\linewidth}
\centering
\includegraphics[width=\textwidth]{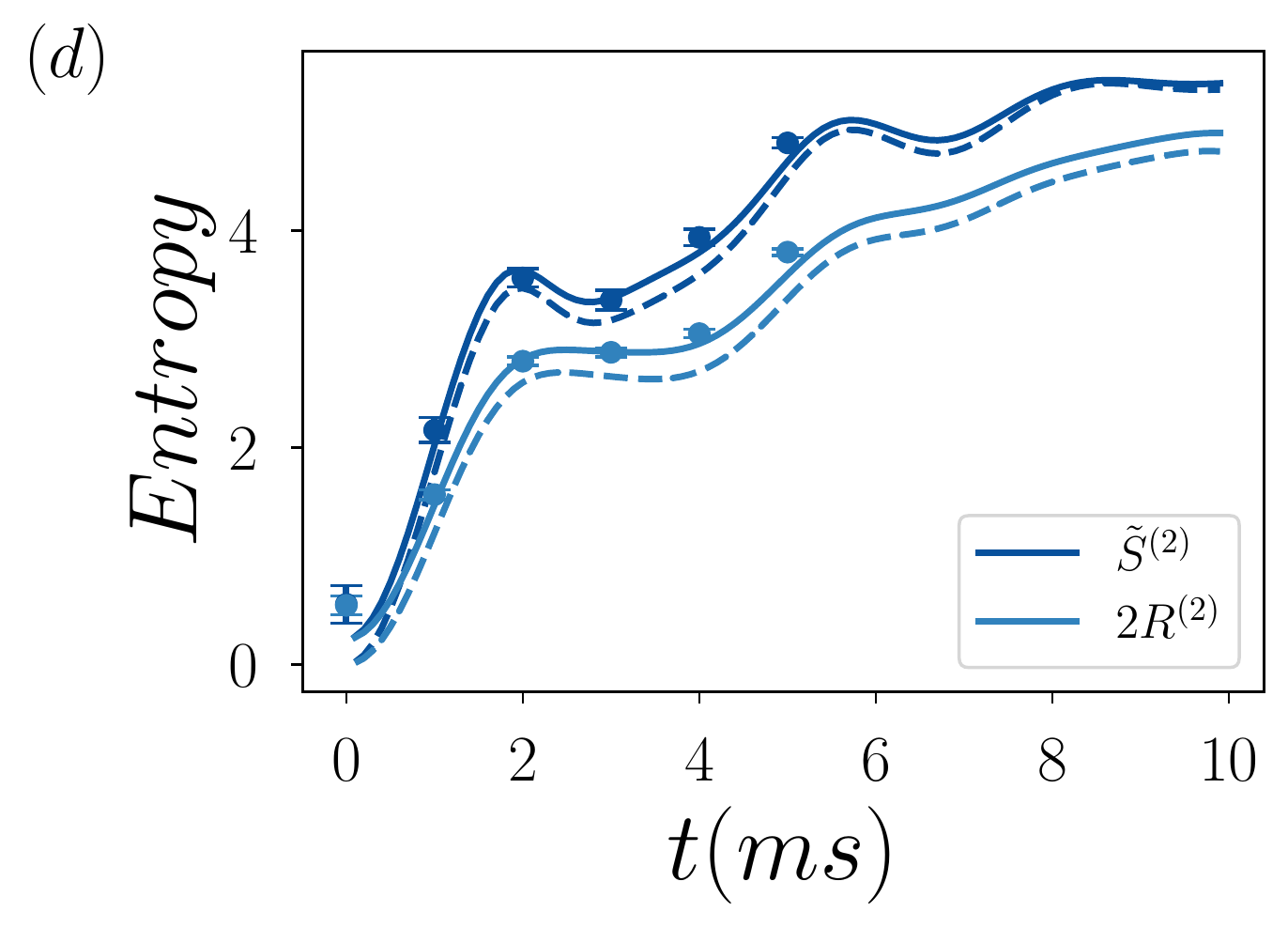}
\end{minipage}
\begin{minipage}[b]{0.48\linewidth}
\centering
\includegraphics[width=\textwidth]{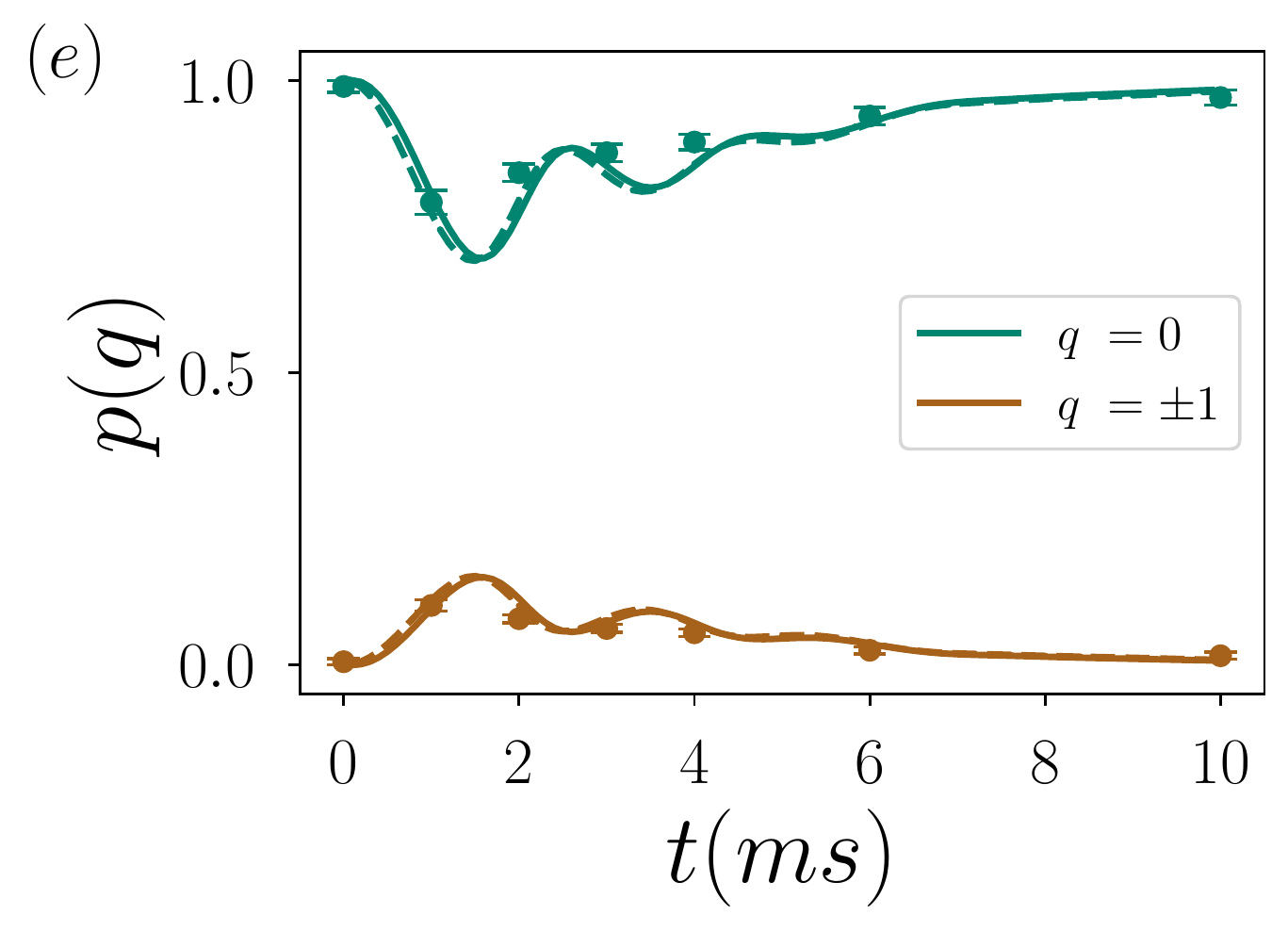}
\end{minipage}
\hskip -0.75ex
\begin{minipage}[b]{0.48\linewidth}
\centering
\includegraphics[width=\textwidth]{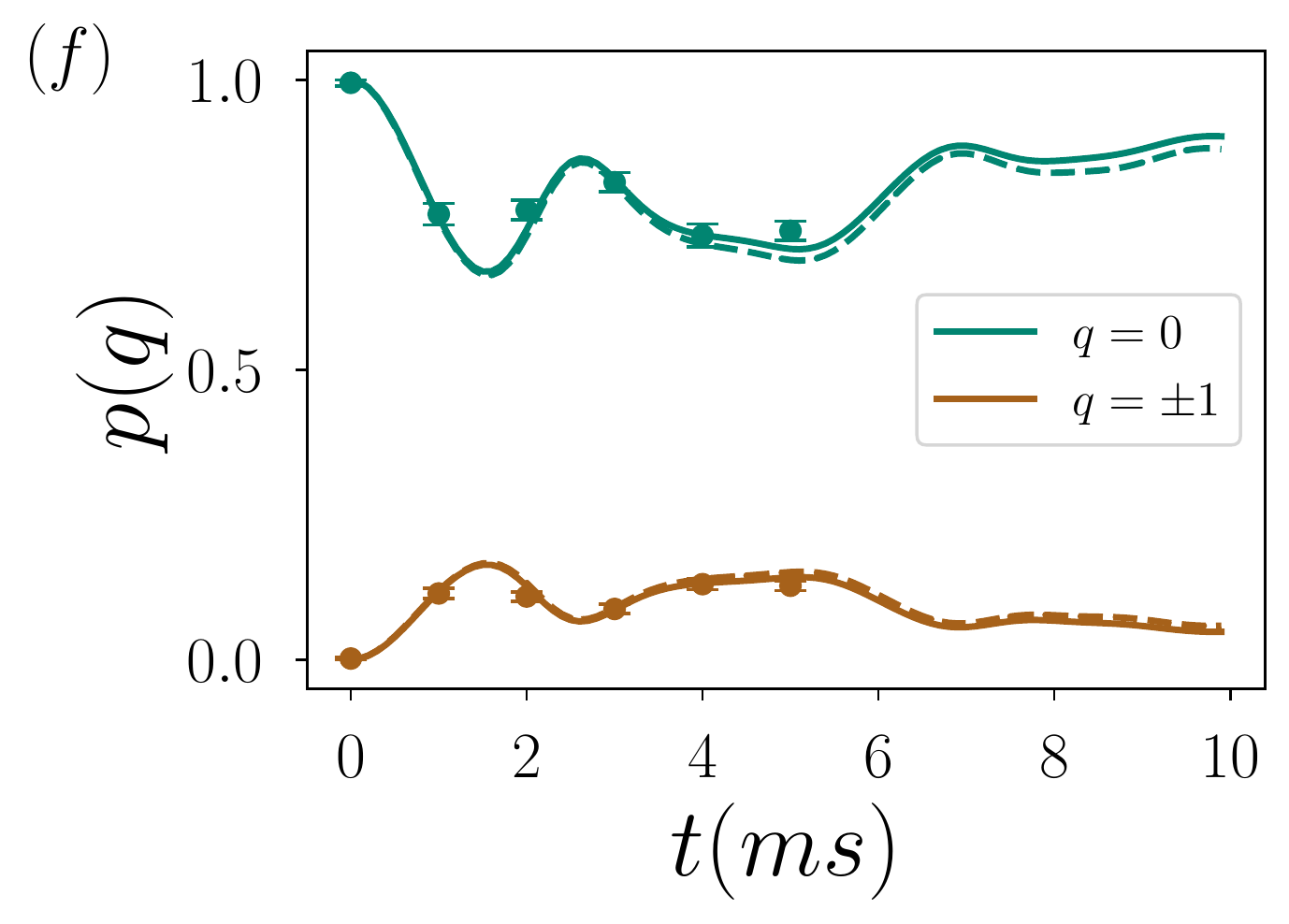}
\end{minipage}
\caption{Additional experimental observations: Panels (a-b) show the measured R\'enyi 2-OE and correspondingly, panels (c-d) the measured values of $\tilde S^{(2)}(\rho_{AB})$ and $R^{(2)}(\rho_{AB})$ relating to R\'enyi 2-OE as in Eq.~\eqref{eq:renyi2OE} for a reduced density matrix of 4 ions from a total a system consisting of $N = 20$ (left panels) and $N = 10$ (right panels). 
We observe the two phases of the entanglement barrier that is separated by a black vertical dashed line for panels a) and c) given by: (1) the growth phase followed by (2) the decay phase. 
Panels (e-f) show the corresponding populations $p(q)$ for symmetry sectors $q = 0 \, ,\, \pm 1$, on a reduced density matrix of 4 ions taken from their respective total system of ($N = 20$ and $N = 10$). The points show experimental results with the error bars calculated with Jackknife resampling.
Lines correspond to numerical simulations of the unitary dynamics (dashed) and including dissipation (solid).}
\label{fig:fig2mt} 
\end{figure}

Overall Fig.~\ref{fig:fig2mt}b) and d) show excellent agreement of the experimental data with the numerically modeled results for the 10 ion experiment.
On the other hand, it is quite surprising to see that even though the individual estimations of $\tilde S^{(2)}(\rho_{AB})$ and $R^{(2)}(\rho_{AB})$ from the 20 ion experiment as shown in Fig.~\ref{fig:fig2mt}c) have systematic shifts of the experimentally measured values caused likely due to an imperfect modeling of decoherence during the experiment and the measurement protocol, the corresponding measured R\'eyni 2-OE shows quite good agreement with the theoretical model as in Fig.~\ref{fig:fig2mt}a). This suggests a robustness feature of the R\'enyi 2-OE where errors in estimations of the two terms compensate each other. We also remark that the measured values of R\'enyi 2-OE are lower as shown in Fig.~\ref{fig:fig2mt}(a-b) from the numerical simulations of the experiment.

For the present model, the conserved quantity is the magnetization, i.e $Q_{AB}=\sum_{i\in AB}\sigma_i^z$ with $\sigma_i^z$ the $z$-Pauli matrix acting on the ion qubit $i$, c.f. App.~\ref{app:detailsexperiment}.
The corresponding symmetry-resolutions for the considered bipartitions of $N=20$ and $N=10$ ions are shown in Fig.~\ref{fig:summary}b) and c). Their respective populations in a given symmetry sector $q$ is given by \mbox{$p(q) = \mathrm{Tr}\big(\Pi_{q} \mathrm{Tr}_B(\ket{\rho_{AB}} \bra{\rho_{AB}})\big)/\mathrm{Tr}(\rho_{AB}^2)$} where $\Pi_q$ is the projector onto the eigenspace of the charge sector $q$ for system $A$ ($q=0$ being the sector initially populated).
This is highlighted in Fig.~\ref{fig:fig2mt}e) and f), respectively.
At $t = 0$, we see that the $q = 0$ sector is substantially populated, while the other sectors $q = \pm 1$ increase in population as a function of time.
In particular, for the 20 ion system, as shown in Fig.~\ref{fig:fig2mt}e), we observe very low population for the section $q = \pm 1$ as it decays as a function of time. This and the finite measurement statistics available from the experiment prevent us from resolving the experimental points for symmetry resolution sector of $q = \pm 1$ for later times.
In general, we also observe from Fig.~\ref{fig:summary}b) and c), that the sector $q = 0$ follows the features of the R\'enyi 2-OE. This translates, as shown in Fig.~\ref{fig:summary}b), to an entanglement barrier for $q = 0$ sector for the 20 ion system. One can also note the absence of the barrier for $q =0$ sector from the symmetry resolution of the 10 ion system.

\subsection{Interpretation in the quasi-particle picture}\label{sec:quasi-pic}

Interestingly, our experimental results can also be interpreted based on free fermions calculations detailed in Sec.~\ref{sec:reduced_quench}, with which we can qualitatively reproduce the behavior of the OE and the SROE for systems of $10$ and $20$ qubits. The analogy between the experimental setup and our free fermion model originates in the fact that the breaking of integrability in the experiment is weak~\cite{Hauke2013,jurcevic2014quasiparticle}.
Therefore, the short-time dynamics is comparable to the one of an integrable system, where entanglement generation can be qualitatively understood in terms of entangled pairs of quasiparticles propagating freely through the system~\cite{albac-14}. Deviations from integrable dynamics become relevant only on longer time scales which are not accessible with the available data.
As pointed out above, comparing Figs.~\ref{fig:summary}b) and \ref{fig:summary}d) for $20$ qubits and $20$ fermionic sites, respectively, we observe the same barrier shape for the OE, with oscillations due to the small subsystem size. The same barrier is found for the SROE for $q=0$, while for $q=\pm 1$, there is no apparent decay of the OE at long times.

We now come back to the fact that we did not observe a single-peaked barrier 
for $10$ ions, Fig.~\ref{fig:summary}c) and Fig.~\ref{fig:fig2mt}b). Importantly, this feature is also noticeable in our free fermions simulations with $10$ sites, see Fig.~\ref{fig:summary}e). Instead of a single barrier, the free fermion OE displays a double-peaked shape. The second peak can be understood from a quasi-particle picture as a consequence of the subsystem $A$ being particularly close to the boundary, as we explain now.

Recall that $A=[2,3]$ and $B=[4,5]$ for the chain of $10$ ions, with ions labelled from $1$ to $10$. Importantly, part $C$ then consists of two asymmetric pieces, $C = \{ 1\} \cup [6,10]$, with a very short domain on the left and a longer one on the right. The first growth phase of the OE is interpreted as originating from pairs of quasi-particles, initially located at the same position, that travel through the system in opposite directions and generate entanglement when one member of the pair is in $A$ and the other is in $B$. This interpretation of entanglement growth is usually given for the standard entanglement entropy~\cite{cc-05}, but also carries over to the OE. After the OE reaches its first maximum, it decreases because some quasiparticles, that formerly belonged to pairs shared between $A$ and $B$, arrive in $C$ and therefore stop contributing to the OE of $\rho_{AB}$. If the subsystems $A$ and $B$ were far away from the boundaries, then the OE would ultimately go to zero as the number of pairs shared between $A$ and $B$ would eventually vanish. This does not happen here, because the particles that escape from $A$ to $C$ (i.e. go from site $2$ to site $1$ in the chain) are soon reflected against the left boundary of the system. Consequently, they come back and are re-injected into $A$. As a result, the OE grows again, which explains the second peak in Fig.~\ref{fig:summary}e). The decay of that second peak occurs because, after the reflection, both members of a pair travel to the right, so they ultimately escape to the right half-system $[6,10]$.

The decay of the second peak is not visible in the experiment, Fig.~\ref{fig:summary}c).
Based on our numerical simulations of the experiment, as shown in Fig.~\ref{fig:fig2mt}b), we observe a decay occurring at a later time which is not accessible within the time window of the 10 ion experiment.

It should be possible to adapt the quasi-particle picture to describe both the experimental data and our free fermion results more quantitatively, following what is done for the time evolution of the entanglement entropy in nearly integrable dynamics, see e.g.  Ref.~\cite{bertini2020prethermalization} which implements previous ideas for local observable~\cite{begn-15}. This is however far beyond the scope of this paper.

\section{Symmetry-resolved operator entanglement in free fermionic chains}
\label{sec:reduced_quench}
So far we have presented results for finite-size systems, in direct connection with the experimental setup. We have shown in Fig.~\ref{fig:summary} that the qualitative features of the trapped ion experimental setup can also be observed in free-fermion chains under unitary evolution, despite the fact that these free-fermion models are short ranged and have no dissipation.

This raises the question as to whether one can understand more about the OE and the SROE of the reduced density matrix by studying free-fermion chains in the thermodynamic limit.
In this section, we show that the thermodynamic limit can be tackled analytically, unveiling some interesting properties of the  SROE, such as the time delay of the charge sectors or the equipartition. 

A direct analytical calculation of the SROE from the definition \eqref{eq:renyi-n-OE-SR} is difficult, but we can apply a trick similar to what has been done for the standard entanglement resolution \cite{moshe2018symmetry,xavier2018equipartition,murciano2020symmetry,asymmetry}, consisting in computing instead the charged moments of the reduced density matrix. Using the vectorization of the operator $\rho_{AB}$, $\left| \rho_{AB}\right>$, the object we want to compute is 
\begin{equation}\label{eq:Zcal}
    \mathcal{Z}_{\alpha}(q)=\sum_j (\lambda_j^{(q)})^{2\alpha}=\frac{\mathrm{Tr}[\Pi_q \left(\mathrm{Tr}_{B \otimes B}(\ket{\rho_{AB}}\bra{\rho_{AB}})\right)^{\alpha}]}{(\mathrm{Tr}[\rho^2_{AB}])^{\alpha}},
\end{equation}
where $q$ labels the (integer) eigenvalues of $\mathcal{Q}_A$ and $\Pi_q$ is the projector on the corresponding eigenspace of $\mathcal{Q}_A$, as already mentioned above.  To do so, we use the Fourier representation of
$\Pi_q$,
\begin{equation}\label{eq:fourierp}
  \Pi_q=\int_{-\pi}^{\pi}\frac{d\theta}{2\pi}e^{-iq\theta}e^{i\theta \mathcal{Q}_A}.
\end{equation}
Plugging Eq.~\eqref{eq:fourierp} into Eq.~\eqref{eq:Zcal}, we get
\begin{equation}\label{eq:Zcal2}
    \mathcal{Z}_{\alpha}(q)=\int_{-\pi}^{\pi}\frac{d\theta}{2\pi}e^{-iq\theta} Z_{\alpha}(\theta),
\end{equation}
where the charged moment $Z_\alpha(\theta)$ is defined as
\begin{equation}\label{eq:charged_moments}
    Z_{\alpha}(\theta)=\frac{1}{(\mathrm{Tr}[\rho^2_{AB}])^{\alpha}} \mathrm{Tr} [\left(\mathrm{Tr}_{B \otimes B}(\ket{\rho_{AB}}\bra{\rho_{AB}})\right)^{\alpha}e^{i\theta \mathcal{Q}_A}],
\end{equation}
The charged moment is the main object that we need to evaluate; we explain how to do so in the next subsection. In terms Eq.~\eqref{eq:Zcal}, the SROE reads
\begin{equation}\label{eq:FTOE}
\begin{split} S_{q}^{(\alpha)}(\rho_{AB})&=\frac{1}{1-\alpha}\log \frac{\mathcal{Z}_{\alpha}(q)}{[\mathcal{Z}_{1}(q)]^\alpha},
\end{split}
\end{equation}
while in terms of Eq.~\eqref{eq:charged_moments} the total OE is
\begin{equation}
S^{(\alpha)}(\rho_{AB})=\frac{1}{1-\alpha}\log Z_{\alpha}(0).
\end{equation}

\subsection{Free-fermion techniques for the OE}\label{sec:correlation}
For the eigenstates of quadratic lattice Hamiltonians, it is possible to compute the entanglement entropies in terms of the eigenvalues of the correlation matrix of the subsystem \cite{Peschel2,p-03}. This trick can be applied also for the computation of the OE and, more generally, of the charged moments in Eq.~\eqref{eq:charged_moments}.

Let us take a free-fermionic chain of length $N$ with $U(1)$ symmetry, described by the Hamiltonian
\begin{equation}\label{eq:XX}
H=-\frac{J}{2}\sum_{i=1}^N(c^{\dagger}_{i+1}c_i +\mathrm{h.c.})
\end{equation}
where $c_i^{\dagger}$ ($c_i$) is the creation (annihilation) operator such that the anticommutator obeys $\{c_i,c_{j}^{\dagger}\}=\delta_{ij}$ and also \mbox{$c_{N+1}=c_1, c^{\dagger}_{N+1}=c^{\dagger}_1$}, i.e. we impose periodic boundary conditions. For the sake of simplicity, we set $J=1$ from now on and remind the reader that we do the same for the reduced Planck constant ($\hbar = 1$) and the lattice constant $a = 1$.
The reduced density matrix $\rho_{AB}$ for a subsystem $A \cup B$, where $A \cup B=[1,\ell_A] \cup [\ell_A+1,\ell_A+\ell_B]$ consists of two adjacent intervals,
can be put in a diagonal form as 
\begin{equation}\label{eq:rhodiag}
    \rho_{AB}=\bigotimes_{k=1}^{\ell_A+\ell_B}\dfrac{e^{-\lambda_kd^{\dagger}_kd_k}}{1+e^{-\lambda_k}},
\end{equation}
where $e^{-\lambda_k}=n_k/(1-n_k)$, with $n_k$ being the occupation number at a given wave vector $k$ and $d_k$'s defined as fermionic operators satisfying  $\{d_k,d_{k'}^{\dagger}\}=\delta_{kk'}$.
It is more convenient to write Eq.~\eqref{eq:rhodiag} as
\begin{equation}\label{eq:rhodiag2}
\begin{split}
\rho_{AB} =&\bigotimes_{k=1}^{\ell_A+\ell_{B}} \frac{\ket{0}_k\bra{0}_k+e^{-\lambda_k}\ket{1}_k\bra{1}_k}{1+e^{-\lambda_k}}\\
=&\bigotimes_{k=1}^{\ell_A+\ell_{B}} [(1-n_k)\ket{0}_k\bra{0}_k+n_k\ket{1}_k\bra{1}_k],
\end{split}
\end{equation}
so that by applying the vectorization trick in Eq.~\eqref{eq:vectorisation} for $\rho_{AB}$, we get 
\begin{equation}
\begin{split}
    \frac{\ket{\rho_{AB}}}{\sqrt{\mathrm{Tr}[\rho^2_{AB}]}} &=\bigotimes_{k=1}^{\ell_A+\ell_{B}} \frac{[(1-n_k)\ket{0}_k\ket{0}_{\tilde{k}}+n_k\ket{1}_k\ket{1}_{\tilde{k}}]}{\sqrt{n^2_k+(1-n_k)^2}}\\ &=\bigotimes_{k=1}^{\ell_A+\ell_{B}} \frac{[1-n_k+n_kd^{\dagger}_k\tilde{d}^{\dagger}_k] \ket{0}}{\sqrt{n^2_k+(1-n_k)^2}},
    \end{split}
\end{equation}
where the $\tilde{d}_k$ operators are the copies of the $d_k$'s introduced in the vectorization process, and $\ket{0}$ is the state annihilated by all the $d_k$’s and $\tilde{d}_k$'s. 
The correlation matrix of the state $\ket{\rho_{AB}}$ reads
\begin{equation} \label{eq:correl_mtx}
\begin{split}
C_{kk'}=&\bra{\rho_{AB}} \begin{pmatrix}
d_k^{\dagger}\\
\tilde{d}_k
\end{pmatrix}
\begin{pmatrix}
d_{k'} \, \tilde{d}_{k'}^{\dagger}
\end{pmatrix}\ket{\rho_{AB}}\\&=\frac{\delta_{kk'}}{n^2_k+(1-n_k)^2}\begin{pmatrix}
& n_k^2 &n_k(1-n_k) \\
&n_k(1-n_k)  & (1-n_k)^2 \\
\end{pmatrix}.
\end{split}
\end{equation}
In the basis of $d_k, \tilde{d}_k$'s, the supercharge operator takes the form 
\begin{equation}\label{eq:charge2}
\mathcal{Q}=(\sum_k d^{\dagger}_k d_k)\otimes \mathbbm{1}-\mathbbm{1}\otimes (\sum_k\tilde{d}^{\dagger}_k \tilde{d}_k)^T. 
\end{equation}
We can collect the operators into the vector $\mathbf{f}=(d_1,\dots d_{\ell_A+\ell_B},\tilde{d}^{\dagger}_1 \dots  \tilde{d}^{\dagger}_{\ell_A+\ell_B})^T$ (making the identity operators in Eq.~\eqref{eq:charge2} implicit, for simplicity, and noting we can ignore the transpose) such that $\mathcal{Q}$ reads \mbox{$\mathcal{Q}=\mathbf{f}^{\dagger}\mathbf{f}-(\ell_A+\ell_B)$}, where $\ell_A+\ell_B$ acts just as an additive constant here. 

At this point, we can compute the $2(\ell_A+\ell_B)\times 2(\ell_A+\ell_B)$ correlation matrix as 
\begin{equation}\label{eq:corrfull}
    C_{AB}=\bigoplus_{k=1}^{\ell_A+\ell_B}C_{kk},
\end{equation}
and by doing a Fourier transform, we can write $C_{AB}$ in the spatial basis.
To evaluate the charged moments in Eq.~\eqref{eq:charged_moments}, we just have to focus on the subsystem $A$, i.e. we can restrict the supercharge operator to $\mathcal{Q}_A$ and the Fourier transform of the correlation matrix in Eq.~\eqref{eq:corrfull} to the subspace corresponding to the subsystem $A$. Diagonalising the latter matrix, we get $2\ell_A$ real eigenvalues $\xi_i$ between 0 and 1. 

Therefore, one can compute the charged moments of the reduced density matrix built from $\ket{\rho_{AB}}$ in terms of the eigenvalues $\xi_i$ as
\begin{equation}\label{eq:lattice}
Z_{\alpha}(\theta)=e^{-i\theta (\ell_A+\ell_B)}\prod_{a=1}^{2\ell_A}(\xi_a^{\alpha}e^{i\theta}+(1-\xi_a)^{\alpha}).
\end{equation}
Using Eqs. \eqref{eq:Zcal2} and \eqref{eq:FTOE}, we can compute exactly the SROE for the reduced density matrix of a free fermionic chain.
The same trick also allows the computation of the total R\'enyi-$\alpha$ OE as 
\begin{equation}
    S^{(\alpha)}(\rho_{AB})=\frac{1}{1-{\alpha}}\sum_{a=1}^{2\ell_A}\log[\xi_a^{\alpha}+(1-\xi_a)^{\alpha}].
\end{equation}
\subsection{Charged moments:  a quasiparticle picture}

Let us now consider a global quantum quench from an initial conformal invariant state with an evolution Hamiltonian given by the continuum limit of Eq.~\eqref{eq:XX} \cite{f-13}. The emerging
quasiparticles move with a single velocity and, in the space-time scaling limit $t,\ell_A,\ell_B\gg \tau_0$  (with $\tau_0$ an ultraviolet cutoff), we can introduce the function (assuming, without loss of generality, $\ell_A \leq \ell_B$)
\begin{equation}
f_{\ell_A,\ell_B}(t) =\left\{ \begin{array}{ll}
    t & \text{for } 0 \leq 2t \leq \ell_A \\
    \ell_A/2 & \text{for } \ell_A \leq 2t \leq \ell_B \\
    (\ell_A+\ell_B)/2 - t & \text{for } \ell_B \leq 2t \leq (\ell_A+\ell_B) \\
    0 & \text{for } (\ell_A+\ell_B) \leq 2t
\end{array} \right.
\end{equation}
so that the charged moments read
\begin{equation}\label{eq:renyi-d-cft}
\log Z_{\alpha}(\theta)=\frac{\pi \Delta_{\alpha}^{\theta}}{\tau_0} f_{\ell_A,\ell_B}(t),
\end{equation}
where \cite{moshe2018symmetry}
\begin{equation}\label{eq:scalingK}
    \Delta_{\alpha}^{\theta}=\frac{1}{12}\left(\alpha-\frac{1}{\alpha} \right)+\frac{1}{\alpha}\left(\frac{\theta}{2\pi} \right)^2.
\end{equation}
From this result, which is valid for a conformal field theory (CFT), one can formulate a quasiparticle picture for the charged moments of free fermionic models with global conserved $U(1)$ charge, whose  quench dynamics starts from initial states 
that are also invariant under $U(1)$ symmetry. 
This is obtained from the 
CFT result in Eq.~\eqref{eq:renyi-d-cft} by first replacing $t\to |v(k)|t$, with $|v(k)|$ being the velocity of quasiparticles, which for conformal invariant systems is fixed to be $v(k)=1$. Then, we should 
integrate over the quasiparticles with quasimomentum $k$, but properly accounting for the density (in momentum space) of the thermodynamic charged moments $z_{\alpha}(k,\theta)$ in the stationary state \cite{albac-14,ac-18,c-20}. The latter, can be inferred from the results for charged moments of state entanglement \cite{bpc-21,pbc-21} and the final result is the replacement $\pi\Delta^{\theta}_{\alpha}/\tau_0 \to 2z_{\alpha}(k,\theta)$
\begin{equation}\label{eq:qpcharged}
\begin{split}
&\log Z_{\alpha}(\theta)=\int_{-\pi}^{\pi} \frac{dk}{2\pi}2z_{\alpha}(k,\theta)f_{\ell_A,\ell_B}(|v(k)|t).
\end{split}
\end{equation}
In order to have a predictive formula, one has to fix the function $z_{\alpha}(k,\theta)$ in Eq.~\eqref{eq:qpcharged}. 
Here we focus on out-of-equilibrium protocols for free-fermion models, whose time evolution is given by the Hamiltonian in Eq.~\eqref{eq:XX}.
In this case, $z_{\alpha}(k,\theta)$ is determined from the  population of the modes $n_k$ of the post-quench Hamiltonian in the stationary state ~\cite{c-18,c-20} and it reads
\begin{equation}
z_{\alpha}(k,\theta)=\log [e^{i\theta}n_k^{\alpha}+(1-n_k)^{\alpha}]-i\theta /2.
\end{equation}
 For concreteness, from now on we restrict to a quench from the N\'eel state, for which $n_k=1/2$ for all $k$ \cite{pbc-21}, so that the charged moment (in Eq.~\eqref{eq:qpcharged}) becomes 
\begin{equation}\label{eq:cmquasi}
\begin{split}
&\log Z_{\alpha}(\theta)=\big[2(1-\alpha)\log 2 +2\log (\cos (\theta/2))\big]\mathcal{J}(t)
\end{split}
\end{equation}
where $\mathcal{J}(t)$ is defined as
\begin{equation}
    \mathcal{J}(t)=\int_{-\pi}^{\pi} \frac{dk}{2\pi}f_{\ell_A,\ell_B}(|v(k)|t).
\end{equation}
and $|v(k)|=|\sin(k)|$. The function $\mathcal{J}(t)$ displays the same qualitative features of Eq.~\eqref{eq:renyi-d-cft}, where $|v(k)|=1$: $\mathcal{J}(t)$ grows until $t<\ell_A/2$, then it presents a plateau barrier between $\ell_A/2<t<\ell_B/2$, it decays again for $\ell_B/2<t<(\ell_A+\ell_B)/2$, and, eventually, it saturates to $0$ for $t>(\ell_A+\ell_B)/2$. In other words, $\mathcal{J}(t)$ behaves as a barrier, a characteristic that we will find also for the SROE in the following section.
\subsection{Time delay, barrier and equipartition}
From the computation of the charged moments done above, the symmetry resolved moments read
\begin{equation}\label{eq:saddletodo}
\mathcal{Z}_{\alpha}(q)=2^{2(1-\alpha)\mathcal{J}(t)}\int_{-\pi}^{\pi}\frac{d\theta}{2\pi}e^{-i\theta q}\left(\cos \frac{\theta}{2} \right)^{2\mathcal{J}(t)}.
\end{equation}
As already pointed out for the usual symmetry resolved entropies in \cite{bpc-21,pbc-21}, this expression formally assumes negative values for $\mathcal{J}(t)<|q|$, so it means we have to replace it with $\mathcal{Z}_{\alpha}(q)=0$. This allows us to identify a {\it delay time} $t_D$ such that the SROE in a given charge sector starts only after $t_D$. 
The equation $\mathcal{J} (t_D) = |q|$ reads (as long as $v_Mt_D< \frac{1}{2} \mathrm{Min}(\ell_A,\ell_B)$  self-consistently and $v_M \equiv \mathrm{max}( v(k) )= 1$)
\begin{equation}
    \int_{-\pi}^{\pi}\frac{dk}{2\pi}|\sin(k)|t_D = |q|, 
\end{equation} 
and we can conclude that $t_D=\pi |q|/2$ for $|q|<\mathrm{Min}(\ell_A,\ell_B)/\pi$.
Therefore, and after simplification we find that the SROE is given by
\begin{equation}\label{eq:sq1}
S_{q}^{(\alpha)}(\rho_{AB})=
\begin{cases}
0 & (t\leq t_D) \\
2\mathcal{J}(t)\log 2+\log \mathcal{Z}_1(q) & (t>t_D).
\end{cases}
\end{equation} 
We remark that this expression does not depend on the R\'enyi index $\alpha$ , also for $\alpha = 1$. Notice that for large $\mathcal{J}(t)$, i.e. in the scaling limit $\mathcal{J}(t)>|q| \gg 1$, the integral \eqref{eq:saddletodo} can be computed by saddle-point approximation,  obtaining (for more details see \cite{pbc-21} where the same integral appears) 
\begin{equation}\label{eq:sqsp}
S_{q}^{(\alpha)}(\rho_{AB})=2\mathcal{J}(t) h\bigg(\frac{1+q/\mathcal{J}(t)}{2}\bigg)
\end{equation}
 where $h(x) = -x \log x - (1-x) \log(1-x)$ is the well-known binary entropy function.
The comparison between this
formula and the numerical results in the tight-binding model chain is displayed in  the top-right panel of Fig. \ref{fig:summary}f).
The solid lines correspond to Eq.~\eqref{eq:sqsp} for $t>t_D$, obtained from the saddle-point approximation of Eq.~\eqref{eq:saddletodo}.
The
agreement is good and we can also observe that there are some charge sectors with zero entanglement for $t<t_D$. However, for $t > (\ell_A+\ell_B)/2$ the discrepancy between the numerics
and the analytical prediction in Eq.~\eqref{eq:sqsp} is larger. One explanation could be that at finite $\ell_A$ and $t$, the data exhibit some small corrections, and our prediction is recovered only
in the scaling limit $t,\ell_A,\ell_B\to\infty$ with their ratio fixed.

For $|q| \ll \mathcal{J}(t)$ we find from Eq.~\eqref{eq:sqsp}
\begin{equation}\label{eq:equip}
    S_{q}^{(\alpha)}(\rho_{AB})=\mathcal{J}(t)\left(2\log 2-\frac{q^2}{\mathcal{J}(t)^2}\right).
\end{equation}
This result states that for small $|q|$ there is an effective equipartition of the OE with violations of order \mbox{$q^2/\mathcal{J}(t)$}. We compare the exact result for the SROE in the scaling limit reported in Eq.~\eqref{eq:sq1} (solid lines) with its asymptotic expansions in Fig. \ref{fig:asympt}. We notice that as $\ell=\ell_A+\ell_B$ increases (here $\ell_A=\ell_B$), the approximation in Eq.~\eqref{eq:sqsp} (large dashed lines) improves since $\mathcal{J}(t)$ also increases. The tiny dashed lines represent the further approximation in Eq.~\eqref{eq:equip}, which also improves as $\mathcal{J}(t)$ increases for the small charge value ($q=4$) that we plot. We observe that the SROE is small both at short and at large times, and it blows up linearly in the transient regime $t\leq (\ell_A+\ell_B)/2$, as for the total OE \cite{dubail}. 

We conclude by commenting on Figs.~\ref{fig:summary}d) and \ref{fig:summary}e), obtained through the free-fermion techniques described in Sec. \ref{sec:correlation}. We show that the dynamics of the SROE in the different charge sectors is affected by the finite size of the system. In particular, for $N=20$ one can observe the entanglement barrier only in the sector $q=0$, while for $q=1$ the absence of the decay is consistent with the experimental results of Fig. \ref{fig:summary}b). Moreover, for this system size the total OE presents a single peak, while for $N=10$, we notice the presence of two peaks in the total OE, which can be justified by the quasiparticle picture explained at the end of Sec. \ref{sec:experimentalresults}. 
 \begin{figure}
 \centering
  \includegraphics[width=\linewidth]{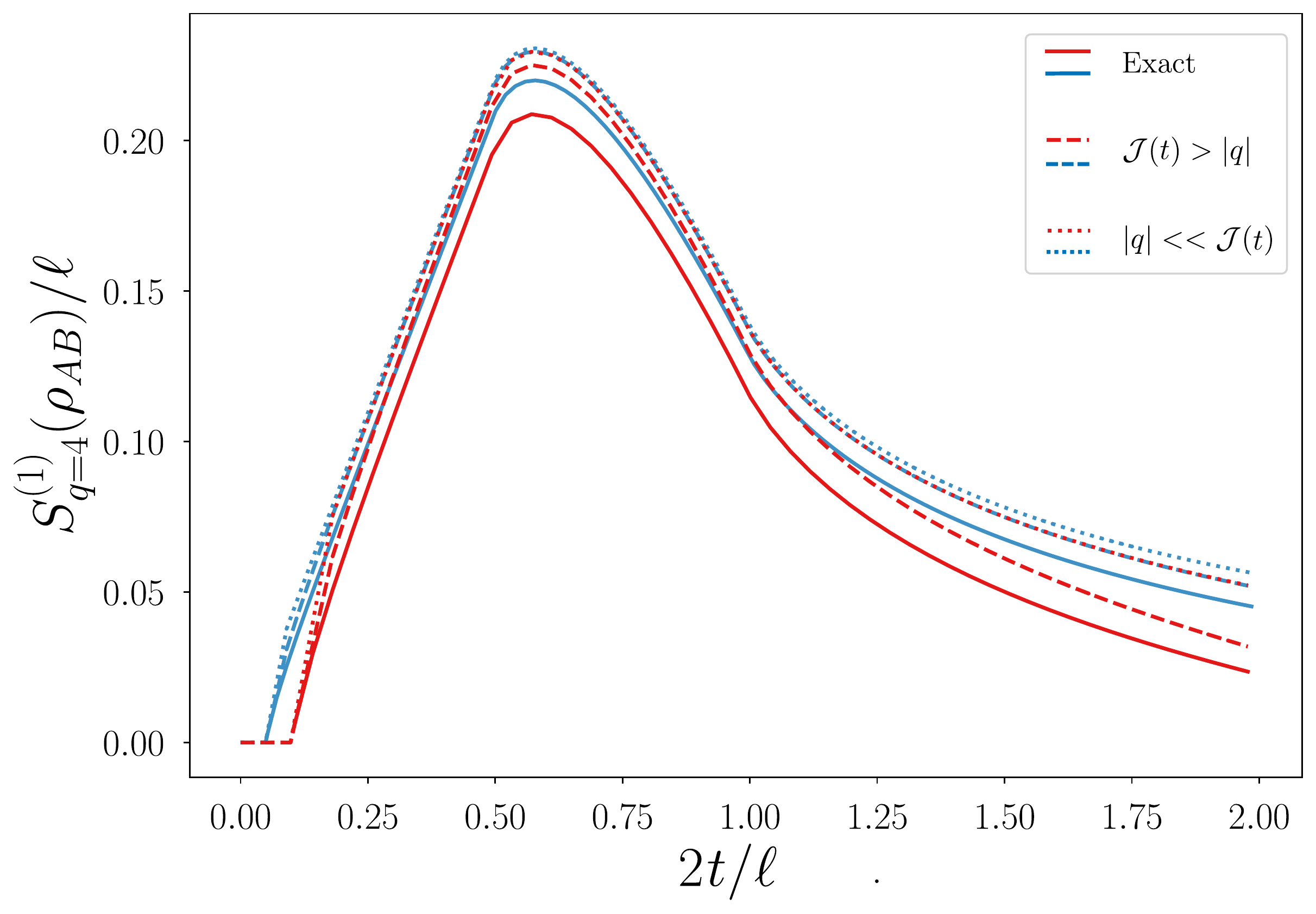}
 \caption{SROE in the scaling limit: Comparison between the analytical expression in Eq.~\eqref{eq:sq1} (solid line) and its asymptotic approximation in Eq.~\eqref{eq:sqsp} (large dashed line) and Eq.~\eqref{eq:equip} (tiny dashed line).
 Here $q=4$, $\ell=\ell_A+\ell_B$ with $\ell_A=\ell_B$, the blue lines corresponds to $\ell=256$, while the red ones to $\ell=128$. }\label{fig:asympt}
 \end{figure}

\section{Conclusions}\label{sec:conclusion}
This manuscript is devoted to a thorough analysis regarding the operator entanglement (OE) of a reduced density matrix after a global quantum quench, as well as its symmetry resolution.
These quantities first grow linearly in time, before they decrease again and eventually saturate to a finite value. The presence of such an entanglement barrier is strongly affected by the finite size of the system, as we demonstrated here, based on experimental data. This feature is also visible for free fermionic systems evolving under a unitary evolution.

The experimental results, also supported by tensor network simulations, have been obtained by a novel post-processing method of randomized measurement data, dubbed \emph{batch shadow estimator}, that has practical applications to probe non-linear properties of quantum many-body systems.
This method provides a faster and more efficient data-treatment technique with respect to the known ones~\cite{RMtoolbox} and enabled us to actually estimate the OE from existing experimental data~\cite{Brydges2019probing}. 

We observe the presence of the entanglement barrier of the reduced density matrix of a partition of 4 ions out of $N=20$, both for the total OE and and its symmetry resolved counterpart (SROE) in the charge sector $q=0$. However, finite size effects prevent the experimental observation of such a barrier in the charge sectors $q =\pm 1$ and for $q=0$ in the case of $N=10$.
For $N=20$, in the charge sectors $q =\pm 1$, the available measurement statistics has only allowed us to explore the early time behavior of the SROE. 

For small system sizes $N$, the phenomenology discussed in the previous paragraph 
can be also observed in free fermionic systems without dissipation. Therefore, guided by conformal field theory and free-fermion techniques, we showed that the semi-classical picture of moving quasi-particles \cite{cc-05,albac-14} can be adapted to this context. This leads to a general conjecture for the charged OEs whose Fourier transform gives the desired SROE. Beyond the barrier, we observe a time delay proportional to the charge sector and an effective equipartition for small $q$.

Because of this phenomenology, 
we expect our main physical  findings to show up for rather generic quench protocols. However, it would be very interesting to engineer situations in which some of them are absent, e.g. with the entanglement barrier appearing only in certain charge sectors, breaking equipartition. 

It is worthwhile to point out that the time evolution of the total OE is closely related to other entanglement measures such as the reflected entropy~\cite{reflected1,reflected2} (which is the OE of $\sqrt{\rho_{AB}}$), negativity~\cite{cct-12,cct-13,ez-15,shapourian2017partial}, and temporal entanglement ~\cite{Lerose2021,Sonner2021,Giudice2022}: in these latter cases, the connection is merely technical, but the fact that they can be computed in a similar way leads to analogous results, like the entanglement barrier of the logarithmic negativity after a quench \cite{Alba_2019}. Our works naturally paves the way for their symmetry resolution and to understand whether their connection could be understood sector by sector. 

To conclude, we remark that the OE of operators different from the reduced density matrix are known to capture important universal properties of the dynamics~\cite{Prosen2007,Pizorn2009,dubail,Zhou2017,jonay2018coarse,Wang2019barrier,Alba2019operator,Bertini2020,bertini2020operator,Alba2022rise}. For instance, the OE of the evolution operator $U(t) = e^{-i H t}$ grows linearly in ergodic phases~\cite{dubail,jonay2018coarse}, but only logarithmically in localized phases~\cite{Zhou2017,dubail}. Another example is the OE of a local operator $O$ evolving in the Heisenberg-picture, i.e.\ $O(t) = e^{i H t} O e^{-i H t}$. There, the OE grows linearly in systems with chaotic dynamics~\cite{jonay2018coarse}, but only logarithmically for integrable dynamics~\cite{Alba2019operator,Bertini2020,bertini2020operator}. 
It is then natural to wonder what happens to their symmetry resolved (SR) version, which certainly deserves future investigation.

\section{Acknowledgements}

V.V. is grateful to A. Santini for useful discussions on related topics.
We thank A. Elben for helping us to reprocess the experimental data of Ref.~\cite{Brydges2019probing}. 
P.C. and S.M. acknowledge support from ERC under Consolidator grant number 771536 (NEMO).
A.R. is grateful for the support by Laboratoire d'excellence LANEF in Grenoble (ANR-10-LABX-51-01) and from the Grenoble Nanoscience Foundation.
B.V.\ and M.V. acknowledge funding from the French National Research Agency (ANR-20-CE47-0005, JCJC project QRand), and from the Austrian Science Foundation (FWF, P 32597 N). The work of V.V. was partly
supported by the ERC under grant number 758329
(AGEnTh), and by the MIUR Programme FARE (MEPH). J.D. acknowledges support from the ANR (ANR-20-CE30-0017-02, project QuaDy) and from CNRS IEA (project QuDOD).
For our numerical simulations we used the quantum toolbox QuTiP~\cite{Johansson2013}.

\onecolumngrid
\pagebreak[4]

\appendix

\section{Entanglement conditions } \label{app:entanglementconditions}

In this Appendix, we derive our rigorous conditions to detect operator entanglement in bipartite mixed states $\rho_{AB}$. The starting point is the operator Schmidt decomposition 
\begin{equation}
\frac{\rho_{AB}}{\sqrt{\mathrm{Tr}\left[\rho_{AB}^{2}\right]}}=\sum_{i=1}^R \lambda_i O_{A,i} \otimes O_{B,i}\,,
\label{eq:schmidt-appendix}
\end{equation}
see also Eq.~\eqref{eq:schmidtO0} in the main text. Here, $R \geq 1$ denotes the operator Schmidt rank
and the Schmidt values $\lambda_1,\ldots,\lambda_R$ are nonnegative ($\lambda_i \geq 0$) and obey $\sum_{i=1}^R \lambda_i^2 =1$.
A seminal result in entanglement theory states that the Schmidt values of any separable state $\rho_{AB}$ must obey
\begin{equation}
\sum_{i=1}^R \lambda_i \leq 1/ \sqrt{\Tr \left(\rho_{AB}^2\right)}\,, \label{eq:ccnr-appendix}
\end{equation}
see e.g. \cite[Theorem~6]{GUHNE20091} and also Eq.~\eqref{eq:ccnr} in the main text. 
Conversely, if Eq.~\eqref{eq:ccnr-appendix} is violated, then $\rho_{AB}$ must be  entangled (across the bipartition into subsystems $A$ and $B$).
This entanglement criterion is called the computable cross norm or realignment (CCNR) condition and applies to any type of bipartite state. The main drawback is that it seems to rely on the explicit availability of an operator Schmidt decomposition~\eqref{eq:schmidt-appendix}. Obtaining such a decomposition requires full state tomography of the density matrix $\rho_{AB}$.

This apparent drawback was recently overcome in Ref.~\cite{liu2022detecting}. There, the authors point out that sums of higher powers of Schmidt values can be reformulated in terms of linear observables in tensor products of the original density matrix $\rho_{AB}$. 
This can be achieved by concatenating subsystem swap operators. Let $\mathbb{S}^{(X)}_{k,l}$ be the operator that swaps  the $k^\text{th}$ and $l^\text{th}$ copies of system $X$ ($X = A, B$, or $AB$ below). Namely, it acts as $\mathbb{S}_{k,l}^{(X)}(\ket{i}^{X_k} \otimes \ket{j}^{X_l} ) = \ket{j}^{X_k} \otimes \ket{i}^{X_l}$ on any pair of basis states for systems $X_k$ and $X_l$ (as indicated by the superscripts), and as the identity on all other systems. Then, the following relation holds:
\begin{equation}
\sum_{i=1}^R \lambda_i^4 = \frac{\mathrm{Tr} \left( \mathcal{S} \rho_{AB}^{\otimes 4}\right)}{\Tr \left(\rho_{AB}^2 \right)^2}\quad \text{where} \quad \mathcal{S} =  \mathbb{S}^{(A)}_{1,4} \otimes  \mathbb{S}^{(A)}_{2,3} \otimes \mathbb{S}^{(B)}_{1,2} \otimes \mathbb{S}^{(B)}_{3,4}. \label{eq:fourth-moment}
\end{equation}
The denominator $\mathrm{Tr} \left( \rho_{AB}^2 \right)^2$ is a consequence of the normalization in the lhs of Eq.~\eqref{eq:schmidt-appendix}, and validity of the overall expression readily follows from inserting the operator Schmidt decomposition into the rhs of Eq.~\eqref{eq:fourth-moment} and from using the fact that the operators $O_{A,i}$ and $O_{B,i}$ are all orthonormal, i.e.\ $\Tr\left( O_{A,i} O_{A,j}\right) = \Tr \left( O_{B,i} O_{B,j} \right) = \delta_{i,j}$.
It is also worth pointing out that the purity can also be reformulated as a linear observable on tensor products: 
\begin{equation}
\mathrm{Tr} \left( \rho_{AB}^2 \right)=\mathrm{Tr} \left(\mathbb{S}_{1,2}^{(AB)} \rho_{AB}^{\otimes 2} \right).
\end{equation}
This is relevant, because trace polynomials of the form $\mathrm{Tr} \left( O^{(n)} \rho_{AB}^{\otimes n}\right)$ can be measured directly in actual experiments by employing techniques from the randomized measurement toolbox~\cite{RMtoolbox}. This will be the content of the subsequent Appendix section. 
For now it is enough to remember that we know how to directly estimate the lhs of Eq.~\eqref{eq:fourth-moment}, while we are not aware of a direct estimation protocol for the lhs of Eq.~\eqref{eq:ccnr-appendix}.

So, how do we overcome this discrepancy between what can be measured (Eq.~\eqref{eq:fourth-moment}) and what is required to detect entanglement (violation of Rel.~\eqref{eq:ccnr-appendix})? We collect the positive Schmidt-values into an $R$-dimensional vector $l=(\lambda_1,\ldots,\lambda_R)$ and use fundamental $\ell_p$-norm relations to obtain a relation between $\| v \|_{\ell_1}=\sum_{i=1}^R |\lambda_i| = \sum_{i=1}^R \lambda_i$ (the last equation uses the fact that all Schmidt values are nonnegative) and $\|v \|_{\ell_4}^4 = \sum_{i=1}^R \lambda_i^4$. To achieve such a conversion, we can also use the fact that Schmidt values are normalized, i.e.\ $\|v \|_{\ell_2}^2=\sum_{i=1}^R \lambda_i^2 =1$. 
We can use the Berger's inequality~\cite{Berger97} that relates the $\ell_1$, $\ell_2$ and $\ell_4$ norms of any vector. This inequality then ensures
\begin{equation}
\sum_{i=1}^R \lambda_i = \| v \|_{\ell_1} \geq \frac{\| v \|_{\ell_2}^3}{\|v \|_{\ell_4}^2}
= \frac{1}{\left(\sum_{i=1}^R \lambda_i^4\right)^{1/2}}.
\label{eq:berger}
\end{equation}
 This relation is a simple consequence of Hölder's inequality and we refer to Ref.~\cite[Proof of Lemma~12]{Matthews09} for a quick derivation. 
 
We now have all ingredients in place to derive our experimentally accessible entanglement condition. A combination of Rel.~\eqref{eq:ccnr-appendix}, Eq.~\eqref{eq:berger} and Eq.~\eqref{eq:fourth-moment} implies that \emph{every} separable state $\rho_{AB}$ must obey
\begin{align}
\frac{1}{\Tr (\rho_{AB}^2)} \geq \left( \sum_{i=1}^R \lambda_i \right)^2 \geq \frac{1}{\sum_{i=1}^R \lambda_i^4} = \frac{\Tr \left(\rho_{AB}^2\right)^2}{\Tr \left( \mathcal{S} \rho_{AB}^{\otimes 4}\right)}\,, \label{eq:CCNR_Holder}
\end{align}
or equivalently:
\begin{equation}
\Tr \left( \mathbb{S}^{(A)}_{1,4} \otimes  \mathbb{S}^{(A)}_{2,3} \otimes \mathbb{S}^{(B)}_{1,2} \otimes \mathbb{S}^{(B)}_{3,4} \rho_{AB}^{\otimes 4}\right)=
\Tr \left( \mathcal{S} \rho_{AB}^{\otimes 4}\right) \geq \Tr \left( \rho_{AB}^2 \right)^3 = \Tr \left( \mathbb{S}^{(AB)}_{1,2} \rho_{AB}^{\otimes 2}\right)^3.
\label{eq:entanglement-condition}
\end{equation}
If this relation is violated, then we can be sure that the state $\rho_{AB}$ must be entangled. And, in stark contrast to the original CCNR condition, both the expression on the very left and the expression on the very right are directly accessible in an experiment. 
From Eq.~\eqref{eq:CCNR_Holder} we can also take logarithms and negate the sign to obtain an equivalent statement in terms of R\'enyi entropies.

\begin{proposition}[Entanglement condition]
Let $\rho_{AB}$ be a bipartite quantum state with R\'enyi 2-OE \mbox{$S^{(2)}(\rho_{AB})=-\log (\sum_i \lambda_i^4)$} and R\'enyi 2-entropy \mbox{$R^{(2)}(\rho_{AB})=-\log (\Tr (\rho_{AB}^2))$}. 
Then, the relation
\begin{equation}
    S^{(2)}(\rho_{AB}) > R^{(2)}(\rho_{AB})
\end{equation}
implies that $\rho_{AB}$ must be entangled (across the bipartition $A$ vs $B$).

\label{prop-ent}
\end{proposition}

Finally, we point out that a centering operation on the level of density matrices can substantially enhance the 
ability to detect entanglement. The key idea is to shift the original density matrix by
\begin{equation}
\rho_{AB} \mapsto \rho_{AB}-\rho_A \otimes \rho_B = X_{AB},
\end{equation}
where $\rho_A=\Tr_B (\rho_{AB})$ and $\rho_B = \Tr_A (\rho_{AB})$ are the reduced density matrices of $\rho_{AB}$. Note that this centered density matrix $X_{AB}$ is not physical, because it has negative eigenvalues and a vanishing trace.
The Schmidt coefficients $(\chi_1,\ldots,\chi_{R'})$ of this shifted density matrix are known to obey the \emph{enhanced realignement and computable cross norm condition}~\cite{Zhang2008,liu2022T2}:
\begin{equation}
\sum_{i=1}^{R'} \chi_i \leq \frac{\sqrt{1-\mathrm{Tr}(\rho_A^2)}\sqrt{1-\mathrm{Tr}(\rho_B^2)}}{\sqrt{\mathrm{Tr}([X_{AB}]^2)} },
\end{equation}
see e.g. Ref.~\cite{liu2022detecting}. We can now adjust the arguments from before to obtain the following relation that must hold for \emph{every} shifted version $X_{AB}$ of a separable state $\rho_{AB}$:
\begin{align}
\Tr \left( \mathcal{S} X_{AB}^{\otimes 4}\right) \geq  \frac{\Tr \left( X_{AB}^2\right)^3}{\left(1- \Tr (\rho_A^2)\right) \left(1-\Tr (\rho_B^2)\right)}. \label{eq:ent_enhanced}
\end{align}
If this condition is violated, the underlying state must be entangled.
Although it requires some additional work, the expression on both sides of this equations can be re-expressed in terms of linear observables in tensor products of $\rho_{AB}$, which makes them experimentally accessible.

 Moreover, entanglement conditions based on realignment moments have been introduced in~\cite{zhang_quantum_2022}.
 Finally, let us point out that the idea of using linear observables in tensor products -- which are also known as index permutation matrices in this context -- to detect entanglement is not new. Specifically connected to this work, optimal entanglement detection criteria have already been found in this framework~\cite{liu2022detecting}. Here, we provide additional criteria which are perhaps less powerful, but simpler to state, simpler to estimate and which follow from a simpler proof argument. Below in Fig.~\ref{fig:ent-detection}, with the batch shadow estimators, we illustrate an example of mixed state entanglement detection from the experimental data of Ref.~\cite{Brydges2019probing} using Proposition.~\ref{prop-ent} and the optimal condition in Eq.~(7) of~\cite{liu2022detecting}, where we clearly observe an enhanced detection capability of the optimal condition. 
 We additionally note that with the finite measurement statistics available from the experiment of Ref.~\cite{Brydges2019probing}, we were unable to extract experimentally, the enhanced condition derived in Eq.~\eqref{eq:ent_enhanced} and its corresponding optimal condition~\cite[Eq.~(8)]{liu2022detecting} due to large error bars on the experimental data arising from the finite available measurement statistics.
 \clearpage 
 \begin{figure}[h]
\begin{minipage}[b]{0.37\linewidth}
\centering
\includegraphics[width=\textwidth]{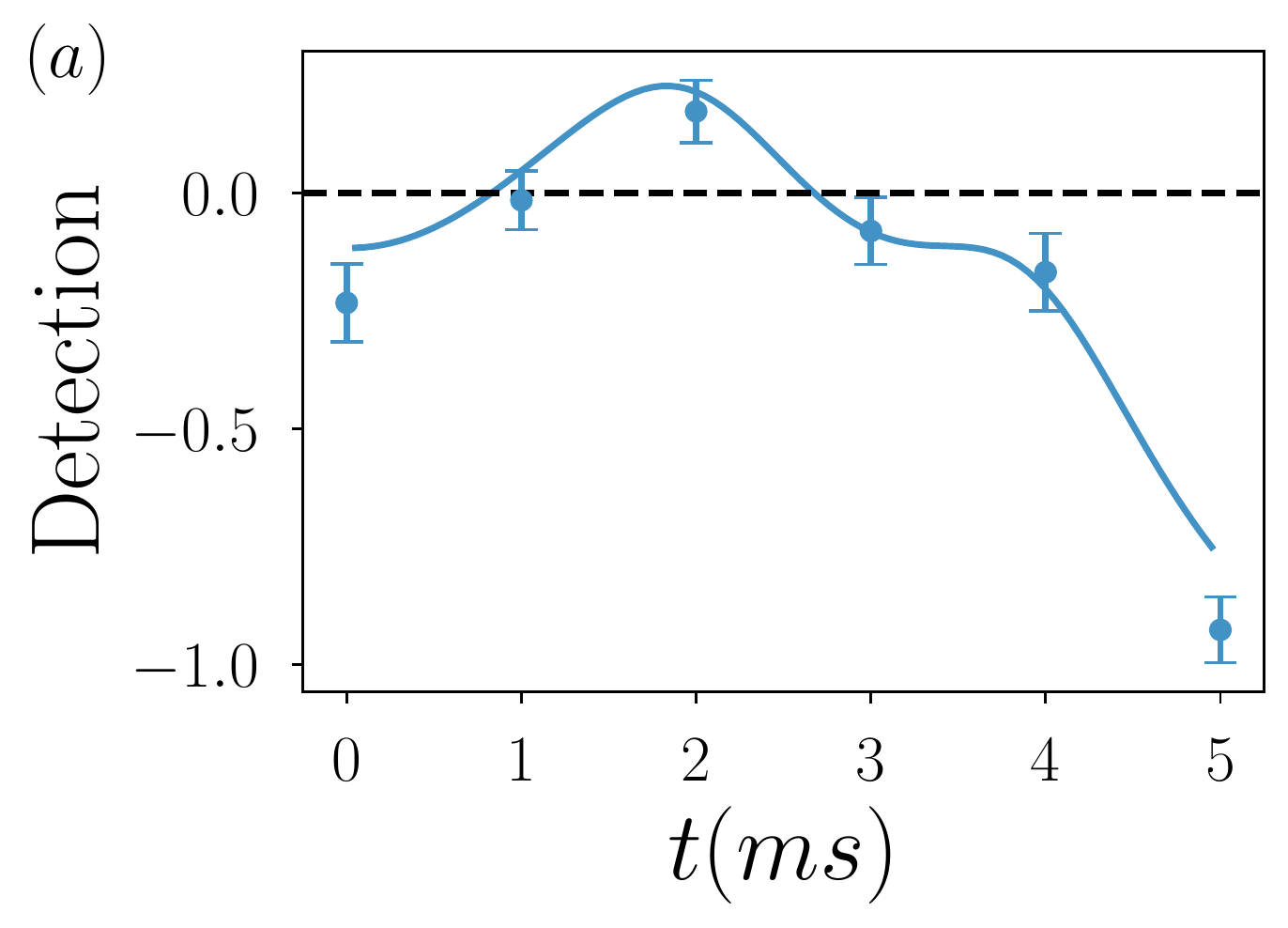}
\end{minipage}
\hskip +5ex
\begin{minipage}[b]{0.37\linewidth}
\centering
\includegraphics[width=\textwidth]{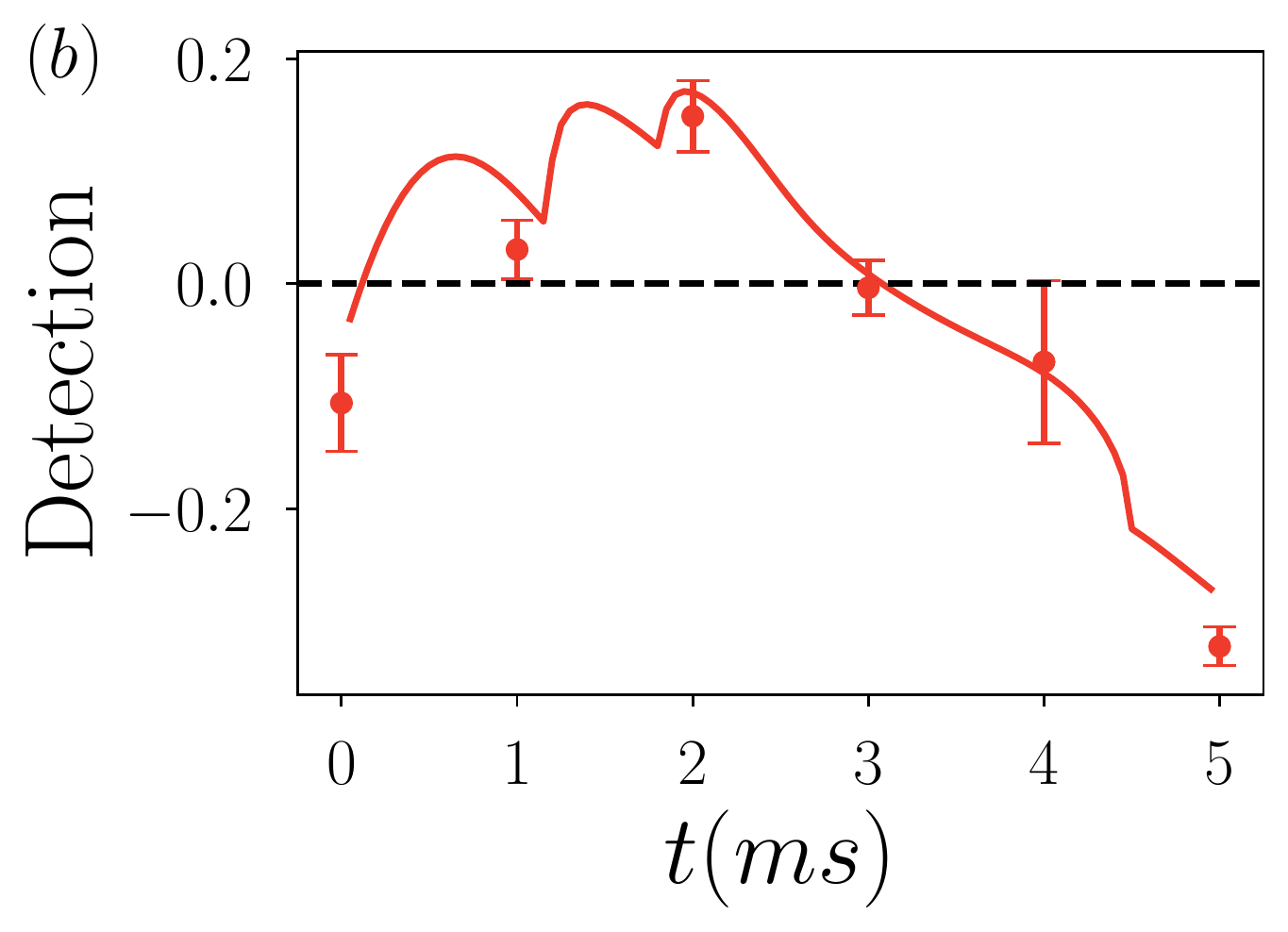}
\end{minipage}
\caption{ Entanglement detection:  We consider a reduced density matrix $\rho_{AB}$ defined on the subsystem $A = [1, \, 2]$ and $B = [3,\,4]$ for the 10-ion experiment of Ref.~\cite{Brydges2019probing}. In panel (a), we plot as detection on the vertical axis, the condition given in Proposition.~\ref{prop-ent} $(S^{(2)}(\rho_{AB}) - R^{(2)}(\rho_{AB}))$, and similarly in panel (b) we use the optimal condition in Eq.~(7) of~\cite{liu2022detecting} $(E_{2n}^\pi(\rho_{AB})-1)$. 
We detect entanglement between the partitions $A$ and $B$ for various times $t$ during the quench dynamics when we observe values greater than 0.
The points show experimental results with the error bars calculated with Jackknife resampling. The solid lines correspond to numerical simulations of the unitary dynamics including dissipation. \label{fig:ent-detection}}
\end{figure}

\section{Symmetry-resolution of operator Schmidt decomposition}\label{app:proofs}
In this appendix we provide a proof of Eqs.~\eqref{eq:rhoABgeneral}--\eqref{eq:qdefinition} of the main text. We start by showing the symmetry resolution of the Schmidt decomposition for a pure state $\ket{\psi} \in \mathcal{H} = \mathcal{H}_A \otimes \mathcal{H}_B$, which is an eigenstate of an additive charge $Q=Q_A+Q_B$ in App.~\ref{app:proof1}. This is equivalent to showing that the reduced density matrix $\rho_A = {\rm Tr}_B | \psi \rangle \langle \psi |$ admits a block diagonal structure with respect to the eigensubspaces of $Q_A$ \cite{moshe2018symmetry}, which we prove in App.~\ref{app:proof2}. In App.~\ref{app:proof3} and App.~\ref{app:proof4} we repeat the same arguments for a generic operator $O$, showing that the decomposition in Eq.~\eqref{eq:rhoABgeneral} is possible.
In App.~\ref{sec:example} we give a simple example of symmetry resolution of an operator for the density matrix of a 3-qubit system.

\subsection{Symmetry-resolved Schmidt decomposition of a pure state}\label{app:proof1}
Let $\ket{\psi} \in \mathcal{H} = \mathcal{H}_A \otimes \mathcal{H}_B$ be a state that satisfies
\begin{equation}
	\label{eq:assumption1_charge}
	(Q_A + Q_B) \left| \psi \right> = 0,
\end{equation}
for Hermitian operators $Q_A$ and $Q_B$ acting on $\mathcal{H}_A$, $\mathcal{H}_B$.

\begin{proposition}\label{app:propB1} There exists a symmetry-resolved Schmidt decomposition
\begin{equation}
	\label{eq:claim1_schmidt}
	\left| \psi \right> =  \sum_q \sum_j  \lambda_j^{(q)}   |  \psi_{A,j}^{(q)} \rangle    |  \psi_{B,j}^{(-q)} \rangle  ,
\end{equation}
with $\langle \psi_{A,j}^{(q)} |  \psi_{A,j'}^{(q')} \rangle  = \delta_{q,q'} \delta_{j,j'}$, $\langle \psi_{B,j}^{(q)} |  \psi_{B,j'}^{(q')} \rangle  = \delta_{q,q'} \delta_{j,j'}$, and 
\begin{equation}
	Q_A  | \psi_{A,j}^{(q)} \rangle  =  q  | \psi_{A,j}^{(q)} \rangle , \qquad  Q_B  | \psi_{B,j}^{(q)} \rangle  =  q  | \psi_{B,j}^{(q)} \rangle .
\end{equation}
\end{proposition}

\begin{proof} $\mathcal{H}_A$ can be decomposed into eigenspaces of $Q_A$: $\mathcal{H}_A = \bigoplus_q   \mathcal{H}^{(q)}_A$ where the $q$'s are the eigenvalues of $Q_A$. We pick an orthonormal basis for each $\mathcal{H}^{(q)}_A$: $\{  | e_{A,1}^{(q)} \rangle , | e_{A,2}^{(q)} \rangle , \dots ,   \} $. Notice that if $q \neq q'$, then $\langle e_{A,i}^{(q)} | e_{A,j}^{(q')}  \rangle = 0$. This is because $Q_A$ is Hermitian, so $q \langle e_{A,i}^{(q)} | e_{A,j}^{(q')}  \rangle  =  \langle e_{A,i}^{(q)}  | Q | e_{A,j}^{(q')}  \rangle = q' \langle e_{A,i}^{(q)} | e_{A,j}^{(q')}  \rangle$. Thus the basis vectors of different eigenspaces are orthogonal, and we can use all of them together as an orthonormal basis  that spans the total space $\mathcal{H}_A$. Similarly, one can pick a basis for $\mathcal{H}_B$ such that every basis vector is an eigenstate of $Q_B$.\\
Any state $\left|  \psi \right> \in \mathcal{H} = \mathcal{H}_A \otimes \mathcal{H}_B$ can be written in the form
\begin{equation}
	| \psi \rangle = \sum_{q} \sum_j \sum_{q'} \sum_{j'}  M^{(q,q')}_{j,j'}  |  e_{A,j}^{(q)} \rangle \otimes | e_{B,j'}^{(q')} \rangle ,
\end{equation}
with some components $M^{(q,q')}_{j,j'}$.
If $\left| \psi \right>$ satisfies Eq.~\eqref{eq:assumption1_charge}, then
\begin{equation}
	0 = (Q_A + Q_B) | \psi \rangle = \sum_{q} \sum_j \sum_{q'} \sum_{j'} (q+q')  M^{(q,q')}_{j,j'}   |  e_{A,j}^{(q)} \rangle \otimes | e_{B,j'}^{(q')} \rangle ,
\end{equation}
and linear independence of the basis vectors implies
\begin{equation}
	 M^{(q,q')}_{j,j'} = 0   \quad {\rm if} \quad q + q' \neq 0 .
\end{equation}
Thus, for a state $| \psi \rangle$ that satisfies Eq.~\eqref{eq:assumption1_charge},
\begin{equation}
	| \psi \rangle = \sum_{q} \sum_j \sum_{j'}  M^{(q,-q)}_{j,j'}  |  e_{A,j}^{(q)} \rangle \otimes | e_{B,j'}^{(-q)} \rangle .
\end{equation}
Then we can treat the different $q$-blocks separately. For every $q$, we  take a Schmidt decomposition of the state
\begin{equation}
	\sum_j \sum_{j'}  M^{(q,-q)}_{j,j'}  |  e_{A,j}^{(q)} \rangle \otimes | e_{B,j'}^{(-q)} \rangle =  \sum_j  \lambda^{(q)}_j |  \psi_{A,j}^{(q)} \rangle \otimes | \psi_{B,j}^{(-q)} \rangle ,
\end{equation}
and putting all these together we obtain the \emph{symmetry-resolved Schmidt decomposition} Eq.~(\ref{eq:claim1_schmidt}). 
\end{proof}

\subsection{Block-diagonal form of the reduced density matrix}
\label{app:proof2}

We now prove the block-diagonal structure of the reduced density matrix, which is equivalent to the symmetry-resolved Schmidt decomposition for the evaluation of the entanglement in the charge sectors \cite{moshe2018symmetry}.
As before, let $\left| \psi \right> \in \mathcal{H} = \mathcal{H}_A \otimes \mathcal{H}_B$ be a state that satisfies
\begin{equation}
	\label{eq:assumption2_charge}
	(Q_A + Q_B) \left| \psi \right> = 0,
\end{equation}
for Hermitian operators $Q_A$ and $Q_B$ acting on $\mathcal{H}_A$, $\mathcal{H}_B$.

\begin{proposition}\label{app:propB2} The reduced density matrix $\rho_A = {\rm Tr}_B | \psi \rangle \langle \psi |$ commutes with $Q_A$.
\end{proposition}

\begin{proof} We notice that \begin{eqnarray}
	\nonumber			[Q_A,  {\rm Tr}_B | \psi \rangle \langle \psi | ] &=&  {\rm Tr}_B (  [Q_A, | \psi \rangle \langle \psi | ] ) \\
	\nonumber			&=&  {\rm Tr}_B (  [Q_A+Q_B, | \psi \rangle \langle \psi | ]  -   [Q_B, | \psi \rangle \langle \psi | ]   ) \\
				&=& - {\rm Tr}_B (    Q_B  | \psi \rangle \langle \psi |   ) +  {\rm Tr}_B (     | \psi \rangle \langle \psi |  Q_B  ) ,
 		\end{eqnarray}
and this vanishes because of the cyclicity of the trace. 
\end{proof}

We can thus write an eigenvalue decomposition for the reduced density matrix with blocks labelled by the charges $q$ of $Q_A$. This leads to the same Schmidt  decomposition as in Eq.~\eqref{eq:claim1_schmidt}.

\subsection{Symmetry-resolved operator Schmidt decomposition}
\label{app:proof3}
What has been done so fare for the state $\ket{\psi}$ can be generalized to operators. In particular, let \mbox{$O  \in {\rm End} (\mathcal{H}) = {\rm End} (\mathcal{H}_A)  \otimes  {\rm End} (\mathcal{H}_B) $} be an operator that satisfies
\begin{equation}
	\label{eq:assumption3_charge}
	[ Q_A + Q_B, O ] = 0,
\end{equation}
for Hermitian operators $Q_A$ and $Q_B$.\\

\begin{proposition}\label{app:propB3}
There exists a symmetry-resolved operator Schmidt decomposition
\begin{equation}
	\label{eq:claim2_schmidt}
	\frac{O}{\sqrt{ {\mathrm {Tr}} (O^\dagger O) }} =  \sum_q \sum_j  \lambda_j^{(q)}   O_{A,j}^{(q)} \otimes  O_{B,j}^{(-q)}  ,
\end{equation}
with $ {\rm Tr}  (O_{A,j}^{(q)})^\dagger   O_{A,j'}^{(q')}   = \delta_{q,q'} \delta_{j,j'}$, $ {\rm Tr}  (O_{B,j}^{(q)})^\dagger   O_{B,j'}^{(q')}   = \delta_{q,q'} \delta_{j,j'}$, and
\begin{equation}
	[Q_A ,  O_{A,j}^{(q)} ] =  q O_{A,j}^{(q)}, \qquad [Q_B ,  O_{B,j}^{(q)} ] =  q O_{B,j}^{(q)}   .
\end{equation}
\end{proposition}

\begin{proof} This is exactly the same statement as above, with $\mathcal{H}_A$ replaced by ${\rm End} (\mathcal{H}_A)$ equipped with the Hilbert-Schmidt inner product. To be more explicit, ${\rm End}( \mathcal{H}_A)$ can be decomposed into eigenspaces of the commutator \mbox{$[Q_A, . ]$: ${\rm End}( \mathcal{H}_A) = \bigoplus_q  {\rm End}(  \mathcal{H}^{(q)}_A)$} where the $q$'s are the eigenvalues of $[Q_A, . ]$. We pick an orthonormal basis of ${\rm End}(  \mathcal{H}^{(q)}_A)$: $\{ E^{(q)}_{A,1} , E^{(q)}_{A,2} , \dots ,   \} $. As above, we can combine all these basis vectors with different $q$'s into a basis for the total Hilbert space ${\rm End}( \mathcal{H}_A)$. We do the same for part $B$.\\
The operator $O \in {\rm End} (\mathcal{H})$ can then be written in that basis, with some components $ M^{(q,q')}_{j,j'} $,
\begin{equation}
	\frac{O}{\sqrt{ {\mathrm{ Tr}} (O^\dagger O )}}  = \sum_{q} \sum_j \sum_{q'} \sum_{j'}  M^{(q,q')}_{j,j'}   E^{(q)}_{A,j  } \otimes  E^{(q')}_{B,j'} ,
\end{equation}
where, by definition,
\begin{equation}
	[ Q_A, E^{(q)}_{A,j  } ] = q  E^{(q)}_{A,j  } , \qquad  [ Q_B, E^{(q)}_{B,j  } ] = q  E^{(q)}_{B,j  } .
\end{equation}

If $O$ satisfies Eq.~\eqref{eq:assumption3_charge}, then
\begin{equation}
	 M^{(q,q')}_{j,j'} = 0   \quad {\rm if} \quad q + q' \neq 0 ,
\end{equation}
for the same reason as above. Again, the $q$-blocks can be separately Schmidt-decomposed, and this leads to Eq.~(\ref{eq:claim2_schmidt}).
\end{proof}

\subsection{Block-diagonal form of the \emph{super-reduced density matrix}}\label{app:proof4}

Let $O \in {\rm End} (\mathcal{H}) = {\rm End} (\mathcal{H}_A)  \otimes  {\rm End} (\mathcal{H}_B)$ be an operator that satisfies
\begin{equation}
	\label{eq:assumption4_charge}
	[ Q_A + Q_B, O ] = 0 ,
\end{equation}
for Hermitian operators $Q_A$ and $Q_B$.

\begin{proposition}\label{app:propB4} The super-reduced density matrix ${\mathrm Tr}_{B\otimes B}  (|  O \rangle \langle O |)$ commutes with the supercharge \mbox{$\mathcal{Q}_A = Q_A \otimes 1 - 1 \otimes Q_A^T$}.\\
\end{proposition}

\begin{proof} We notice that

\begin{eqnarray}
	\nonumber			[\mathcal{Q_A},  {\rm Tr}_{B\otimes B} \ket{O}\bra{O}] &=&  {\rm Tr}_{B\otimes B} (  [\mathcal{Q}_A, \ket{O}\bra{O} ] ) \\
	\nonumber			&=&  {\rm Tr}_{B\otimes B} (  [\mathcal{Q}_A+\mathcal{Q}_B, \ket{O}\bra{O} ]  -   [\mathcal{Q}_B, \ket{O}\bra{O} ]   ) \\
				&=& - {\rm Tr}_{B\otimes B} (    \mathcal{Q}_B  \ket{O}\bra{O}   ) +  {\rm Tr}_{B\otimes B} (     \ket{O}\bra{O}
				\mathcal{Q}_B  ) ,
		\end{eqnarray}
and this vanishes because of the cyclicity of the trace.
\end{proof}

We can thus write an eigenvalue decomposition for the super-reduced density matrix with blocks labelled by the charges $q$ of $\mathcal{Q}_A$. This leads to the same Schmidt  decomposition as in Eq.~\eqref{eq:claim2_schmidt}.

\subsection{3-qubit example}\label{sec:example}

To get familiar with the symmetry resolution of an operator, let us look at a minimal illustrative example: a 3-qubit system, whose qubits are labeled $A$, $B$ and $C$, in a state of the form (\mbox{$|\alpha|^2+|\beta|^2+ |\gamma|^2=1$})
\begin{equation}\label{eq:psi}
    \ket{\psi}_{ABC} =  \alpha \ket{100} +  \beta \ket{010} + \gamma \ket{001}.
\end{equation}
This is an eigenstate of the total charge operator \mbox{$Q_{ABC}= \sum_{j=A,B,C} Q_j$} with $Q_j=\ketbra{1}{1}$. The reduced density matrix of the subsystem $AB$ is
\begin{equation}
    \label{eq:rhoAB1}
    \rho_{AB} = \left(\alpha \ket{10} + \beta \ket{01} \right)\left(\alpha^* \bra{10} + \beta^* \bra{01} \right) + |\gamma|^2 \ket{00} \bra{00},
\end{equation}
which commutes with $Q_{A} + Q_B$. Therefore the definitions introduced in the main text can be used and it makes sense to study the SROE of the reduced density matrix $\rho_{AB}$ in this minimal example.

Let us proceed by vectorizing the reduced density matrix $\rho_{AB}$:
\begin{equation}
\begin{aligned}
    \ket{\rho_{AB}} \, = \, & |\gamma|^2 \ket{00}_A\ket{00}_B + |\beta|^2 \ket{00}_A \ket{11}_B
     + |\alpha|^2 \ket{11}_A \ket{00}_B \\
    &
   + \alpha^* \beta \ket{01}_A \ket{10}_B +\alpha \beta^* \ket{10}_A \ket{01}_B.
\end{aligned}
\end{equation}
From this, we can build the object $\ket{\rho_{AB}}\bra{\rho_{AB}}$ and take the trace over the subsystem $B$. This gives
\begin{equation}
\begin{aligned}
    \mathrm{Tr}_B \ket{\rho_{AB}}\bra{\rho_{AB}} = &|\beta|^4 \ket{00}\bra{00} +|\alpha|^2|\beta|^2 (\ket{01}\bra{01}+\ket{10}\bra{10}) \\
    &
   +(|\gamma|^2\ket{00} + |\alpha|^2\ket{11})(|\gamma|^2\bra{00} + |\alpha|^2\bra{11}).
    \end{aligned}
\end{equation}
By reshuffling the elements of the basis, we find out that the matrix has a block-diagonal decomposition as 
\begin{equation}
\label{eq:matrixtrmod}
\mathrm{Tr}_B \ket{\rho_{AB}}\bra{\rho_{AB}} =
\begin{pmatrix}
|\alpha|^2|\beta|^2
\end{pmatrix}_{ q=1}
\oplus 
\begin{pmatrix} 
|\alpha|^2|\beta|^2
\end{pmatrix}_{ q=-1} \oplus \\
\begin{pmatrix}
|\beta|^4+ |\gamma|^4 &   |\alpha|^2|\gamma|^2    \\
 |\alpha|^2|\gamma|^2  & |\alpha|^4 
\end{pmatrix}_{ q=0} ,
\end{equation}
where each block lives in the eigensubspace of the supercharge operator ${\mathcal Q}_A = Q_{A}\otimes \mathbbm{1}-\mathbbm{1}\otimes Q_{A}^T$ (with the corresponding eigenvalues indicated as subscripts).

From this we can treat each block separately, as in the proofs of Propositions~\ref{app:propB1} or~\ref{app:propB3} above. We thus obtain the
following ‘symmetry resolved operator Schmidt decomposition':
\begin{equation}   \label{eq:rhoAB3}
\frac{\rho_{AB}}{\sqrt{ {\rm Tr} [\rho_{AB}^2] }} = \lambda^{(1)}  ~O_{A}^{(1)} \otimes  O_{B}^{(-1)}  
    + \lambda^{(-1)} ~ O_{A}^{(-1)} \otimes  O_{B}^{(1)} + \sum_{j=1,2} \lambda^{(0)}_{j} O_{A,j}^{(0)} \otimes  O_{B,j}^{(0)} ,
\end{equation}
where the Schmidt coefficients $\lambda^{(q)}_{j}$ are given by
\begin{equation}
\lambda^{(1)}=\lambda^{(-1)} = \frac12\sqrt{\chi}, \quad \lambda^{(0)}_{j} = \frac{1}{2}(1\pm\sqrt{1-\chi})
\end{equation}
with $\chi = \frac{4|\alpha|^2|\beta|^2}{{\rm Tr} [\rho_{AB}^2]}$, ${\rm Tr}[\rho_{AB}^2]=(|\alpha|^2+|\beta|^2)^2+|\gamma|^4$, and the operators $O_{A,j}^{(q)}$ and $O_{B,j}^{(q)}$, which form orthonormal sets, read
\begin{equation}\label{eq:rhoAB4}
    \begin{split}
O_A^{(1)} = \ket{1} \bra{0}_A,& \quad O_B^{(-1)} = \ket{0} \bra{1}_B,\\
   O_A^{(-1)} = \ket{0} \bra{1}_A,&\quad O_B^{(1)} = \ket{1} \bra{0}_B,\\
    O_{A,1}^{(0)} = a_{+}\ket{0} \bra{0}_A+ a_{-}\ket{1} \bra{1}_A,&\quad O_{B,1}^{(0)} = b_{+}\ket{0} \bra{0}_B+ b_{-}\ket{1} \bra{1}_B ,\\ 
   O_{A,2}^{(0)}=a_{-}\ket{0} \bra{0}_A - a_{+}\ket{1} \bra{1}_A,& \quad O_{B,2}^{(0)} = -b_{-}\ket{0} \bra{0}_B+ b_{+}\ket{1} \bra{1}_B,
    \end{split}
\end{equation}
for some real parameters $a_\pm, b_\pm$ that are functions of $|\alpha|, |\beta|$ and $|\gamma|$ (which can be obtained explicitly, but are rather tedious to write).

As a more specific example, if we fix the parameters in Eq.~\eqref{eq:psi} as for instance $\alpha=\sqrt{5/12}$, $\beta=1/2$, $\gamma=1/\sqrt{3}$, then the coefficients $\lambda_j^{(q)}, a_{\pm}, b_{\pm}$ introduced above are found to be:
\begin{equation}\label{eq:eigs}
\begin{split}
&\lambda^{(1)}=\lambda^{(-1)} = \frac{\sqrt{3}}{4}, \quad \lambda^{(0)}_1 = \frac{3}{4}, \quad
\lambda^{(0)}_2 = \frac{1}{4}, \\
&a_+ = a_- = \frac{1}{\sqrt{2}},\quad
b_{+}=\frac{3}{\sqrt{10}},\quad b_{-}=\frac{1}{\sqrt{10}} .
\end{split}
\end{equation}
From Eq.~\eqref{eq:eigs}, we then get the following expressions for the OE
\begin{equation}
   S(\rho_{AB})=4\log 2-\frac{3}{2}\log 3
\end{equation}
and SROE
\begin{equation}
    S_{\pm 1}(\rho_{AB})=0,\quad S_{0}(\rho_{AB})=\log 10-\frac{9}{5}\log 3, 
\end{equation}
and we can check that Eq.~\eqref{eq:sum_rule} is indeed satisfied.

To finish with this example, let us comment on the non-hermiticity of the operators $O_{A/B}^{(\pm 1)}$ introduced in Eq.~\eqref{eq:rhoAB4} above. As mentioned in the main text, in general it is always possible to impose that the operators $O_{A/B,i}$ that enter the standard Schmidt decomposition, Eq.~\eqref{eq:schmidtO0} or~\eqref{eq:schmidtO}, of a Hermitian operator must be Hermitian themselves. However, this can no longer be ensured when imposing the symmetry-resolved form of the Schmidt decomposition, as in Eqs.~\eqref{eq:rhoABgeneral}--\eqref{eq:qdefinition}.

To see this, recall that while the Schmidt coefficients are unique, the operators that are used in the standard (non-symmetry-resolved) decomposition are not: they are only unique---up to some complementary phases (or up to a sign if one wants them to be Hermitian)---when the corresponding Schmidt coefficient has multiplicity one. E.g. in our example above, the operators $O_{A/B,j}^{(0)}$, corresponding to the different Schmidt coefficients $\lambda^{(0)}_j$, are unique (up to a phase or a sign), but there remains some freedom, in a standard Schmidt decomposition, for the choice of operators $O_{A/B}^{(\pm 1)}$ corresponding to the multiplicity-2 Schmidt coefficients $\lambda^{(1)}=\lambda^{(-1)}$.

In fact, one finds that the most general form of the non-symmetry-resolved Schmidt decomposition of $\rho_{AB}$ above would be as in Eq.~\eqref{eq:rhoAB3}, but replacing $O_{A/B}^{(\pm 1)}$ from Eq.~\eqref{eq:rhoAB4} by
\begin{equation}
    \begin{split}
O_A^{(1)} = \mu \ketbra{1}{0}_A + \nu \ketbra{0}{1}_A,& \quad O_B^{(-1)} = \nu^* \ketbra{1}{0}_B + \mu^* \ketbra{0}{1}_B,\\
   O_A^{(-1)} = e^{i\varphi} \big(-\nu^* \ketbra{1}{0}_A + \mu^* \ketbra{0}{1}_A\big),&\quad O_B^{(1)} = e^{-i\varphi} \big(\mu \ketbra{1}{0}_B - \nu \ketbra{0}{1}_B\big),
    \end{split}
\end{equation}
for some complex coefficients $\mu,\nu$ such that $|\mu|^2+|\nu|^2=1$ and some phase $\varphi$
(and, for full generality, replacing $O_{A,j}^{(0)}$ and $O_{B,j}^{(0)}$ by $e^{i\varphi_j}O_{A,j}^{(0)}$ and $e^{-i\varphi_j}O_{B,j}^{(0)}$, resp., for some phases $\varphi_j$). It is easily seen that these operators can indeed be chosen to be Hermitian, by taking e.g. $\mu=\nu=\frac{1}{\sqrt{2}}$ and $\varphi=\frac{\pi}{2}$.
 However, imposing that the Schmidt decomposition is symmetry-resolving, i.e., that the above operators satisfy Eq.~\eqref{eq:qdefinition}, turns out to be incompatible with these being Hermitian: e.g., $[Q_A , O_A^{(1)}] = O_A^{(1)}$ requires $|\mu| = 1$, $\nu = 0$ (as in Eq.~\eqref{eq:rhoAB4}).

One may note, for completeness, that Eq.~\eqref{eq:qdefinition} can actually be satisfied by Hermitian operators $O_{A,j}^{(q)}$, $O_{A,j}^{(q)}$ only in the case where $q=0$: indeed, assuming that $O_{A,j}^{(q)} = (O_{A,j}^{(q)})^\dagger$, Eq.~\eqref{eq:qdefinition} then implies (using the cyclicity of the trace) $\mathrm{Tr}\big([Q_A , O_{A,j}^{(q)}]O_{A,j}^{(q)}\big) = 0 = q\,\mathrm{Tr}\big((O_{A,j}^{(q)})^\dagger O_{A,j}^{(q)}\big) = q$.

\section{The batch shadow randomized measurement toolbox}\label{app:batchshadows}

\subsection{Classical shadows with local CUE and Pauli measurements}
\label{app:paulishadows}

Given an $N$ qubit state prepared on a quantum device, we can construct a Haar classical shadow $\hat{\rho}^{(r)}$ (equivalently called a Haar shadow)
of the state defined in Eq.~\eqref{eq:def_shadow_rm} (with $N_M = 1$) of the main text~\cite{huang2020predicting}: 
\begin{equation}
    \hat{\rho}^{(r)} = \bigotimes_{i = 1}^N \Big( 3 (u^{(r)}_i)^{\dagger} \ket{s^{(r)}_i}\bra{s^{(r)}_i} u^{(r)}_i - \mathbbm{I}_2 \Big)
\end{equation}
where the applied local random unitary is sampled from the CUE equivalently from the Haar measure (local CUE measurements).
Alternatively we could consider random single-qubit operations that equivalently lead to measuring each qubit in one of the random Pauli basis of $\mathcal{X},\, \mathcal{Y}$ or $\mathcal{Z}$ (local Pauli measurements).
These lead to six possible states that can be succinctly summarized as:
\begin{equation}
    \ket{\mathcal{B},s} \quad \text{with} \quad \mathcal{B} \in \{\mathcal{X},\mathcal{Y},\mathcal{Z}\}, \,\,\, s \in \{\pm\}
\end{equation}
More precisely, these states correspond to the following six possibilities:
\begin{equation}
    \ket{0} = \ket{\mathcal{Z},+}, \quad \ket{1} = \ket{\mathcal{Z},-}, \quad \ket{+} = \ket{\mathcal{X},+}, \quad \ket{-} = \ket{\mathcal{X},-}, \quad \ket{i+} = \ket{\mathcal{Y},+}, \quad \ket{i-} = \ket{\mathcal{Y},-}.
\end{equation}
To construct a Pauli shadow $\hat{\rho}$, we choose randomly and uniformly for each single qubit $i$, a basis $\mathcal{B}_i$ in $\mathcal{X},\, \mathcal{Y}$ or $\mathcal{Z}$ which is subsequently followed by the resulting basis measurement that provides a string of signs $\textbf{s} = (s_1, \ldots, s_N) \in \{\pm \}$. With this information and defining $N$ chosen bases $\mathbfcal{B} = (\mathcal{B}_1, \ldots, \mathcal{B}_N)$, we can provide an unbiased estimator of the density matrix $\rho$ as~\cite{huang2021provably}:
\begin{equation}
    \hat{\rho}(\mathbfcal{B}, \textbf{s}) = \bigotimes_{i = 1}^N \Big( 3\ket{\mathcal{B}_i,s_i}\bra{\mathcal{B}_i,s_i} - \mathbb{I}_2 \Big)
    \quad \text{such that} \quad \mathbb{E}[\hat{\rho}(\mathbfcal{B}, \textbf{s})] = \rho
    \label{eq:paulishadow}
\end{equation}
Here, $\mathbb{E}$ denotes the expectation value over the uniformly sampled random bases, as well as the resulting measurement outcomes.
Note that the single qubit Pauli shadows have some interesting properties due to the fact that their chosen measurement bases are mutually unbiased. 

For $\mathcal{B} \ne \mathcal{B}'$, we have
\begin{equation}
    \mathrm{Tr}\Big[ \big(3\ket{\mathcal{B},s}\bra{\mathcal{B},s} - \mathbb{I}_2\big)\big(3\ket{\mathcal{B}',s'}\bra{\mathcal{B}',s'} - \mathbb{I}_2\big)\Big] = \frac{1}{2} \quad \forall s,\,\, s' \in \{\pm\}
\end{equation}
and for $\mathcal{B} =  \mathcal{B}'$ we have
\begin{equation}
    \mathrm{Tr}\Big[ \big(3\ket{\mathcal{B},s}\bra{\mathcal{B},s} - \mathbb{I}_2\big)\big(3\ket{\mathcal{B}',s'}\bra{\mathcal{B}',s'} - \mathbb{I}_2\big)\Big] = 
    \begin{cases}
     -4 \quad &\text{if} \quad s \ne s'. \\
     5 \quad &\text{if}\quad  s = s'.
    \end{cases}
\end{equation}
This rich geometric structure allows us to deduce streamlined upper bounds on the trace overlap between different Pauli shadows.

\newtheorem{lem}{Lemma}
\begin{lem}
Given two N-qubit basis strings $\mathbfcal{B}, \, \mathbfcal{B}' \in \{ \mathcal{X}, \, \mathcal{Y},\, \mathcal{Z}\}^{\times N}$, for any sign of the outcome strings $s, \, s' \in \{\pm\}^{\times N}$ the following two statements hold: 
\begin{align}
    &&
    \mathrm{Tr} \Big(\hat{\rho}\big(\mathbfcal{B}, \textbf{s}\big)\hat{\rho}'\big(\mathbfcal{B'}, \textbf{s}'\big) \Big)^2 \leq \prod_{i = 1}^N \Bigg( 5^2 \textbf{1} \{ \mathcal{B}_i = \mathcal{B}'_i\} + \Big(\frac12\Big)^2 \textbf{1} \{ \mathcal{B}_i \ne \mathcal{B}'_i\}\Bigg) \label{eq:lemm1.1}
\end{align}
and
\begin{align}
    &&
    \mathbb{E} \Bigg[ \prod_{i = 1}^N \bigg( 5^2 \textbf{1} \{ \mathcal{B}_i = \mathcal{B}'_i\} + \Big(\frac12\Big)^2 \textbf{1} \{ \mathcal{B}_i \ne \mathcal{B}'_i\} \bigg) \Bigg] = 8.5^N, \label{eq:lemm1.2}
\end{align}
where $\textbf{1} \{ \mathcal{B}_i = \mathcal{B}'_i\}$ and $\textbf{1} \{ \mathcal{B}_i \neq \mathcal{B}'_i\}$ denote the indicator function of the advertised events.\label{lemm1}
\end{lem}

The proof strategy for this auxiliary statement is inspired by a recent analysis of classical shadows for single-qubit SIC POVMs, see \cite[Appendix~IX.B]{stricker2022SIC}.

\begin{proof}
The proof of the first inequality follows from the observation that the single qubit states $\ket{\mathcal{B}_i, s_i}$ and $\ket{\mathcal{B}'_i, s'_i}$ are mutually unbiased whenever $\mathcal{B}_i \ne \mathcal{B}'_i$. 
If two bases coincide ($\mathcal{B}_i=\mathcal{B}_i$), the squared overlap either contributes $(-4)^2$ ($s'\neq s$) or $5^2$ ($s=s'$) and can be bounded by choosing the larger term amongst them. Eq.~\eqref{eq:lemm1.1} now follows from applying this single-qubit argument to each contribution in the $N$-fold tensor product that make up the two shadows as the trace inner product of two shadows factorises into $N$ single-qubit contributions from Eq.~\eqref{eq:paulishadow}. 

Secondly, noting that all random basis choices are independent, we can develop Eq.~\eqref{eq:lemm1.2} as:
\begin{align}
        &&
        &\mathbb{E} \Bigg[ \prod_{i = 1}^N \bigg( 5^2 \textbf{1} \{ \mathcal{B}_i = \mathcal{B}'_i\} + \Big(\frac12\Big)^2 \textbf{1} \{ \mathcal{B}_i \ne \mathcal{B}'_i\} \bigg) \Bigg] =  \Bigg[  \mathbb{E} \bigg( 5^2 \textbf{1} \{ \mathcal{B}_i = \mathcal{B}'_i\} + \Big(\frac12\Big)^2 \textbf{1} \{ \mathcal{B}_i \ne \mathcal{B}'_i\} \bigg)  \Bigg]^N\\
        &&
        &= \bigg[ 5^2 \mathbb{E}[\textbf{1} \{ \mathcal{B}_i = \mathcal{B}'_i\}] + \Big(\frac12\Big)^2 \mathbb{E}[\textbf{1} \{ \mathcal{B}_i \ne \mathcal{B}'_i\}] \bigg]^N = \Bigg[ 5^2 \times \frac{1}{3} +\Big(\frac12\Big)^2 \times \frac{2}{3}\Bigg]^N  = 8.5^N,
    \end{align}
where we have used the fact that the expectation of an indicator function is the probability of the associated event. More precisely, $\mathbb{E} [\mathbf{1}\left\{ \mathcal{B}_i = \mathcal{B}_i' \right\} = \mathrm{Pr} \left[ \mathcal{B}_i = \mathcal{B}_i'\right] =1/3$, because there is a total of three basis choices to choose from. The same argument also ensures $\mathbb{E}[ \mathbf{1} \left\{\mathcal{B}_i \neq \mathcal{B}_i'\right\}] = \mathrm{Pr} \left[ \mathcal{B}_i \neq \mathcal{B}_i'\right] =1 -\mathrm{Pr} \left[ \mathcal{B}_i = \mathcal{B}_i'\right] = 1-1/3=2/3$.
\end{proof}

We now collect a number of helpful auxiliary statements that will enable us to deduce tight bounds on the estimation protocol for operator entanglement further down the road. Some statements directly follow from the properties of classical shadows and are therefore valid for both 
Pauli and CUE shadows. Other results, however, do explicitly use the structure of Pauli basis measurements and are therefore only valid for Pauli shadows. 

\begin{lem}
Given a Pauli or Haar shadow $\hat{\rho}$ that acts on $N$ qubits and $O$ be an observable on the same dimension. We have:
\begin{align}
    \mathrm{Var} \Big[\mathrm{Tr}(O\hat{\rho}) \Big] \leq \mathbb{E} \Big[  \mathrm{Tr}(O\hat{\rho})^2\Big] \leq \mathrm{Tr}(O^2) 2^N.
\end{align}
\begin{proof}
    The above statement follows from the proof of the original bound on the shadow norm of linear observables in Ref.~\cite{huang2020predicting} (proof of Proposition~3).
\end{proof}
\label{lemm2}
\end{lem}
\begin{lem}
Let $\hat{\rho}$ and $\hat{\rho}'$ be two independent Pauli shadows on $N$ qubits. Then we have
\begin{align}
    \mathrm{Var} \Big[\mathrm{Tr}(\hat{\rho}\hat{\rho}') \Big] \leq \mathbb{E} \Big[ \mathrm{Tr}(\hat{\rho}\hat{\rho}')^2\Big] \leq 8.5^N.
\end{align}
\label{lemm3}
\end{lem}
\begin{proof}
    The proofs directly follows from Lemma.~\ref{lemm1} by taking the expectation value of Eq.~\eqref{eq:lemm1.1}.
\end{proof}
\begin{lem}
Let $\rho_{AB}$ be a bipartite density matrix acting on $N = N_A + N_B$ qubits. Let  $\hat{\rho} = \hat{\rho}_A \otimes \hat{\rho}_B $ and $\hat{\rho}' = \hat{\rho}'_A \otimes \hat{\rho}'_B$ be two Pauli or Haar shadows defined on the same space that are sampled independently. 
Given two observables $O_{AB}$ and $O'_{A'B'}$ with compatible dimension, we have 
\begin{align}
    \mathrm{Var} \Big[ \mathrm{Tr}(O_{AB} \hat{\rho}_A \otimes \hat{\rho}'_B) \, \mathrm{Tr}(O'_{A'B'} \hat{\rho}'_A \otimes \hat{\rho}_B) \Big] &\leq \mathbb{E} \Big[ \mathrm{Tr}(O_{AB} \hat{\rho}_A \otimes \hat{\rho}'_B) \, \mathrm{Tr}(O'_{A'B'} \hat{\rho}'_A \otimes \hat{\rho}_B)^2 \Big] \nonumber \\
    &\leq  \mathrm{Tr}(O_{AB}^2)\,\mathrm{Tr}((O'_{A'B'})^2) \,2^{2N}    .
\end{align}
\label{lemm5}
\end{lem}
\begin{proof}
We can easily rewrite the product of traces as a larger trace over a tensor product:
\begin{equation}
    \mathrm{Tr} \Big[ (O_{AB}\otimes O'_{A'B'}) (\hat{\rho}_A \otimes \hat{\rho}'_B \otimes \hat{\rho}'_{A'} \otimes \hat{\rho}_{B'}) \Big] = \mathrm{Tr} \Big[ \mathbb{S}_{BB'}(O_{AB}\otimes O'_{A'B'}) \mathbb{S}_{BB'} (\hat{\rho}_A \otimes \hat{\rho}_B \otimes \hat{\rho}'_{A'} \otimes \hat{\rho}'_{B'}) \Big]
\end{equation}
with $\mathbb{S}_{BB'} \equiv \mathbb{I}_{AA'} \otimes \mathbb{S}_{BB'}$ that implicitly includes identity operator on the unmarked subsystem $A$ and $A'$.
Writing the shadow $\hat{\rho} \otimes \hat{\rho}' = \hat{\rho}_A \otimes \hat{\rho}_B\otimes \hat{\rho}'_{A'} \otimes \hat{\rho}'_{B'}$ in a $2^{2N}$ dimensional Hilbert space and recalling Lemma~\ref{lemm2}, we obtain
\begin{equation}
     \mathrm{Var} \Big[ \mathrm{Tr}(O_{AB} \hat{\rho}_A \otimes \hat{\rho}'_B) \, \mathrm{Tr}(O'_{A'B'} \hat{\rho}'_A \otimes \hat{\rho}_B) \Big] \leq \mathrm{Tr}\Big[(\mathbb{S}_{BB'}(O_{AB}\otimes O'_{A'B'}) \mathbb{S}_{BB'})^2\Big] 2^{2N} \leq \mathrm{Tr}(O^2_{AB})\,\mathrm{Tr}((O'_{A' B'})^2) \,2^{2N}.
\end{equation}
\end{proof}

\subsection{General treatment for batch shadow estimators} \label{App:Generaltreatement}

In this section, we introduce the batch shadow estimator -- one of the main technical contributions of this work. We also provide general statements that allow us to bound its variance when estimating trace polynomials $\mathrm{tr} \left( O^{(n)} \rho^{\otimes n}\right)$ of arbitrary order $n$.

We shall for all the subsequent sections start by performing randomized measurements to construct classical or Pauli shadows of an $N$-qubit state $\rho$. As mentioned in the main text, on each run of the protocol we sample $N$ single-qubit random unitaries from the CUE or a 2-design and apply them locally on each qubit. This is followed by a single computational basis measurement on each qubit $(N_M = 1)$. 
This procedure is repeated $M$ times (on fresh copies of the state $\rho$) and allows us to construct $M$ \emph{classical shadows} $\hat{\rho}^{(r)}$ of $\rho$, for $r = 1, \, \dots, \, M$~\cite{huang2020predicting} (here $M \equiv N_u$ as written in the main text).
We know that the expectation value of the classical shadows is $\mathbb{E}[\hat{\rho}] = \rho$~\cite{huang2020predicting}.  
We would like to estimate an $n$-order functional $X_n = \mathrm{Tr}(O^{(n)} \rho^{\otimes n})$ defined as a function of an $n$-copy operator $O^{(n)}$ using classical shadows. 

From the classical shadows we can define its U-statistics estimator $\hat{X}_n$ as
\begin{equation}
    \hat{X}_n =\frac{(M-n)!}{M!} \sum_{r_1 \ne \dots \ne r_n} \mathrm{Tr} \Big[ O^{(n)} \bigotimes_{i = 1}^n \hat{\rho}^{(r_i)}\Big] , \label{est_Ustat}
\end{equation}
where the sum ranges over all possible disjoint shadow indices $(r_1,\ldots,r_n) \in \left\{0,\ldots,M\right\}^{\times n}$ with $r_1 \neq \cdots \neq r_n$. 
The U-statistics estimator is an unbiased estimator, i.e., $\mathbb{E}[\hat{X}_n] = X_n$~\cite{Hoeffding1992,hung2021entanglement,RathFisher2021}, but evaluating it requires a summation over all possible disjoint sets of $n$ indices. While this is doable for $n=1,2$, this summation quickly becomes unfeasible as $n$ increases.
For sake of illustration, in order to estimate the U-statistics estimator of the purity ($\mathrm{Tr}(\rho^2)$) with classical shadows, the  post-processing runtime scales quadratically $\mathcal{O}(M^2)$ with the number of measurements $M$ scaling exponentially wrt the system size $N$ \mbox{$(M \propto 2^N)$}~\cite{hung2021entanglement,  Elben2020b, RathFisher2021}. 
On the other hand, a function involving $n = 4$ copies of $\rho$ immediately exposes the bottleneck of the U-statitics estimator. The number of summands to be calculated in Eq.~\eqref{est_Ustat} quickly becomes overburdening even for moderate system sizes as the runtime scales as $\mathcal{O}(M^4)$ with an overhead exponential scaling of the number of measurements $M$ and requires other alternatives.

To solve this real scaling problem problem, we propose another unbiased estimator of the same functional $X_n$ by distributing our $M$ shadows into $n' \ge n $ subsets, and first averaging the shadows in each subset. Each such defined subset is independent with respect to any other and can independently approximate $\rho$. More specifically, let us define the $b^\text{th}$ \emph{batch shadow} (denoted by a tilde rather than a hat) as
\begin{equation}
    \tilde \rho^{(b)} = \frac{n'}{M}\sum_{t_b \in T_b} \hat{\rho}^{(t_b)} \in \mathbb{C}^{2^N \times 2^N} \quad \text{where} \quad T_b = \{1+ (b-1)M/n', \dots, bM/n'\}
\end{equation}
for batches ranging from $b=1$ to $b=n'$ (for simplicity we assume $n'$ divides  $M$ such that each subset contains $M/n'$ original classical shadows).
We note, as claimed above, that $\mathbb{E}[ \tilde \rho^{(b)}] = \rho$ for every $b$. We then define the alternate unbiased estimator $\tilde X^{(n')}_n$ of $X_n$ in a similar fashion to Eq.~\eqref{est_Ustat}. But, we now symmetrize over $n'$ batch shadows:
\begin{equation}
    \tilde X^{(n')}_n = \frac{(n'-n)!}{n'!} \sum_{b_1 \ne \dots \ne b_n} \mathrm{Tr} \Big[ O^{(n)} \bigotimes_{i = 1}^n \tilde{\rho}^{(b_i)}\Big] = \frac{(n'-n)!}{n'!} \frac{n'^n}{M^n} \sum_{b_1 \ne \dots \ne b_n} \sum_{t_{b_1} \in T_{b_1}, \ldots, t_{b_n} \in T_{b_n}} \mathrm{Tr} \Big[ O^{(n)} \bigotimes_{i = 1}^n \hat{\rho}^{(t_{b_i})}\Big] . \label{est_DS}
\end{equation}
Again, by construction, $\mathbb{E}[\tilde{X}^{(n')}_n] = X_n$, i.e.\ the \emph{batch shadow estimator} is unbiased.
The principal advantage of introducing this data splitting estimator lies in the fact that, in the limit of $n' \ll M$, one can more efficiently post-process arbitrary $n$-order functionals $\tilde{X}^{(n')}_n$ compared to the basic U-statistics estimators $\hat{X}_n$. This is because all the batch shadows $\tilde{\rho}^{(b)}$ are independent and can be computed in parallel. 

By increasing $n'$ the performance of $\tilde{X}^{(n')}_n$ improves in terms of convergence as more distinct ordered pairings of $n$ different shadows $\hat{\rho}^{(r_1)}, \ldots, \hat{\rho}^{(r_n)}$ are incorporated in the batch estimator that were not considered before. In the final limit of $n' = M$, we actually recover the full U-statistics estimator $\tilde{X}^{(M)}_n = \hat{X}_n$ which has already been studied in detail~\cite{RathFisher2021}.
But the larger $n'$, the more resource intensive the classical postprocessing. Hence, we analyze the performance of this estimator in regimes where the batch size $n'$ is as small as possible, i.e.\ $n'=n$ (this is the smallest batch size that still produces an unbiased estimator for $\mathrm{Tr}(O^{(n)}\rho^{\otimes n})$).
In this case,  the sum over $b_1 \ne \dots \ne b_n$ simply boils down to computing all possible position shuffles in the tensor product of the $n$ independent batch shadows. This can be more formally written as a sum over all the permutation operator $\pi$ that acts on $n$ copies of the shadows and leads to
\begin{equation}
    \tilde{X}^{(n)}_n = \frac{1}{n!} \frac{n^n}{M^n} \sum_\pi \sum_{t_1 \in T_1, \ldots, t_n \in T_n} \mathrm{Tr} \Big[ O^{(n)} \pi \bigotimes_{i = 1}^n \hat{\rho}^{(t_i)} \pi^\dagger \Big] \label{est_DSn}
\end{equation}
with $\pi$ denoting the operator that permutes the $n$ shadows correspondingly, $\pi 
= \sum_{j_1,\ldots,j_n} \ket{j_{\pi(1)}}\!\bra{j_1}\otimes\cdots\otimes\ket{j_{\pi(n)}}\!\bra{j_n}$ (where the $\ket{j_i}$'s are orthonormal basis states).
We can gauge its performance by calculating the required number of measurements $M$ to estimate $X_n$ with an error $|\tilde{X}^{(n)}_n - X_n| \leq \epsilon$ and a certain confidence level.
Chebyshev inequality yields
\begin{equation}
    \mathrm{Pr}[|\tilde{X}^{(n)}_n - X_n| \geq \epsilon ] \leq \frac{\mathrm{Var}[\tilde{X}^{(n)}_n]}{\epsilon^2}, \label{Cby}
\end{equation}
and isolates the variance $\mathrm{Var}[\tilde{X}^{(n)}_n]$ of the batch shadow estimator as the central object to study convergence. Using the decomposition of Eq.~\eqref{est_DSn}, this variance can be rewritten as
\begin{equation}
    \mathrm{Var}[\tilde{X}^{(n)}_n] = \Big(\frac{1}{n!} \frac{n^n}{M^n}\Big)^2 \sum_{\pi,\pi'} \sum_{t_1, t_1' \in T_1, \ldots, t_n, t_n' \in T_n} \mathrm{Cov}\Bigg[\mathrm{Tr} \Big[ O^{(n)} \pi \bigotimes_{i = 1}^n \hat{\rho}^{(t_i)} \pi^\dagger \Big], \mathrm{Tr} \Big[ O^{(n)} \pi' \bigotimes_{i = 1}^n \hat{\rho}^{(t_i')} \pi'^\dagger \Big] \Bigg] .
\end{equation}
Note that all shadows that appear only once in the covariances above (i.e., those with indices $t_i\neq t_i'$) simply average to $\rho$. The shadows that appear twice (those with indices $t_i = t_i'$), on the other hand, contribute less trivially. Furthermore, because of the averaging over all permutations $\pi,\pi'$, the positions of the shadows appearing twice (i.e., the indices $i$ such that $t_i = t_i'$) does not matter. Hence, we can first sum over the number $k$ of shadows appearing twice: the $\binom{n}{k}$ corresponding terms then contribute with the same values of the covariances. We thus obtain
\begin{align}
    \mathrm{Var}[\tilde{X}^{(n)}_n] &= \Big(\frac{1}{n!} \frac{n^n}{M^n}\Big)^2 \sum_{\pi,\pi'} \sum_{k=0}^n \binom{n}{k} \sum_{\substack{t_1 \in T_1 \\[.5mm] \ldots \\[.5mm] t_k \in T_k}} \ \sum_{\substack{\tau_{k+1} \neq \tau_{k+1}' \in T_{k+1} \\[.5mm] \ldots \\[.5mm] \tau_n \neq \tau_n' \in T_n}} \mathrm{Cov}\Bigg[\mathrm{Tr} \Big[ \pi^\dagger O^{(n)} \pi [ \otimes_{i = 1}^k \hat{\rho}^{(t_i)} \otimes_{j = k+1}^n \hat{\rho}^{(\tau_j)} ] \Big], \notag \\[-10mm]
    &\hspace{90mm} \mathrm{Tr} \Big[ \pi'^\dagger O^{(n)} \pi' [ \otimes_{i = 1}^k \hat{\rho}^{(t_i)} \otimes_{j = k+1}^n \hat{\rho}^{(\tau_j')} ] \Big] \Bigg] \notag \\
    &= \Big(\frac{1}{n!} \frac{n^n}{M^n}\Big)^2 \sum_{\pi,\pi'} \sum_{k=0}^n \binom{n}{k} \Big(\frac{M}{n}\Big)^k \Big(\frac{M}{n}\big(\frac{M}{n}-1\big)\Big)^{n-k} \mathrm{Cov}\Bigg[\mathrm{Tr} \Big[ \pi^\dagger O^{(n)} \pi [ \otimes_{r = 1}^k \hat{\rho}^{(r)} \otimes\rho^{\otimes(n-k)} ] \Big], \notag \\[-3mm]
    &\hspace{90mm} \mathrm{Tr} \Big[ \pi'^\dagger O^{(n)} \pi' [ \otimes_{r = 1}^k \hat{\rho}^{(r)} \otimes\rho^{\otimes(n-k)} ] \Big] \Bigg] \notag \\
    &= \sum_{k=0}^n \binom{n}{k} \Big(\frac{n}{M}\Big)^k \Big(1-\frac{n}{M}\Big)^{n-k} \mathrm{Var}\Bigg[\frac{1}{n!}\sum_\pi\mathrm{Tr} \Big[ \pi^\dagger O^{(n)} \pi [ \otimes_{r = 1}^k \hat{\rho}^{(r)} \otimes\rho^{\otimes(n-k)} ] \Big] \Bigg] , \label{eq:VarXnn}
\end{align}
where from the first to the second lines (in addition to averaging the shadows $\hat{\rho}^{(\tau_j)}, \hat{\rho}^{(\tau_j')}$ to $\rho$, see above) we noticed that all different shadows $\hat{\rho}^{(t_i)}$ give the same statistics (hence, the same covariances), so that we could without loss of generality replace the $k$ shadows $\hat{\rho}^{(t_i)}$ by any other $k$ shadows $\hat{\rho}^{(r)}$, e.g. those for $r = 1, \ldots, k$. All $\frac{M}{n}$ terms from each of the $k$ sums over $t_i \in T_i$, and all $\frac{M}{n}\big(\frac{M}{n}-1\big)$ terms from each of the $n-k$ sums over $\tau_j \neq \tau_j' \in T_j$ then give the same values. For the last line we just rearranged all prefactors and included the sums over $\pi,\pi'$ inside the covariances, noticing that the two arguments of the covariances are then the same.

Let us note already that the variance term inside the sum of Eq.~\eqref{eq:VarXnn} cancels for $k=0$: the sum can therefore be taken to start from $k=1$. 
Defining for convenience
\begin{align}
    V_k = \mathrm{Var}\Bigg[\frac{1}{n!}\sum_\pi\mathrm{Tr} \Big[ \pi^\dagger O^{(n)} \pi [ \otimes_{r = 1}^k \hat{\rho}^{(r)} \otimes\rho^{\otimes(n-k)} ] \Big] \Bigg], \label{eq:def_Vk}
\end{align}
we obtain, 
\begin{align}
    \mathrm{Var}[\tilde{X}^{(n)}_n] &= \sum_{k=1}^n \binom{n}{k} \Big(\frac{n}{M}\Big)^k \Big(1-\frac{n}{M}\Big)^{n-k} V_k  \label{eq:VarXnn_v2} \\
    &= \sum_{\ell=1}^n \binom{n}{\ell} \Big(\frac{n}{M}\Big)^\ell \Big[ \sum_{k=1}^{\ell} \binom{\ell}{k} (-1)^{\ell-k} V_k \Big] = \frac{n^2}{M} V_1 + \frac{n^3(n-1)}{2M^2} (V_2-2V_1) +\mathcal{O}\Big(\frac{1}{M^2}\Big). \label{eq:VarXnn_order2}
\end{align}
This analysis of the special case $n'=n$ (the number of batches equals the order of the trace functional) readily extends to general batch sizes $n'$. Similar calculations produce the following generalization of Eq.~\eqref{eq:VarXnn_order2}:
\begin{align}
    &&
    \mathrm{Var}[\tilde{X}^{(n')}_n] &= \sum_{j=1}^n \binom{n}{j} \frac{\binom{n'-n}{n-j}}{\binom{n'}{n}} \sum_{k=1}^j \binom{j}{k} \big(\frac{n'}{M}\big)^k \big(1-\frac{n'}{M}\big)^{j-k} V_k 
    = \sum_{\ell=1}^n \frac{\binom{n}{\ell}^2}{\binom{n'}{\ell}} \big(\frac{n'}{M}\big)^\ell \Big[\sum_{k=1}^\ell \binom{\ell}{k} (-1)^{\ell-k} V_k \Big] 
    \\
    &&
    &= \frac{n^2}{M} V_1 + \frac{n^2(n-1)^2\frac{n'}{n'-1}}{2M^2} (V_2 - 2V_1) + \mathcal{O}\Big(\frac{1}{M^2}\Big).
\end{align}
We can provide bounds to all the above variance expressions by using the fact that the variance of an average of random variables is upper-bounded by the average of the variances. This can be seen as follows. 
For $K$ random variables $X_i$ the Cauchy–Schwarz inequality yields:
\begin{align}
    &&
    \Big( &\frac{1}{K}\sum_{i=1}^K X_i - \mathbb{E}[\frac{1}{K}\sum_{i=1}^K X_i] \Big)^2 = \langle \vec{\mathbf{1}}/K , (\vec{X} - \mathbb{E}[\vec{X}]) \rangle)^2 
    \leq \|\vec{\mathbf{1}}/K\|^2 \, \|\vec{X} - \mathbb{E}[\vec{X}]\|^2 = \frac{1}{K} \sum_{i=1}^K (X_i - \mathbb{E}[X_i])^2
\end{align}
with $\vec{X} = (X_1,\ldots,X_K)$ and $\vec{\mathbf{1}} = (1,\ldots,1)$. Taking the expectation values on both sides gives \mbox{$\mathrm{Var}[\frac{1}{K}\sum_{i=1}^K X_i] \leq \frac{1}{K} \sum_{i=1}^K \mathrm{Var}[X_i]$}. This provides us the bound:
\begin{align}
    V_k =\mathrm{Var}\Bigg[\frac{1}{n!}\sum_\pi\mathrm{Tr} \Big[ \pi^\dagger O^{(n)} \pi [ \otimes_{r = 1}^k \hat{\rho}^{(r)} \otimes\rho^{\otimes(n-k)} ] \Big] \Bigg] \leq \overline{V_k} =  \frac{1}{n!}\sum_\pi \mathrm{Var}\Bigg[\mathrm{Tr} \Big[ O^{(n)} \pi [ \otimes_{r = 1}^k \hat{\rho}^{(r)} \otimes\rho^{\otimes(n-k)} ] \pi^\dagger \Big] \Bigg] \label{eq:bnd_Vk}
\end{align}
and helps us formalize the variance bound for arbitrary batch shadow estimator.
\begin{proposition}
Let $\mathrm{Tr}(O^{(n)} \rho^{\otimes n})$ be a $n^{th}$ order trace function and let $\tilde X^{n'}_n$ with $n' \geq n$ be the associated batch shadow estimator as defined in Eq.~\eqref{est_DS}. Then, the associated variance obeys
\begin{equation}
    \mathrm{Var}[\tilde{X}^{(n')}_n] \leq \sum_{j=1}^n \binom{n}{j} \frac{\binom{n'-n}{n-j}}{\binom{n'}{n}} \sum_{k=1}^j \big(\frac{n'}{M}\big)^k \big(1-\frac{n'}{M}\big)^{j-k} \,\overline{V_k}.
\end{equation}
For $n'=n$, this bound further simplifies to
\begin{equation}
    \mathrm{Var}[\tilde{X}^{(n)}_n] \leq \sum_{k=1}^n \binom{n}{k} \Big(\frac{n}{M}\Big)^k \Big(1-\frac{n}{M}\Big)^{n-k} \, \overline{V_k}.
\end{equation}
\label{prop1}
\end{proposition}
One can further bound $\overline{V_k}$ using the formalism introduced in in~\cite{RathFisher2021}. Then, using the Chebyshev bound, one can obtain concrete sample complexity bounds to evaluate arbitrary functions $\tilde{X}_n^{n'}$ using batch shadows. 
More concretely, for comparison with the U-statistics estimator, Eq.~(D7) of~\cite{RathFisher2021}, we can re-write the U-statistics estimator as
\begin{align}
    \mathrm{Var}[\hat{X}_n] &= \sum_{k=1}^n \binom{n}{k} \frac{\binom{M-n}{n-k}}{\binom{M}{n}} V_k = \sum_{\ell=1}^n \frac{\binom{n}{\ell}^2}{\binom{M}{\ell}} \Big[\sum_{k=1}^\ell \binom{\ell}{k} (-1)^{\ell-k} V_k \Big] .
\end{align}
To first and second order in $\frac{1}{M}$, we now obtain by taking $M \gg 1$:
\begin{align}
    \mathrm{Var}[\hat{X}_n] &\simeq \frac{n^2}{M} V_1 + \frac{n^2(n-1)^2}{2M^2} (V_2 - 2V_1) +\mathcal{O}\Big(\frac{1}{M^2}\Big).
\end{align}
We note that the behavior for large $M$ depends on how $n'$ relates to $M$. Whether $n'$ is taken to be independent of $M$ (as in the case $n'=n$) or whether it is taken to just be proportional to $M$ (as in the other extreme case $n'=M$ which reproduces standard U-statistics). One finds that $\mathrm{Var}[\tilde{X}^{(n')}_n]$ and $\mathrm{Var}[\hat{X}_n]$ have the same behavior as $\frac{n^2}{M}V_1$ at first order in $\frac{1}{M}$ for any value of $n'$. 
At second order, $\mathrm{Var}[\tilde{X}^{(n')}_n]$ is only slightly (by a factor $n'/(n'-1)$) larger than $\mathrm{Var}[\hat{X}_n]$.
Hence, we really do not lose much in the precision when we use our new batch shadow technique instead of the standard U-statistics estimator of the classical shadows while we evidently achieve great improvements in runtimes of the classical treatment of the measurement data.
\medskip

With these general statements at hand, we are now in a position to employ them in order to deduce concrete and simple variance bounds for the simplest (and most relevant) batch shadow estimator $(n' = n)$ for the functions of interest in this paper as introduced in the main text: for the purity $X_2 = \mathrm{Tr}(\rho_{AB}^2) = \mathrm{Tr}(\mathbb{S}_{1,2}^{(AB)} \rho_{AB} \otimes \rho_{AB})$ and for the functional where $O^{(4)} = \mathcal{S}$: \mbox{$X_4 = \mathrm{Tr}(\mathcal{S} \rho_{AB}^{\otimes 4}) = \mathrm{Tr}\Big(\mathbb{S}_{14,}^{(A)} \otimes \mathbb{S}_{2,3}^{(A)} \otimes \mathbb{S}_{1,2}^{(B)} \otimes \mathbb{S}_{3,4}^{(B)} \rho^{\otimes 4}\Big)$}. 
We will then provide the sample complexity bounds to evaluate our quantities using properties given in Appendix~\ref{app:paulishadows}.

\subsection{Sample complexity to evaluate the purity} \label{app:puritybounds}
This section aims at providing sample complexity bounds to evaluate the purity using the batch shadow estimator formed using two batches of Pauli shadows.
The purity of a $N$-qubit quantum state $\rho_{AB}$ can be expressed as:
\begin{equation}
    X_2 = \mathrm{Tr}(\mathbb{S}_{1,2}^{(AB)} \rho_{AB} \otimes \rho_{AB}) = \mathrm{Tr}(\rho_{AB}^2),
\end{equation}
where $\mathbb{S}_{1,2}^{(AB)}$ is the Swap operator.
Given $M$ Pauli shadows, the corresponding batch shadow estimator $\tilde{X}_2^{(2)}$ (with $n' = 2$) of the purity can be written as

\begin{equation}
     \tilde X^{(2)}_2 = \frac{1}{2!} \sum_{b_1 \ne b_2} \mathrm{Tr} \Big[ \mathbb{S}_{1,2}^{(AB)} \, \bigotimes_{i = 1}^2 \tilde{\rho}^{(b_i)}\Big],
\end{equation}
where each batch shadow $\tilde \rho^{(b)}$, for $b = 1, \, 2$ writes:
\begin{equation}
    \tilde \rho^{(b)}=\frac2M\sum_{r=(b-1)M/2+1}^{bM/2}\hat \rho^{(r)}.
\end{equation}
Our goal is to bound $\mathrm{Var}[\tilde X^{(2)}_2]$ for Pauli shadows (this restriction is important, because we will need all auxiliary statements from App.~\ref{app:paulishadows}) . Using Proposition.~\ref{prop1}, this variance explicitly can be bounded as 
\begin{equation}
    \mathrm{Var}[\tilde X^{(2)}_2] \leq \frac{4}{M}\overline{V_1} + \frac{4}{M^2} (\overline{V_2} - 2\overline{V_1})  \leq \frac{4}{M}\overline{V_1} + \frac{4}{M^2} \overline{V_2}. \label{eq:VarX22}
\end{equation}
The next step consists of obtaining the bounds on the terms $\overline{V_k}$. From the expression of Eq.~\eqref{eq:bnd_Vk}, we notice that for $n= 2$, only two permutations need to be considered: $\pi = \mathbb{I}$ and $\pi = \mathbb{S}$. In each case $\pi^\dagger O^{(2)} \pi = \pi^\dagger \mathbb{S} \pi= \mathbb{S}$. Now recalling Lemma~\ref{lemm2} and Lemma~\ref{lemm3}, we can compute the bounds on $\overline{V_1}$ and $\overline{V_2}$ respectively:
\begin{align}
    \overline{V_1} & = \mathrm{Var}\Big[\mathrm{Tr} \big[ \mathbb{S}_{1,2}^{(AB)} ( \hat{\rho} \otimes\rho_{AB} ) \big] \Big] = \mathrm{Var}\Big[\mathrm{Tr} \big[ \hat{\rho} \rho_{AB} \big] \Big] \leq \mathrm{Tr}[\rho_{AB}^2]2^N \leq 2^N, \\
    \overline{V_2} & = \mathrm{Var}\Big[\mathrm{Tr} \big[ \mathbb{S}_{1,2}^{(AB)} ( \hat{\rho}^{(1)} \otimes \hat{\rho}^{(2)} ) \big] \Big] = \mathrm{Var}\Big[\mathrm{Tr} \big[ \hat{\rho}^{(1)} \hat{\rho}^{(2)} \big] \Big] \leq 8.5^N \leq 3^{2N}.
\end{align}
Then from Eq.~\eqref{eq:VarX22}
, we obtain the following bound on $\mathrm{Var}[\tilde X^{(2)}_2]$:
\begin{align}
    \mathrm{Var}[\hat{X}^{(2)}_2] &\leq \frac{4}{M} \overline{V_1} + \frac{4}{M^2} \overline{V_2} \leq \frac{4}{M} 2^N + \frac{4}{M^2} 3^{2N}. \label{eq:bndp2last}
\end{align}
Recalling the Chebyshev’s inequality mentioned in Eq.~\eqref{Cby}, we conclude
\begin{equation}
    \mathrm{Pr}[|\tilde{X}_2^{(2)} - X_2| \geq \epsilon] \leq \frac{\mathrm{Var}[\tilde{X}^{(2)}_2]}{\epsilon^2} \leq \frac{4}{\epsilon^2}\Bigg[ \frac{2^N}{M} + \frac{3^{2N}}{M^2}  \Bigg]. \label{var-p2_bnd_M}
\end{equation}
This allows us to formulate a concise sample complexity bound.
\begin{proposition}
Suppose that we wish to estimate the purity $X_2=\Tr (\rho_{AB}^2) = \Tr (\mathbb{S}_{1,2}^{(AB)}\, \rho_{AB}^{\otimes 2})$ of an $N$-qubit state $\rho$ using the batch shadow estimator $\tilde{X}^{(2)}_2$ constructed from Pauli shadows. Then for $\epsilon, \, \delta \in (0,1)$, a total of 
\begin{equation}
    M \ge 2\frac{3^N}{\epsilon\sqrt{\delta}} \bigg(\sqrt{1+a_N^2}+a_N \bigg) \quad \text{with} \quad a_N = \left( \frac{2}{3}\right)^N/\epsilon\sqrt{\delta}
\end{equation}
measurements suffices to ensure \mbox{$\mathrm{Pr}[|\tilde{X}_2^{(2)} - X_2| \geq \epsilon] \leq \delta$}. 
\end{proposition}
The scaling in terms of system size $N$ is dominated by $3^N$ which is a strict improvement over general quantum state tomography (which would require at least $4^N/\epsilon^2$ measurements).
This kind of scaling for the purity is also observed for SIC POVM measurements on independent copies where the sample complexity bound scales as $M \propto 3^N/\epsilon^2\delta$~\cite{stricker2022SIC}. 
But, when $M$ becomes sufficiently large, the scaling in Eq.~\eqref{eq:bndp2last} is dominated  by the first term $(k = 1)$ which is $ \propto 2^N/M$. This then produces a measurement complexity that scales as $M \propto 2^N/\epsilon^2 \delta$. Similar scaling behavior in this limit $M \to \infty$ have also  been observed in~\cite{Elben2020b,RathFisher2021} and reproduces an error decay rate proportional to $1/\sqrt{M}$ -- the ultimate limit for any Monte Carlo averaging procedure.

\subsection{Sample complexity of $X_4$} \label{app:boundX4}
In this section, we shall derive analytical expressions to compute the sample complexity bound on the $X_4$ functional. We shall use the tensor network graphical language to facilitate the understanding of the subsequent calculations for the reader.
Tensor network diagrams are popular in rendering heavy expressions of calculations in terms of simple graphical representations. For interested readers, we refer to the following references for a thorough introduction~\cite{biamonte2017TensorNet,Bridgeman_2017,Kueng2019}.
Fig.~\ref{fig1-app4} summarizes all the essential graphical tools that are required for our arguments.  
\begin{figure}[htbp!]
\centering
\begin{minipage}{\linewidth}
    \includegraphics[width=\linewidth]{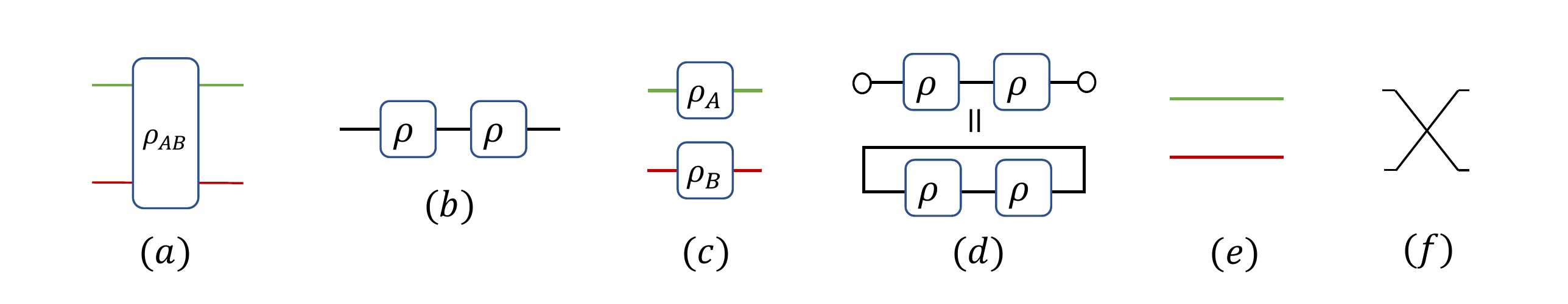} 
\end{minipage}
\caption{Important tensor network diagrams: (a) A bipartite quantum state $\rho_{AB}$ with the green (or red) legs defining the indices of subsystem A (or B) respectively. (b) By index contraction, we have multiplication of two matrices giving $\rho^2$. (c) $\rho_A \otimes \rho_B$ (d) $\mathrm{Tr}(\rho^2)$ where we have replaced the standard trace loop by circles at the end points that virtually connect to each other only horizontally at the same level. 
(e) The identity function of each sub-system $\mathbb{I}_A \otimes \mathbb{I}_B = \mathbb{I}_{AB}$ (f) The swap operator: $\mathbb{S}_{k,l}(\ket{i_k} \otimes \ket{i_l} ) = \ket{i_l} \otimes \ket{i_k}$. } \label{fig1-app4}
\end{figure}

The function $X_4$ defined in terms of the four-copy operator $O^{(4)} = \mathcal{S}$ (as introduced in the main text in Eq.~\eqref{eq:renyi2OE} and also in Ref.~\cite{liu2022detecting}) for a bipartite $(N = N_A + N_B)$-qubit state can be rewritten as
\begin{align}
    &&
    X_4 &= \mathrm{Tr}(\mathcal{S} \, \rho_{AB}^{\otimes 4}) = \mathrm{Tr}\Big(\mathbb{S}_{14}^{(A)} \otimes \mathbb{S}_{23}^{(A)} \otimes \mathbb{S}_{12}^{(B)} \otimes \mathbb{S}_{34}^{(B)} \rho_{AB}^{\otimes 4}\Big) \label{eq:X4_1}
    \\
    &&
    &= \mathrm{Tr}_{A'ABB'}\Big[\big(\rho_{AB}\,\mathbb{S}_{AA'}\,\rho_{AB}\,\mathbb{S}_{BB'}\big)\big(\rho_{AB}\,\mathbb{S}_{AA'}\,\rho_{AB}\,\mathbb{S}_{BB'}\big) \Big],\label{eq:X4_2}
\end{align}
where $\mathbb{S}^{(\Gamma)}_{cd}$ with $c,d \in [1,\ldots,4]$ and $\Gamma \in \{A,B\}$ is the swap operator acting on system $\Gamma$ on the copies $c$ and $d$ of the density matrices $\rho_{AB}$. We also assume in the above expression, an implicit reordering of the tensor products and identity operators on the unmarked subsystems 
$\rho_{AB} \equiv \rho_{AB} \otimes \mathbb{I}_{A'B'}$, $\mathbb{S}_{AA'} \equiv \mathbb{S}_{AA'} \otimes \mathbb{I}_{BB'}$ and $\mathbb{S}_{BB'} \equiv \mathbb{I}_{AA'} \otimes \mathbb{S}_{BB'}$.
Fig.~\ref{fig2-app4}, shows an equivalent expression of $X_4$ as a tensor network diagram.
We consider here the simplest batch shadow estimator $\tilde{X}^{(4)}_4$ of this function that can be evaluated from $M$ Pauli shadows as: 
\begin{equation}
    \tilde X^{(4)}_4 =  \frac{1}{4!} \sum_{b_1 \ne \dots \ne b_4} \mathrm{Tr} \Big[ \mathcal{S} \, \bigotimes_{i = 1}^4 \tilde{\rho}^{(b_i)}\Big].
\end{equation}
where each batch shadow $\tilde \rho^{(b)}$, for $b = 1, \dots , \, 4$ is an average over $M/4$ Pauli shadows given as:
\begin{equation}
    \tilde \rho^{(b)}=\frac4M\sum_{r=(b-1)M/4+1}^{bM/4}\hat \rho^{(r)}
\end{equation}
Our task is to bound the  variance $\mathrm{Var}[\tilde{X}_4^{(4)}]$. With the help of Proposition~\ref{prop1}, we can simply bound the corresponding variance as
\begin{align}
    \mathrm{Var}[\tilde{X}^{(4)}_4] &\leq  \sum_{k=1}^4 \binom{4}{k} \Big(\frac{4}{M}\Big)^k \Big(1-\frac{4}{M}\Big)^{4-k} \overline{V_k}
    \leq \sum_{k=1}^4 \binom{4}{k} \Big(\frac{4}{M}\Big)^k \overline{V_k}, \label{eq:bndX4}
\end{align}
where each of the $\overline{V_k}$ can expressed from Eq.~\eqref{eq:bnd_Vk} as 
\begin{equation}
     \overline{V_k} = \frac{1}{4!}\sum_\pi \mathrm{Var}\Bigg[\mathrm{Tr} \Big[ \mathcal{S} \pi [ \otimes_{r = 1}^k \hat{\rho}^{(r)} \otimes\rho_{AB}^{\otimes(4-k)} ] \pi^\dagger \Big] \Bigg]. \label{eq:vk_remind}
\end{equation}
Our goal now, is to calculate explicitly the bounds on the term $\overline{V_k}$ for each value of $k = 1, \ldots, 4$. 
\begin{figure}[htbp!]
\centering
\begin{minipage}{\linewidth}
    \includegraphics[width=\linewidth]{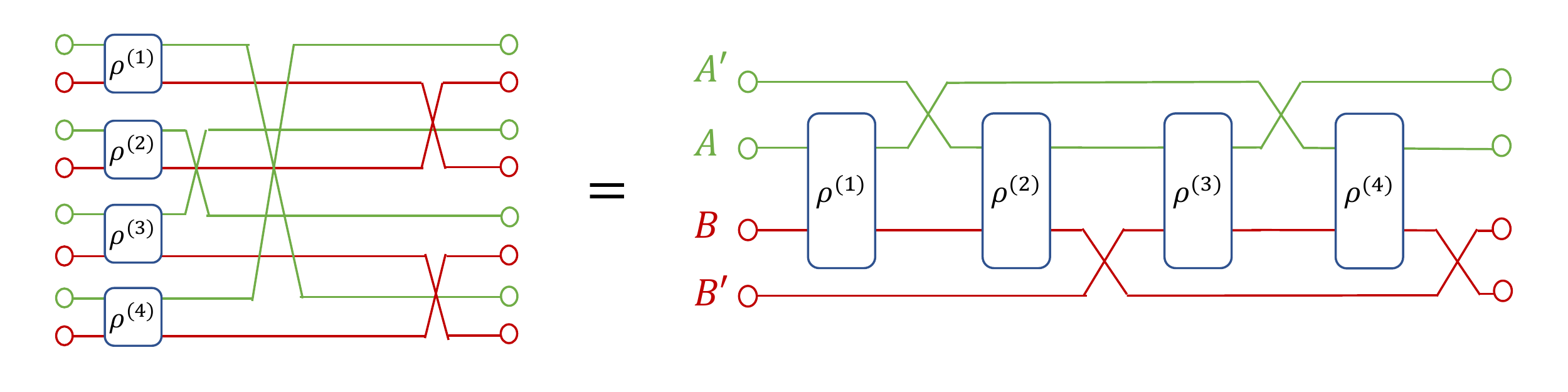}
\end{minipage}
\caption{Graphical expression of $X_4$: Expression of $X_4 = \mathrm{Tr}\Big((\mathbb{S}_{14}^{(A)} \otimes \mathbb{S}_{23}^{(A)} \otimes \mathbb{S}_{12}^{(B)} \otimes \mathbb{S}_{34}^{(B)}) (\rho^{(1)} \otimes \rho^{(2)} \otimes \rho^{(3)} \otimes \rho^{(4)} \Big) = \mathrm{Tr}\Big[\mathbb{S}_{BB'} \rho^{(4)} \mathbb{S}_{AA'} \rho^{(3)} \mathbb{S}_{BB'} \rho^{(2)} \mathbb{S}_{AA'} \rho^{(1)} \Big]$ in terms of the diagrammatic notations introduced earlier.}\label{fig2-app4}
\end{figure}
\begin{itemize}
    \item 
    \begin{figure}[htbp!]
    \centering
    \begin{minipage}{\linewidth}
        \includegraphics[width=\linewidth]{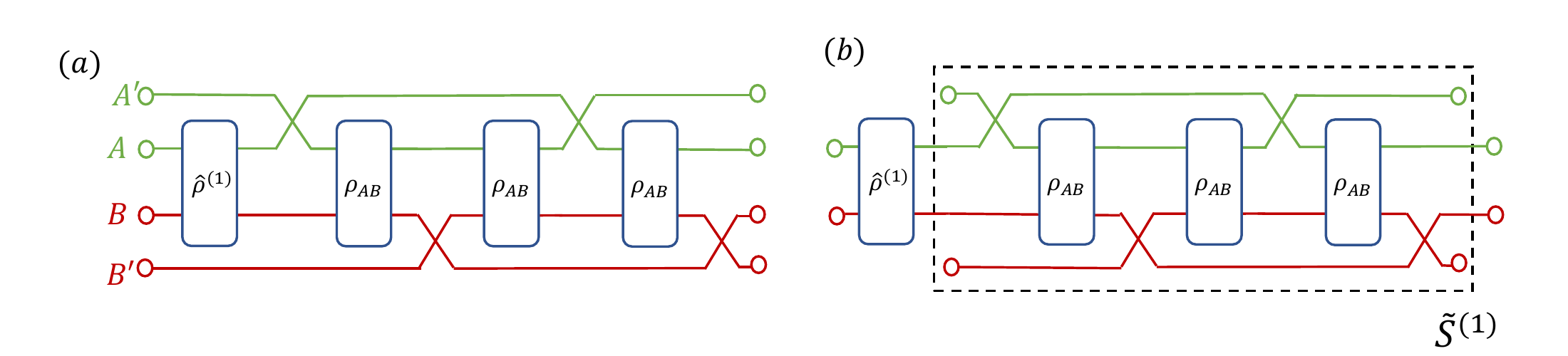}
    \end{minipage}
    \caption{Graphical representation for the case $k = 1$: (a) Diagrammatic expression of $\mathrm{Tr}\Big (\mathcal{S} \,\big [ \hat{\rho}^{(1)} \otimes \rho_{AB}^{\otimes (3)} \big]\Big)$ from which it is easily seen that the expression is invariant with respect to any position of the shadow $\hat{\rho}^{(1)}$.
    (b) The operator $\Tilde{\mathcal{S}}^{(1)}$ in Eq.~\eqref{eq:S1} (marked in the dashed rectangle) that acts on $\hat{\rho}^{(1)}$.}\label{fig3-app4}
    \end{figure}
    For $k = 1$, the trace terms in $\overline{V_1}$ contain a single Pauli shadow $\hat{\rho}^{(1)}$ and three density matrices $\rho_{AB}$.
    Regardless of $\pi$, we always obtain the same expression for the traces in Eq.~\eqref{eq:bnd_Vk} that explicitly read as
    \begin{equation}
        \mathrm{Tr} \Big[ \mathcal{S}\, [ \hat{\rho}^{(1)} \otimes\rho^{\otimes 3} ] \Big] = \mathrm{Tr} \big[ \tilde{\mathcal{S}}^{(1)} \hat{\rho}^{(1)} \big] \quad \text{with} \quad \tilde{\mathcal{S}}^{(1)} = \mathrm{Tr}_{A'B'}[\mathbb{S}_{BB'}\,\rho_{AB}\,\mathbb{S}_{AA'}\,\rho_{AB}\,\mathbb{S}_{BB'}\,\rho_{AB}\,\mathbb{S}_{AA'}], \label{eq:S1}
    \end{equation}
    see also the diagrammatic expression given in Fig.~\ref{fig3-app4} for a visual illustration.
    We notice that $\Tilde{\mathcal{S}}^{(1)}$ is a Hermitian operator.
    Using Lemma~\ref{lemm2}, we can bound $\overline{V_1}$ as
    \begin{align}
        \overline{V_1} \leq \mathrm{Var}\Bigg[\mathrm{Tr} \Big( \mathcal{S} \, [ \hat{\rho}^{(1)} \otimes\rho_{AB}^{\otimes 3} ] \Big) \Bigg] = \mathrm{Var} \Big[\mathrm{Tr} \big( \Tilde{\mathcal{S}}^{(1)} \hat{\rho}^{(1)} \big) \Big] \leq \mathrm{Tr}[(\Tilde{\mathcal{S}}^{(1)})^2] 2^N.
    \end{align}
    From Fig.~\ref{fig3-app4}(b), we can further expand the trace term $\mathrm{Tr}[(\Tilde{\mathcal{S}}^{(1)})^2]$ by performing the appropriate diagrammatic tensor contractions. We can then explicitly write it as
    \begin{equation}
        \mathrm{Tr}[(\Tilde{S}^{(1)})^2] = \mathrm{Tr}( \tilde O^{(6)} \rho_{AB}^{\otimes 6} ) \quad \text{with} \quad \tilde O^{(6)} = \mathbb{S}_{12}^{(A)} \otimes \mathbb{S}_{36}^{(A)}\otimes \mathbb{S}_{45}^{(A)} \otimes \mathbb{S}_{14}^{(B)} \otimes \mathbb{S}_{23}^{(B)} \otimes \mathbb{S}_{56}^{(B)}.
    \end{equation}
    We now use Hölder's inequality for matrices (\mbox{$|\mathrm{Tr}(AB)| \leq \|A \|_1 \|B \|_\infty$}, where $\| \cdot \|_1$ and $\| \cdot \|_\infty$ denote trace and operator norm, and also define for completeness, for a matrix $C$: $|C| = \sqrt{C^\dagger C}$ and the Schatten $p-$norm of $C$: $\|C\|_p = [\mathrm{Tr}(|C|^p)]^\frac1p$). Now we can relate this upper bound to a product of matrix norms that is easier to parse:
    \begin{equation}
        \mathrm{Tr}\big((\Tilde{\mathcal{S}}^{(1)})^2\big) \leq  \big|\mathrm{Tr} \big( \tilde O^ {(6)} \rho_{AB}^{\otimes 6}\big)\big| \leq \|\tilde O^{(6)}\|_\infty \|\rho_{AB}^{\otimes 6}\|_1 = 1\times 1 = 1, \label{prop_end}
    \end{equation}
because the 6-fold tensor product $\rho_{AB}^{\otimes 6}$ of a quantum state is again a quantum state that is normalized in trace norm ($\|\rho_{AB}^{\otimes 6} \|_1 = \|\rho_{AB} \|_1^6=\Tr (\rho_{AB})^6 =1^6=1$) and $\tilde{O}^{(6)}$ is a tensor product of swap operators and therefore unitary. Unitary operators $U$, in particular, obey $\|U \|_\infty =1$.
    Thus, we obtain the following streamlined bound on $\overline{V_1}$:
    \begin{equation}
        \overline{V_1} \leq 2^N \leq 3^N. \label{eq:bndx4k1}
    \end{equation}
    The final inequality ($2^N \leq 3^N$) is very loose, but will considerably simplify the final stage, where we put all our bounds together.
    \item
    \begin{figure}[htbp!]
    \centering
    \begin{minipage}{\linewidth}
        \includegraphics[width=\linewidth]{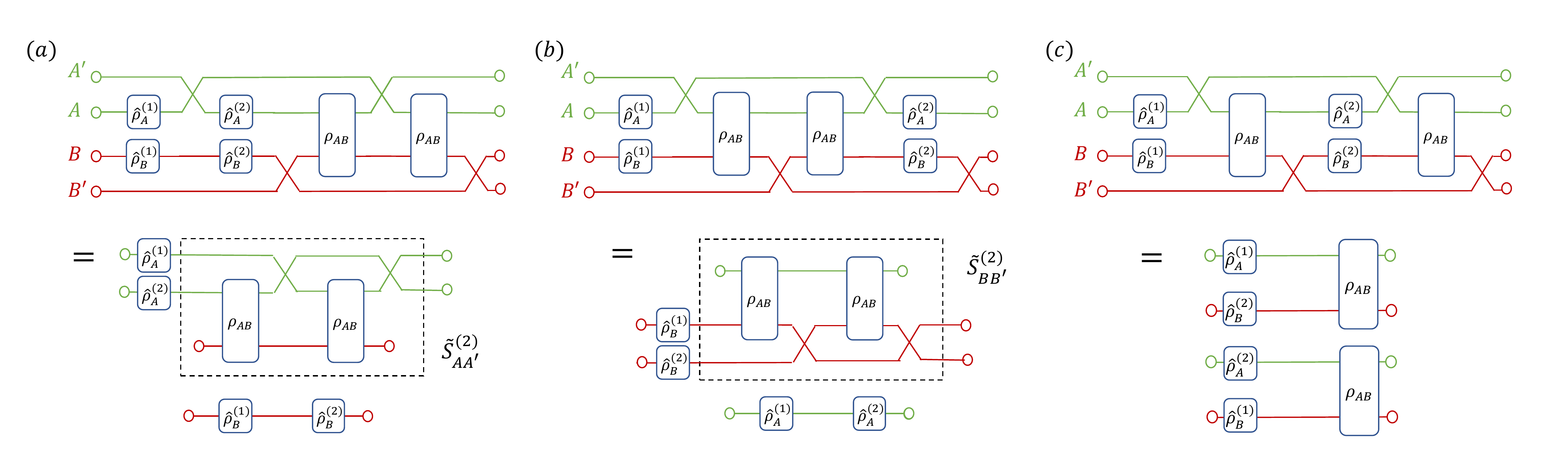}
    \end{minipage}
    \caption{Diagrammatic expression of relevant terms for $k = 2$:  (a) \mbox{$\mathrm{Tr}[\mathcal{S}\,( \hat{\rho}^{(1)}\otimes\hat{\rho}^{(2)}\otimes\rho_{AB}^{\otimes 2})] = \mathrm{Tr} \big[ \Tilde{\mathcal{S}}^{(2)}_{AA'}(\hat{\rho}^{(1)}_A\otimes\hat{\rho}^{(2)}_{A'}) \big] \, \mathrm{Tr}[\hat{\rho}^{(1)}_B\hat{\rho}^{(2)}_B] $} as in Eq.~\eqref{eq:S2AA}, (b) $\mathrm{Tr}[\mathcal{S}( \hat{\rho}^{(1)}\otimes\rho_{AB}^{\otimes 2}\otimes\hat{\rho}^{(2)})] = \mathrm{Tr} \big[ \Tilde{\mathcal{S}}^{(2)}_{BB'}(\hat{\rho}^{(1)}_B\otimes\hat{\rho}^{(2)}_{B'}) \big] \, \mathrm{Tr}[\hat{\rho}^{(1)}_A\hat{\rho}^{(2)}_A]$ as in Eq.~\eqref{eq:S2BB} and (c) \mbox{$\mathrm{Tr}[\mathcal{S} \,( \hat{\rho}^{(1)}\otimes\rho_{AB}\otimes\hat{\rho}^{(2)}\otimes\rho_{AB})] = \mathrm{Tr} \big[ (\hat{\rho}^{(1)}_A\otimes\hat{\rho}^{(2)}_B) \rho_{AB} \big] \, \mathrm{Tr} \big[(\hat{\rho}^{(2)}_A\otimes\hat{\rho}^{(1)}_B) \rho_{AB} \big]$} as in Eq.~\eqref{eq:S2AB}.} \label{fig4-app4}
    \end{figure}
    
    For $k = 2$, the variance term contains combinations of two distinct Pauli shadows $\hat{\rho}^{(1)} = \hat{\rho}^{(1)}_A \otimes \hat{\rho}^{(1)}_B $ and $\hat{\rho}^{(2)} = \hat{\rho}^{(2)}_A \otimes \hat{\rho}^{(2)}_B$, as well as two density matrices $\rho_{AB}$. 
    From all possible permutations, we need to consider three families of permutations that give different contributions to Eq.~\eqref{eq:bnd_Vk},  with each family containing the same number of permutations:
    \begin{enumerate}
        \item If $\pi \big[ \hat{\rho}^{(1)} \otimes\hat{\rho}^{(2)} \otimes\rho_{AB} \otimes\rho_{AB} \big] \pi^\dagger = \hat{\rho}^{(1)} \otimes\hat{\rho}^{(2)} \otimes\rho_{AB} \otimes\rho_{AB}$ or $\rho_{AB} \otimes\rho_{AB} \otimes\hat{\rho}^{(1)} \otimes\hat{\rho}^{(2)}$, or with the indices 1 and 2 exchanged on the right-hand-sides, then (cf Fig.~\ref{fig4-app4}(a))
        \begin{align}
        \mathrm{Tr} \Big[ \mathcal{S} \, \pi \big[ \hat{\rho}^{(1)} \otimes\hat{\rho}^{(2)} \otimes\rho_{AB} \otimes\rho_{AB} \big] \pi^\dagger \Big] = \mathrm{Tr} \big[ \Tilde{\mathcal{S}}^{(2)}_{AA'}(\hat{\rho}^{(1)}_A\otimes\hat{\rho}^{(2)}_{A'}) \big] \, \mathrm{Tr}[\hat{\rho}^{(1)}_B\hat{\rho}^{(2)}_B] 
        \label{eq:S2AA}
        \end{align}
        with $\Tilde{\mathcal{S}}^{(2)}_{AA'} = \mathrm{Tr}_B[\mathbb{S}_{AA'}\,\rho_{AB}\,\mathbb{S}_{AA'}\,\rho_{AB}].$
        The variance contribution $\overline{V_2^{(1)}}$ given by this set of permutations can be bound as
        \begin{equation}
            \overline{V_2^{(1)}} \leq \mathbb{E}\bigg[\mathrm{Tr} \big[ \Tilde{\mathcal{S}}^{(2)}_{AA'}(\hat{\rho}^{(1)}_A\otimes\hat{\rho}^{(2)}_{A'}) \big]^2 \, \mathrm{Tr}[\hat{\rho}^{(1)}_B\hat{\rho}^{(2)}_B]^2\bigg].
        \end{equation}
        To see this, we first use Lemma.~\ref{lemm1} to bound the original term as
        \begin{equation}
            \mathbb{E}\bigg[\mathrm{Tr} \big[ \Tilde{\mathcal{S}}^{(2)}_{AA'}(\hat{\rho}^{(1)}_A\otimes\hat{\rho}^{(2)}_{A'}) \big]^2 \, \mathrm{Tr}[\hat{\rho}^{(1)}_B\hat{\rho}^{(2)}_B]^2 \bigg]
            \leq \mathbb{E}\Bigg[\mathrm{Tr} \big[ \Tilde{\mathcal{S}}^{(2)}_{AA'}(\hat{\rho}^{(1)}_A\otimes\hat{\rho}^{(2)}_{A'}) \big]^2  \times \prod_{i = 1}^{N_B} \bigg( 5^2 \textbf{1} \{ \mathcal{B}_i = \mathcal{B}'_i\} + \Big(\frac12\Big)^2 \textbf{1} \{ \mathcal{B}_i \ne \mathcal{B}'_i\}\bigg) \Bigg] .
        \end{equation}
        As the Pauli basis chosen for subsystem $B$ is independent wrt the shadows in subsystem $A$, we can factorize the expectation value further:
        \begin{align}
            & \mathbb{E}\Bigg[\mathrm{Tr} \big[ \Tilde{\mathcal{S}}^{(2)}_{AA'}(\hat{\rho}^{(1)}_A\otimes\hat{\rho}^{(2)}_{A'}) \big]^2  \times \prod_{i = 1}^{N_B} \bigg( 5^2 \textbf{1} \{ \mathcal{B}_i = \mathcal{B}'_i\} + \Big(\frac12\Big)^2 \textbf{1} \{ \mathcal{B}_i \ne \mathcal{B}'_i\}\bigg) \Bigg] \notag \\
            & = \mathbb{E}\Bigg[\mathrm{Tr} \big[ \Tilde{\mathcal{S}}^{(2)}_{AA'}(\hat{\rho}^{(1)}_A\otimes\hat{\rho}^{(2)}_{A'}) \big]^2 \Bigg] \times \mathbb{E}\Bigg[ \prod_{i = 1}^{N_B} \bigg( 5^2 \textbf{1} \{ \mathcal{B}_i = \mathcal{B}'_i\} + \Big(\frac12\Big)^2 \textbf{1} \{ \mathcal{B}_i \ne \mathcal{B}'_i\}\bigg) \Bigg].
        \end{align}
        Now, the first expectation term can be bound by re-interpreting $\hat{\rho}_A^{(1)} \otimes \hat{\rho}_A^{(2)}$ as a classical shadows in a $2^{2N_A}$ dimensional Hilbert space and apply Lemma~\ref{lemm2}. The second expectation value is bounded directly by Lemma.~\ref{lemm1}. Combining both upper bounds then produces
        \begin{equation}
            \overline{V_2^{(1)}} \leq \mathrm{Tr}\big[(\Tilde{\mathcal{S}}^{(2)}_{AA'})^2\big]2^{2N_A} \, 8.5^{N_B}.
        \end{equation}
        Noting that the operator $\Tilde{\mathcal{S}}^{(2)}_{AA'}$ is Hermitian and rewriting it as $\mathrm{Tr}\big((\Tilde{\mathcal{S}}^{(2)}_{AA'})^2\big) = \mathrm{Tr}(\mathcal{S} \rho_{AB}^{\otimes 4})$, where $\mathcal{S}$ is a unitary operator, we can again apply Hölder's inequality for matrices to obtain
        \begin{equation}
            \mathrm{Tr}\big((\Tilde{\mathcal{S}}^{(2)}_{AA'})^2\big) = \mathrm{Tr}(\mathcal{S} \rho_{AB}^{\otimes 4}) \leq \|\mathcal{S}\|_\infty \|\rho_{AB}^{\otimes 4}\|_1 = 1\times 1 = 1. 
        \end{equation}
        Thus the final bound on $\overline{V_2^{(1)}}$ writes:
        \begin{equation}
            \overline{V_2^{(1)}} \leq 2^{2N_A} \, 8.5^{N_B}.
        \end{equation}
    \item If $\pi \big[ \hat{\rho}^{(1)} \otimes\hat{\rho}^{(2)} \otimes\rho_{AB} \otimes\rho_{AB} \big] \pi^\dagger = \hat{\rho}^{(1)} \otimes\rho_{AB} \otimes\rho_{AB} \otimes\hat{\rho}^{(2)}$ or $\rho_{AB} \otimes\hat{\rho}^{(1)} \otimes\hat{\rho}^{(2)} \otimes\rho_{AB}$, or with the indices 1 and 2 exchanged on the right-hand-sides, we obtain (cf Fig.~\ref{fig4-app4}(b))
        \begin{align}
        \mathrm{Tr} \Big[ \mathcal{S} \, \pi \big[ \hat{\rho}^{(1)} \otimes\hat{\rho}^{(2)} \otimes\rho_{AB} \otimes\rho_{AB} \big] \pi^\dagger \Big] = \mathrm{Tr} \big[ \Tilde{\mathcal{S}}^{(2)}_{BB'}(\hat{\rho}^{(1)}_B\otimes\hat{\rho}^{(2)}_{B'}) \big] \, \mathrm{Tr}[\hat{\rho}^{(1)}_A\hat{\rho}^{(2)}_A] 
        \label{eq:S2BB}
        \end{align}
        with $\quad \Tilde{\mathcal{S}}^{(2)}_{BB'} = \mathrm{Tr}_A[\mathbb{S}_{BB'}\,\rho_{AB}\,\mathbb{S}_{BB'}\,\rho_{AB}]$.
        This expression is similar to the previous one, with the roles of $A$ and $B$ exchanged. By following the same thread of arguments as in the previous case, we can express the final bound as:
        \begin{align}
            \overline{V_2^{(2)}} \coloneqq \mathrm{Var}\Big[ \mathrm{Tr} \Big[ \mathcal{S} \pi \big[ \hat{\rho}^{(1)} \otimes\hat{\rho}^{(2)} \otimes\rho_{AB} \otimes\rho_{AB} \big] \pi^\dagger \Big] \Big] \leq 2^{2N_B}\, 8.5^{N_A}.
        \end{align}
        \item If $\pi \big[ \hat{\rho}^{(1)} \otimes\hat{\rho}^{(2)} \otimes\rho_{AB} \otimes\rho_{AB} \big] \pi^\dagger = \hat{\rho}^{(1)} \otimes\rho_{AB} \otimes\hat{\rho}^{(2)} \otimes\rho$ or $\rho_{AB} \otimes\hat{\rho}^{(1)} \otimes\rho_{AB} \otimes\hat{\rho}^{(2)}$, or with the indices 1 and 2 exchanged on the right-hand-sides, then (cf Fig.~\ref{fig4-app4}(c))
        \begin{align}
        \mathrm{Tr} \Big[ \mathcal{S} \, \pi \big[ \hat{\rho}^{(1)} \otimes\hat{\rho}^{(2)} \otimes\rho_{AB} \otimes\rho_{AB} \big] \pi^\dagger \Big] & = \mathrm{Tr} \big[ (\hat{\rho}^{(1)}_A\otimes\hat{\rho}^{(2)}_B) \rho_{AB} \big] \, \mathrm{Tr} \big[(\hat{\rho}^{(2)}_A\otimes\hat{\rho}^{(1)}_B) \rho_{AB} \big] . \label{eq:S2AB}
        \end{align}
        Directly using Lemma~\ref{lemm5}, it then follows in this third case that
        \begin{align}
            \overline{V_2^{(3)}} \coloneqq \mathrm{Var}\Big[ \mathrm{Tr} \Big[ O^{(4)} \pi \big[ \hat{\rho}^{(1)} \otimes\hat{\rho}^{(2)} \otimes\rho_{AB} \otimes\rho_{AB} \big] \pi^\dagger \Big] \Big] \leq \mathrm{Tr}(\rho_{AB}^2)^2 \, 2^{2N} \leq 2^{2N},
        \end{align}
        because $\mathrm{Tr} \left( \rho_{AB}^2 \right)$ denotes the purity of $\rho_{AB}$ which can never exceed $1$.
    \end{enumerate}
    Combining the three cases above and incorporating these into Eq.~\eqref{eq:bnd_Vk}, we finally get
    \begin{align}
    \overline{V_2} \leq \frac{1}{3} \Big( \overline{V_2^{(1)}} + \overline{V_2^{(2)}} + \overline{V_2^{(3)}} \Big) \leq \frac{1}{3} \Big( 2^{2N_A} \, 8.5^{N_B} + 2^{2N_B} \,8.5^{N_A}  + 2^{2N} \Big) \leq 8.5^N \leq 3^{2N}.
    \end{align}
    Again, the last inequality ($8.5^N \leq 3^{2N}$) is rather loose, but will simplify putting all bounds together in the end. 
    \item
    For $k = 3$, $\overline{V_3}$ contains  3 distinct shadows $\hat{\rho}^{(1)} = \hat{\rho}^{(1)}_A \otimes \hat{\rho}^{(1)}_B $, $\hat{\rho}^{(2)} = \hat{\rho}^{(2)}_A \otimes \hat{\rho}^{(2)}_B$, $\hat{\rho}^{(3)} = \hat{\rho}^{(3)}_A \otimes \hat{\rho}^{(3)}_B$ and a single density matrix $\rho_{AB}$. We notice that regardless of $\pi$, we always obtain the same expression up to permuting the indices of the shadows for the traces in Eq.~\eqref{eq:bnd_Vk}. Thus considering the term for $\pi=\mathbb{I}$ in Eq.~\eqref{eq:bnd_Vk} is enough. This choice produces
    \begin{align}
            \mathrm{Tr} \Big[ \mathcal{S} \, \big[ \hat{\rho}^{(1)} \otimes\hat{\rho}^{(2)} \otimes\hat{\rho}^{(3)} \otimes\rho_{AB} \big] \Big] = \mathrm{Tr}[\hat{\rho}^{(2)}_A\hat{\rho}^{(3)}_A] \, \mathrm{Tr}[\hat{\rho}^{(1)}_B\hat{\rho}^{(2)}_B] \, \mathrm{Tr}[(\hat{\rho}^{(1)}_A\otimes\hat{\rho}^{(3)}_B)\rho_{AB}], \label{eq:S3}
    \end{align}
    cf Fig.~\ref{fig5-app4}(a). The corresponding variance term can be bounded as:
    \begin{align}
        \overline{V_3} \leq \mathbb{E} \bigg [  \mathrm{Tr}[\hat{\rho}^{(2)}_A\hat{\rho}^{(3)}_A]^2 \, \mathrm{Tr}[\hat{\rho}^{(1)}_B\hat{\rho}^{(2)}_B]^2 \, \mathrm{Tr}[(\hat{\rho}^{(1)}_A\otimes\hat{\rho}^{(3)}_B)\rho_{AB}]^2  \bigg].
    \end{align}
   To see this, we first use Lemma~\ref{lemm1} twice to write 
    \begin{equation}
        \overline{V_3} \leq \mathbb{E} \Bigg [   \prod_{i = 1}^{N_A} \bigg( 5^2 \textbf{1} \{ \mathcal{B}_i = \mathcal{B}'_i\} + \Big(\frac12\Big)^2 \textbf{1} \{ \mathcal{B}_i \ne \mathcal{B}'_i\}\bigg)  \prod_{i = 1}^{N_B} \bigg( 5^2 \textbf{1} \{ \mathcal{B}_i = \mathcal{B}'_i\} + \Big(\frac12\Big)^2 \textbf{1} \{ \mathcal{B}_i \ne \mathcal{B}'_i\}\bigg) \mathrm{Tr}[(\hat{\rho}^{(1)}_A\otimes\hat{\rho}^{(3)}_B)\rho_{AB}]^2  \Bigg]
    \end{equation}
    As in the previous case of $k = 2$, the measurement bases of subsystems $A$ and $B$ are chosen independent from everything else (including each other). We can use this statistical independence to factorize the remaining expecation values and bound $\overline{V_3}$ by the following expression:
    \begin{equation}
        \mathbb{E} \Bigg [   \prod_{i = 1}^{N_A} \bigg( 5^2 \textbf{1} \{ \mathcal{B}_i = \mathcal{B}'_i\} + \Big(\frac12\Big)^2 \textbf{1} \{ \mathcal{B}_i \ne \mathcal{B}'_i\}\bigg) \Bigg] \mathbb{E} \Bigg [   \prod_{j = 1}^{N_B} \bigg( 5^2 \textbf{1} \{ \mathcal{B}_j = \mathcal{B}'_j\} + \Big(\frac12\Big)^2 \textbf{1} \{ \mathcal{B}_j \ne \mathcal{B}'_j\}\bigg) \Bigg] \mathbb{E} \Bigg [\mathrm{Tr}[(\hat{\rho}^{(1)}_A\otimes\hat{\rho}^{(3)}_B)\rho_{AB}]^2  \Bigg]
    \end{equation}
    Now, with Lemma.~\ref{lemm1}, we bound the first two expectation terms by $8.5^{N_A}$ and $8.5^{N_B}$, respectively. The third term can be controlled using Lemma.~\ref{lemm2}. This finally results in
    \begin{equation}
        \overline{V_3} \leq 8.5^{N_A} \,8.5^{N_B}\, \mathrm{Tr}(\rho_{AB}^2) 2^N \leq 3^{3N},
    \end{equation}
    where we have once more used the fact that the purity obeys $\mathrm{Tr}(\rho_{AB}^2) \leq 1$. 
    \item for $k=4$, the variance term $\overline{V_4}$ is composed of four distinct shadows $\hat{\rho}^{(1)}$, $\hat{\rho}^{(2)}$, $\hat{\rho}^{(3)}$ and $\hat{\rho}^{(4)}$. Considering the term for $\pi=\mathbb{I}$ in Eq.~\eqref{eq:bnd_Vk},we obtain
    \begin{align}
    \mathrm{Tr} \Big[\mathcal{S}\, \big[ \hat{\rho}^{(1)} \otimes\hat{\rho}^{(2)} \otimes\hat{\rho}^{(3)} \otimes\hat{\rho}^{(4)} \big] \Big] = \mathrm{Tr}[\hat{\rho}^{(2)}_A\hat{\rho}^{(3)}_A] \, \mathrm{Tr}[\hat{\rho}^{(1)}_A\hat{\rho}^{(4)}_A] \, \mathrm{Tr}[\hat{\rho}^{(1)}_B\hat{\rho}^{(2)}_B] \, \mathrm{Tr}[\hat{\rho}^{(3)}_B\hat{\rho}^{(4)}_B] ,\label{eq:S4}
    \end{align}
    cf Fig.~\ref{fig5-app4}(b). 
    For other permutations $\pi$, we obtain the same kind of expressions, up to permuting the indices of the shadows (which doesn't affect the overall expectation value). We can bound this term by
    \begin{equation}
        \overline{V_4} \leq \mathbb{E} \Bigg [ \mathrm{Tr}[\hat{\rho}^{(2)}_A\hat{\rho}^{(3)}_A]^2 \, \mathrm{Tr}[\hat{\rho}^{(1)}_A\hat{\rho}^{(4)}_A]^2 \, \mathrm{Tr}[\hat{\rho}^{(1)}_B\hat{\rho}^{(2)}_B]^2 \, \mathrm{Tr}[\hat{\rho}^{(3)}_B\hat{\rho}^{(4)}_B]^2\Bigg] .
    \end{equation}
    Indeed, each trace term in the above expectation can be controlled using Lemma.~\ref{lemm1}. Noting that each measurement basis on subsystem $A$ is sampled independent from the ones in $B$ (and each other) for the four concerned shadows, we can factorize the above expectation value and obtain
    \begin{align}
        &&
        \overline{V_4} &\leq \mathbb{E} \Bigg [   \prod_{i = 1}^{N_A} \bigg( 5^2 \textbf{1} \{ \mathcal{B}_i = \mathcal{B}'_i\} + \Big(\frac12\Big)^2 \textbf{1} \{ \mathcal{B}_i \ne \mathcal{B}'_i\}\bigg) \Bigg]^2 \mathbb{E} \Bigg [   \prod_{j = 1}^{N_B} \bigg( 5^2 \textbf{1} \{ \mathcal{B}_j = \mathcal{B}'_j\} + \Big(\frac12\Big)^2 \textbf{1} \{ \mathcal{B}_j \ne \mathcal{B}'_j\}\bigg) \Bigg]^2 \\
        &&
        &\leq  8.5^{2N_A} \, 8.5^{2N_B} \leq 3^{4N}.
    \end{align}
    \begin{figure}[htbp!]
    \centering
    \begin{minipage}{0.8\linewidth}
        \includegraphics[width=\linewidth]{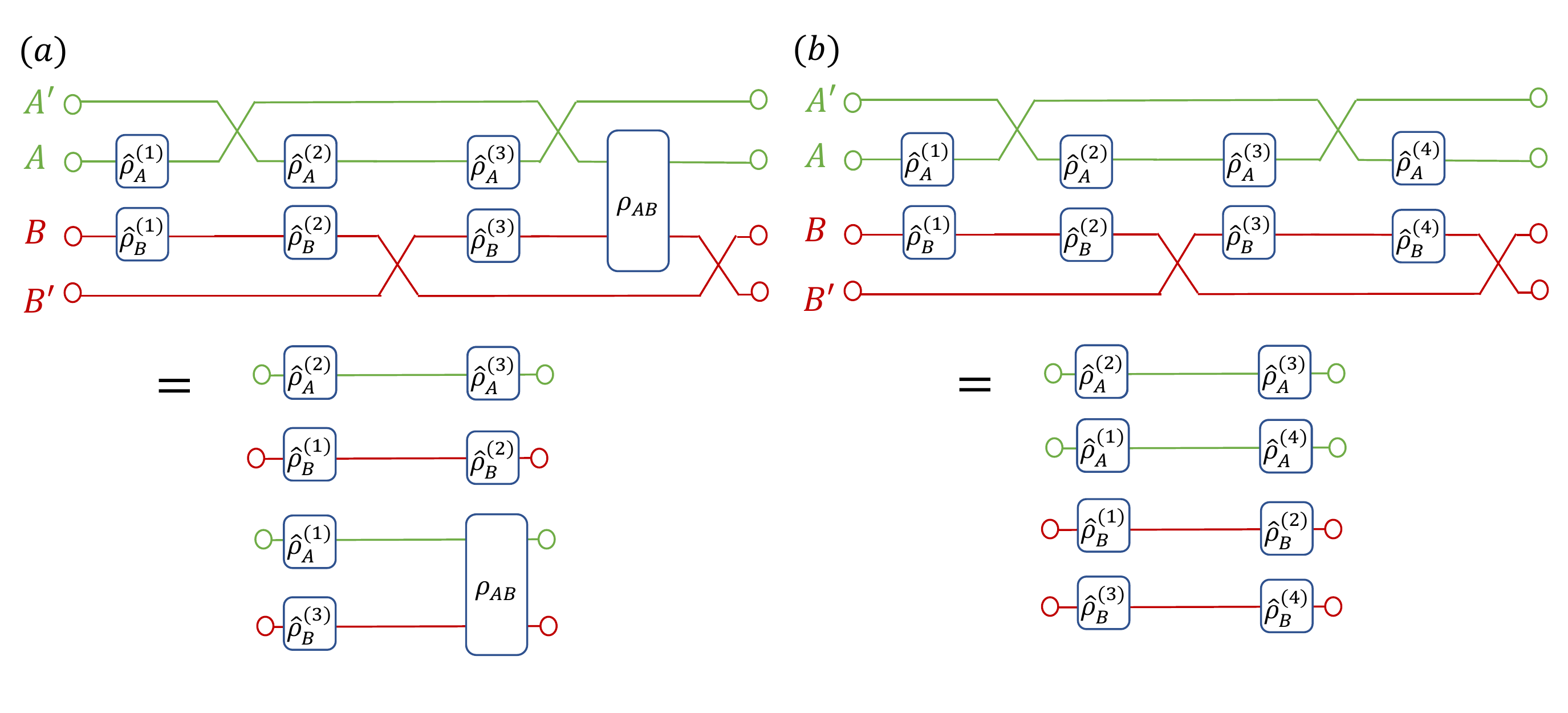}
    \end{minipage}
    \caption{ Tensor diagrams for the case $k = 3$: (a) Graphical representation for the term in Eq.~\eqref{eq:S3} (b) Representation for the term for $k = 4$ shadows as in Eq.~\eqref{eq:S4}.} \label{fig5-app4}
    \end{figure}
\end{itemize}
We now have all the pieces together to combine the results from the above case studies to get a compact expression for the variance of $\hat{X}^{(4)}_4$. More precisely, we will use the following loose bound for each $k$: $\overline{V_k} \leq 3^{kN}$.
Using Eq.~\eqref{eq:bndX4} we get:
\begin{align}
    &&
         \mathrm{Var}[\hat{X}^{(4)}_4] \leq  \sum_{k = 1}^4 \binom{4}{k} \frac{4^k}{M^{k}} \overline{V_k}
        \leq \sum_{k = 1}^4\binom{4}{k} 4^k \frac{3^{kN}}{M^k} = \bigg( 1+ 4\frac{3^N}{M} \bigg)^4 - 1. \label{eq:bndx4var}
\end{align}
The Chebyshev's inequality in Eq.~\eqref{Cby} helps us provide a sample complexity for this estimator.
\begin{proposition} \label{prop:X4-scaling}
Let $\rho_{AB}$ be a bipartite quantum state on $N=N_A+N_B$ qubits and suppose that we wish to estimate the non-linear function \mbox{$X_4 = \mathrm{Tr}(\mathcal{S} \rho_{AB}^{\otimes 4})$}, with $\mathcal{S} = \mathbb{S}_{14}^{(A)} \otimes \mathbb{S}_{23}^{(A)} \otimes \mathbb{S}_{12}^{(B)} \otimes \mathbb{S}_{34}^{(B)}$, using the batch shadow estimator $\tilde{X}^{(4)}_4$ constructed from Pauli shadows. Then for $\epsilon, \, \delta > 0$, a total of
\begin{equation}
     M \ge 4\frac{3^N}{(1+\epsilon^2\delta)^{\frac{1}{4}} - 1} \gtrsim 16\frac{3^N}{\epsilon^2\delta},
\end{equation}
measurements suffices to ensure \mbox{$\mathrm{Pr}[|\tilde{X}_4^{(4)} - X_4| \geq \epsilon] \leq \delta$}. 
\label{prop3}
\end{proposition}
This measurement cost scales (at worst) as $3^N$ in system size $N$ and provides
a scaling of $M \propto 3^N/\sqrt{\epsilon}$ for any given value of $\epsilon$. 
Though the above measurement bound is loose, this scaling still offers  an exponential improvement over the best known scaling $\mathcal{O}(4^N/\sqrt{\epsilon})$ for the U-statistics estimate of $X_4$ in the case of Pauli shadow tomography in~\cite{liu2022detecting}. 
These improvements on the complexity bounds were achieved by exploiting the rich structure of Pauli basis measurements to produce powerful auxiliary statements, most notably Lemma.~\ref{lemm1} and Lemma.~\ref{lemm3}.
At the present stage, these auxiliary results are only valid for Pauli shadows and do not yet cover Haar shadows.
We leave an extension of these arguments, and by extension Proposition.~\ref{prop:X4-scaling}, as an interesting topic for future work.

It is interesting to point out that the measurement complexity bound $X_4$ is always comparable to the measurement complexity bound for $X_2$ (purity). Moving from a a second-order function to a fourth order function does not seem to incur a large penalty in measurement complexity.

We equally note that, in the limit of $M \to \infty$, the dominant contribution to the variance is given by the linear term $(k = 1)$ which scales $\propto 2^N/M$ as given by Eq.~\eqref{eq:bndx4k1}. Then, in this limit, it holds that the measurement bound scales as $2^N/\epsilon^2 \delta$.

\subsection{Numerical investigations}\label{app:numericalsims}

In this section, we would like to consecrate ourselves to support our analytical finding with numerical simulation of the protocol. We mainly would like to study error scalings and the performance of the batch shadow estimator $\tilde{X}^{n'}_4$ by using random Pauli and Haar random shadows in the regime where $M \gg n'$ and compare it to the standard U-statistics estimator $\hat{X}_n$. We consider a $4$-qubit GHZ state and  numerically simulate the protocol by applying $M$ Haar random (CUE) unitaries $u$ followed by fixed basis measurements to construct Haar random shadows (fixing $N_M = 1$). We equally construct numerically $M$ Pauli shadows by choosing $N$ random Pauli basis for each shadow. We calculate the average statistical error $\mathcal{E} = \overline{|\tilde{X}^{n'}_4 - X_4|}/X_4$ for different values of $n'$ and $M$ by simulating the randomized measurement protocol 200 times. 
This is plotted in Fig.~\ref{fig:scalingcliffandhaar} for Pauli and Haar shadows respectively. 
We make two important observations:
\begin{itemize}
    \item The error scaling behaviours of Pauli shadows that involves sampling from a fixed set of three measurement settings, is not very different compared to that of the Haar shadows, which uses infinitely many measurement settings.  
    \item The batch shadow estimator $\tilde{X}_4^{(n')}$ with $n' \sim 10$ has very close performance as that of the U-statistics estimator. This in general translates into huge runtime gain in terms of data treatment ($\mathcal{O}(10^4)$ compared to $\mathcal{O}(M^4)$) and allows us to process the quantities of interest for a larger set of measurement data.
    We clearly observed a limitation in post-processing the U-statistics estimator ($n'= M)$ for a modest system size of $N = 4$ qubits.
    This constraint starts to be extremely prominent when the system size $N$ increases. This is due to the fact that $M$ scales exponentially with $N$ as shown in the previous section.
\end{itemize}

\begin{figure}[h]
\begin{minipage}[b]{0.37\linewidth}
\centering
\includegraphics[width=\textwidth]{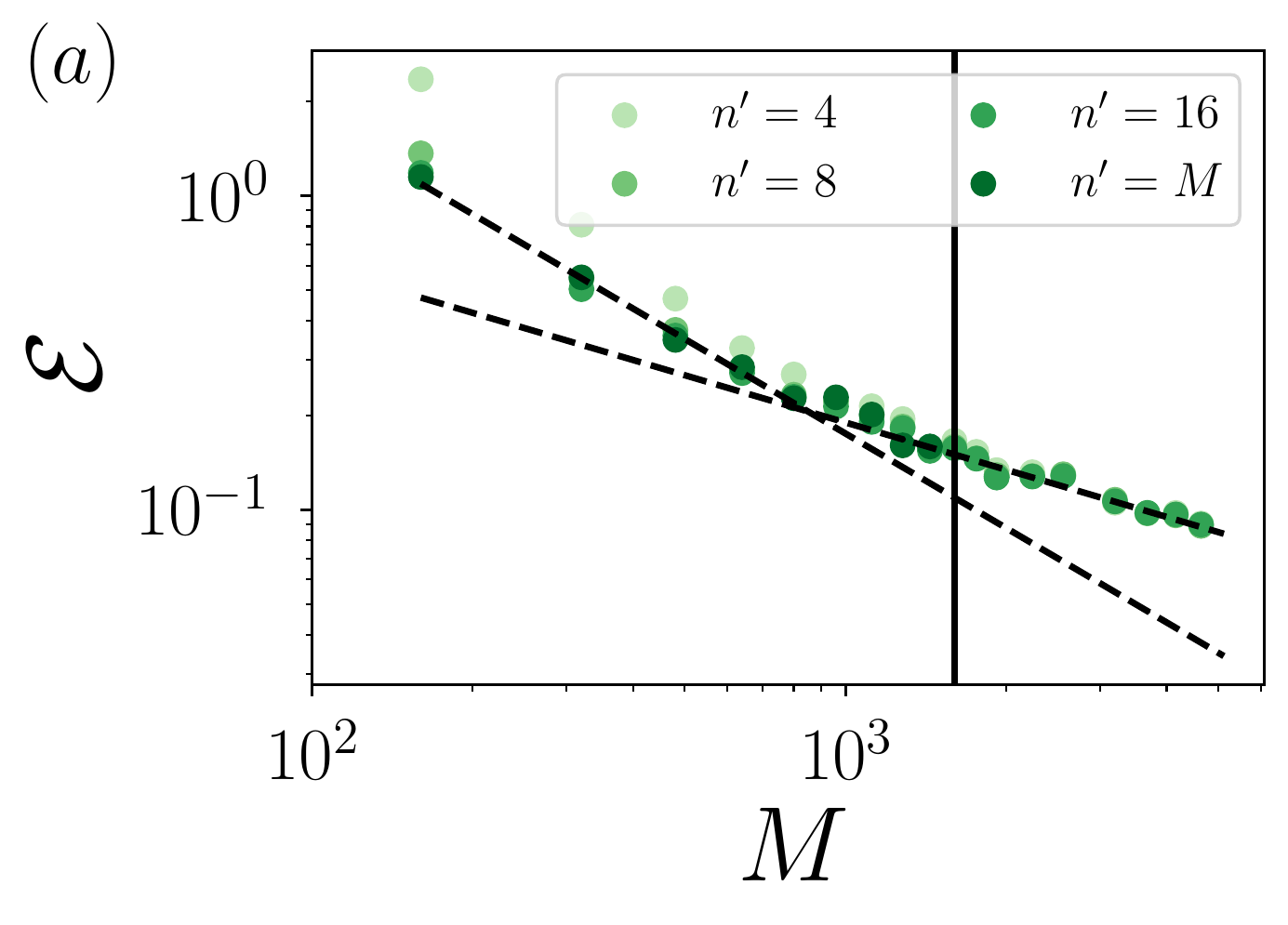}
\end{minipage}
\hskip +5ex
\begin{minipage}[b]{0.37\linewidth}
\centering
\includegraphics[width=\textwidth]{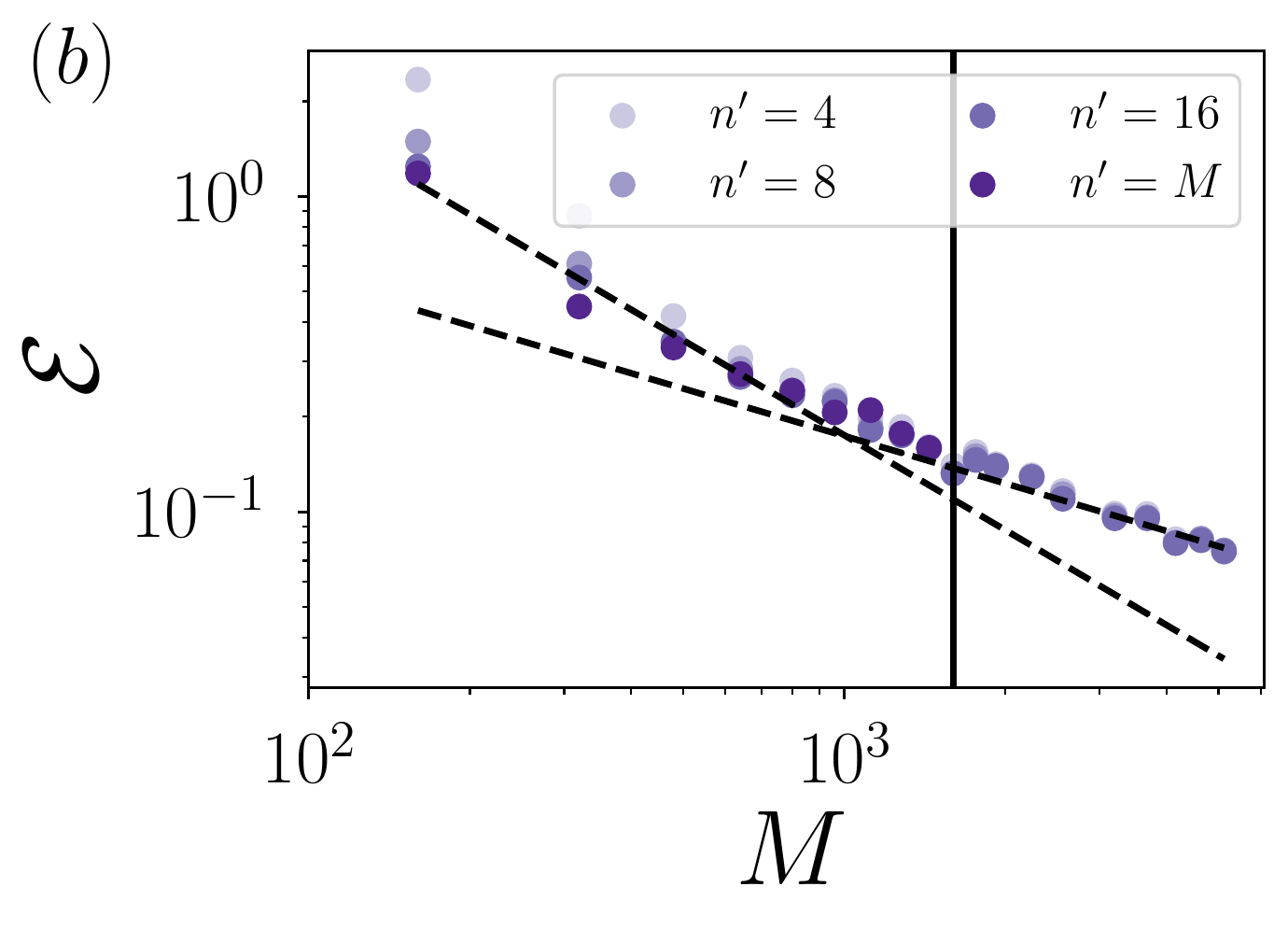}
\end{minipage}
\caption{ Error scaling as a function of $M$:  Panels (a) for Pauli shadows and (b) for haar shadows show the scaling of the average statistical error $\mathcal{E}$ as a function of the number of measurements $M$ for the functional $\tilde X^{(n')}_4$ calculated on a 4-qubit GHZ state for different values of $n'$. The black line marks the value of $M$ until which we could simulate $\tilde X^{(M)}_4$. The dashed black lines highlight the different error scalings $\propto 1/M$ and $1/\sqrt{M}$.  \label{fig:scalingcliffandhaar}}
\end{figure}

\section{Experimental platform and the theoretical modelling}\label{app:detailsexperiment}
The experimental platform in~\cite{Brydges2019probing} is realized with trapped $^{40}$Ca$^{+}$ atoms, each one encoding a single qubit. Coupling all ions off-resonantly with a laser beam subjects the ions to realize long-range Ising model in presence of a transverse field, whose effective Hamiltonian writes:
\begin{equation}
    \label{eq:Ising_ham}
    H=\hbar \sum_{i<j}J_{ij}\sigma^{x}_i\sigma^{x}_j + \hbar B\sum_i\sigma_i,
\end{equation}
with $i,j=1,\dots,N$ and $N$ is the total system size.
To model the experiment using numerical simulations, we approximate the interaction matrix $J_{ij}$ as a power-law $J_{ij}=J_{0}/|i-j|^{\alpha}$, where the values of $J_{0}$ and $\alpha$ depend on the specifics of each experimental realization and will be discussed later. The effective magnetic field $B$ is considered much larger than the interaction term ($B\simeq 22 J_0$) such that terms that would break the conservation of the total magnetization, i.e. $\sigma^{+}_i\sigma^{+}_j + h.c.$ are energetically suppressed.
The effects of decoherence on the system are taken into account considering the time evolution subject to local spin-flips and spin excitation loss. The full system  dynamics is described according to a Gorini-Kossakowski-Sudarshan-Lindblad (GKSL) master equation whose $2N$ local jump operators are written as $C_i=\sqrt{\gamma_x}\sigma^x_i$ (spin flip), $C_{i+L}=\sqrt{\gamma_-}\sigma^-_i$ (excitation loss), $i=1,\dots,L$, with rates $\gamma_x$, $\gamma_m$. 
Furthermore, the experimentally prepared state is not pure. As such, it can be written as the following mixed product state $\rho_0 = \bigotimes_i \left( 	p_i \ket{\uparrow}\bra{\uparrow} + (1-p_i) \ket{\downarrow}\bra{\downarrow} \right)$
with $p_i\approx 0.004$ for $i$ even and $p_i\approx 0.995$ for $i$ odd.\\
In the experiment, local depolarizing noise is acting during the application of the local random unitary. We model it as  
\begin{equation}
\rho(\bar{t}) \rightarrow   (1-{p_{DP}N}) \rho(\bar{t})+ {p_{DP}} \sum_{i} \mathrm{Tr}_i[\rho(\bar{t})]\otimes \frac{ \mathbb{I}_i}{2}
\end{equation}
with $p_{DP}\approx0.02$ and $\bar{t}$ denoting the time at which the measurement is performed.\\

In the case of the 20-ions experiment, the numerical simulations are done using tensor network algorithms. For the unitary part of the dynamics we approximate the interaction matrix $J_{ij}$ as a sum of 3 exponentially decaying terms which can efficiently be represented as MPOs.
To treat the decoherence we use quantum trajectories \cite{Daley2014}, applying the quantum jumps $C_i$, to the state approximated as an MPS with bond dimension $128$. The latter is evolved according to the Time-Dependent Variational Principle (TDVP) \cite{Haegeman2011}. We average our results on 1500 trajectories in total.

\section{Batch shadows to extract R\'enyi 2-OE and its symmetry resolution}\label{app:sroe_details}
We used the batch shadow estimator to access the R\'enyi 2-OE and its symmetry resolution from experimental data. The estimator of R\'enyi 2-OE $\tilde{S}^{(2)}$ constructed using $n'$ batches explicitly writes following Eq.~\eqref{eq:renyi2OE} of the main text, as
\begin{equation}
    \tilde{S}^{(2)} = -\log \frac{\tilde{X}_4^{(n')}}{\big(\tilde{X}_2^{(n')} \big)^2} =  -\log \frac{\frac{1}{4!} \binom{n'}{4}^{-1} \sum_{b_1\neq \dots \neq b_4}  \mathrm{Tr} \Big [ \mathcal{S} \bigotimes_{i = 1}^4 \tilde \rho^{(b_i)}    \Big]}{\Big(\frac{1}{2!} \binom{n'}{2}^{-1} \sum_{b_1 \neq b_2}  \mathrm{Tr} \Big [ \mathbb{S}_{1,2}^{(AB)} \bigotimes_{i = 1}^2 \tilde \rho^{(b_i)}    \Big]\Big)^2}
\end{equation}
To estimate the R\'enyi 2-OE from the experimental data as shown in the main text, we used the simple estimator with $n' = 4$. 
Alternately, the symmetry resolution for the R\'enyi 2-OE can be expressed as
\begin{equation} \label{eq:Sq2}
    S^{(2)}_q =  -\log \frac{\mathrm{Tr} \Big ( \big[ \Pi_{q} \mathrm{Tr}_{B} (\ket{\rho_{AB}} \bra{\rho_{AB}}) \Pi_q \big]^2 \Big )}{p(q)^2 \, \mathrm{Tr}(\rho_{AB}^2)^2} 
    = -\log \frac{\mathrm{Tr} \Big ( \big[ \Pi_{q} \mathrm{Tr}_{B} (\ket{\rho_{AB}} \bra{\rho_{AB}}) \Pi_q \big]^2 \Big )}{\mathrm{Tr}\big(\Pi_{q} \mathrm{Tr}_B(\ket{\rho_{AB}} \bra{\rho_{AB}})\big)^2},
\end{equation}
where $\Pi_q$ is the projector onto the eigenspace of the symmetry sector $q$ for system $A$ and \mbox{$p(q) = \mathrm{Tr}\big(\Pi_{q} \mathrm{Tr}_B(\ket{\rho_{AB}} \bra{\rho_{AB}})\big) / \mathrm{Tr}(\rho_{AB}^2)$} are the probabilities of being in the charge sector $q$  expressed in terms of a fraction of 2 second-order functions.
As $\mathbb{E}[\tilde \rho^{(b_i)}] = \rho_{AB}$ for all batch shadows, we can obtain a batch estimator of the symmetry-resolved R\'enyi 2-OE by replacing each vectorized density matrix by a distinct batch shadow. Firstly, we can express the estimator of the populations $\tilde{p}(q)$ as
\begin{equation}
    \tilde{p}(q) = \frac{\frac{1}{2!} \binom{n'}{2}^{-1} \sum_{b_1 \neq b_2} \mathrm{Tr}\big(\Pi_{q} \mathrm{Tr}_B(\ket{\tilde \rho^{(b_1)}} \bra{\tilde \rho^{(b_2)}}) \big)}{\tilde{X}_2^{(n')}}
\end{equation}
We can now explicitly write the estimator $\tilde{S}^{(2)}_q$ of SR R\'enyi 2-OE as
\begin{equation}
    \tilde{S}^{(2)}_q = -\log \frac{\frac{1}{4!} \binom{n'}{4}^{-1} \sum_{b_1\neq \dots \neq b_4}\mathrm{Tr} \Big ( \Pi_{q} \mathrm{Tr}_{B} \big(\ket{\tilde \rho^{(b_1)}} \bra{\tilde \rho^{(b_2)}} \big)\Pi_q \mathrm{Tr}_{B} \big(\ket{\tilde \rho^{(b_3)}} \bra{\tilde \rho^{(b_4)}} \big) \Big)}{\Big( \frac{1}{2!} \binom{n'}{2}^{-1} \sum_{b_1 \neq b_2} \mathrm{Tr}\big(\Pi_{q} \mathrm{Tr}_B(\ket{\tilde \rho^{(b_1)}}\bra{\tilde \rho^{(b_2)}})\big)\Big)^2} \label{S2q-est}
\end{equation}
The SR R\'enyi 2-OE were extracted from the experimental data by taking $n' = 16$.

\twocolumngrid
\bibliography{biblio.bib}

\end{document}